\documentclass[journal,twoside,twocolumn]{IEEEtran}

\usepackage[cmex10]{amsmath}   
\usepackage{amsfonts}
\usepackage{amssymb}           
\usepackage{amsthm}            

\usepackage{array}

\usepackage[pdftex]{graphicx}   
\usepackage{pgfplots}           
\pgfplotsset{compat=1.18}

\usepackage[hidelinks]{hyperref}             
\usepackage[all]{hypcap}                     
\usepackage[capitalise,noabbrev]{cleveref}   
\usepackage{autonum}                         
\hypersetup{colorlinks=true, linkcolor=blue, citecolor=blue, urlcolor=blue}
\crefformat{equation}{(#2#1#3)}              
\usepackage{tikz}                                         
\usepackage{tikz-cd}                                      
\usetikzlibrary{decorations.pathreplacing, calligraphy}   

\usepackage{url}

\usepackage{breakurl}

\usepackage{bbm}              
\usepackage{cite}             
\usepackage{flushend}         
\usepackage[utf8]{inputenc}   
\usepackage{makecell}         
\usepackage{mathtools}        
\usepackage{mleftright}       
\usepackage{stmaryrd}         
\usepackage{tabularray}       
\usepackage[normalem]{ulem}   

\newtheorem{theorem}{Theorem}
\newtheorem{corollary}{Corollary}
\newtheorem{lemma}{Lemma}
\newtheorem{proposition}{Proposition}

\allowdisplaybreaks
\newcommand{\abs}[1]{\left| #1 \right|}
\newcommand{\ang}[1]{\left\langle #1 \right\rangle}
\newcommand{\bkt}[1]{\left[ #1 \right]}
\newcommand{\bktz}[1]{\left\llbracket #1 \right\rrbracket}
\newcommand{\brc}[1]{\left\{ #1 \right\}}
\newcommand{\floor}[1]{\left\lfloor #1 \right\rfloor}
\newcommand{\frnorm}[1]{\left\lVert #1 \right\rVert_\mathrm{F}}
\newcommand{\nop}[1]{\left. #1 \right.}
\newcommand{\norm}[1]{\left\lVert #1 \right\rVert}
\newcommand{\prn}[1]{\left( #1 \right)}

\newcommand{\mbkt}[1]{\mleft[ #1 \mright]}
\newcommand{\mbrc}[1]{\mleft\{ #1 \mright\}}
\newcommand{\mprn}[1]{\mleft( #1 \mright)}

\newcommand{\iabs}[1]{\bigl| #1 \bigr|}
\newcommand{\ibkt}[1]{\bigl[ #1 \bigr]}
\newcommand{\ibrc}[1]{\bigl\{ #1 \bigr\}}
\newcommand{\iprn}[1]{\bigl( #1 \bigr)}

\DeclareMathOperator*{\cl}{cl}
\DeclareMathOperator*{\diag}{diag}
\DeclareMathOperator*{\poly}{poly}
\DeclareMathOperator*{\spn}{span}

\newcommand{\Bernoulli}[1]{\mathsf{Bernoulli} \mprn{#1}}
\newcommand{\bigO}[1]{O \mprn{#1}}
\newcommand{\bigTheta}[1]{\Theta \mprn{#1}}
\newcommand{\Categorical}[1]{\mathsf{Categorical} \mprn{#1}}
\newcommand{\circled}[1]{   
    \tikz[baseline=(char.base)]{\node[shape=circle, draw, inner sep=0.5pt] (char) {#1};}
}
\newcommand{\E}[1]{\mathbb{E} \mbkt{#1}}
\newcommand{\equad}{\mathrel{\phantom{=}}}
\newcommand{\given}{\:\middle\vert\:}
\renewcommand{\H}[1]{H \mprn{#1}}
\newcommand{\I}[1]{I \mprn{#1}}
\newcommand{\Iv}[1]{\mathbbm{1} \mbrc{#1}}
\newcommand{\lfrac}[2]{\nop{#1 \middle/ #2}}
\newcommand{\littleO}[1]{o \mprn{#1}}
\renewcommand{\mod}[1]{\, \prn{\mathrm{mod} \ #1}}   
\newcommand{\multibinom}[2]{\prn{\kern-.3em \prn{\genfrac{}{}{0pt}{}{#1}{#2}} \kern-.3em}}   
\renewcommand{\P}[1]{\mathbb{P} \mprn{#1}}
\newcommand{\rank}[1]{\mathrm{rank} \mprn{#1}}
\renewcommand{\Re}{\mathrm{Re}}
\newcommand{\smashwave}[1]{\smash{\uwave{#1}}}   
\newcommand{\stackclap}[2]{\stackrel{\mathclap{#1}}{#2}}   
\renewcommand{\top}{\mathrm{T}}
\newcommand{\Uniform}[2]{\mathsf{Uniform} \mprn{#1, #2}}
\newcommand{\zero}{\phantom{0}}

\newcommand{\Cregp}{\mathcal{C}_\mathsf{perm}}
\newcommand{\Csum}{\mathrm{C}_\mathsf{psum}}
\newcommand{\Perror}[1]{\mathrm{P}_\mathsf{error}^{#1}}

\newcommand{\M}{\mathcal{M}}

\newcommand{\R}{\mathbb{R}}

\newcommand{\Sg}{\mathrm{S}}

\begin{document}

\bstctlcite{IEEEexample:BSTcontrol}   

\title{Permutation Capacity Region of \\ Adder Multiple-Access Channels}
\author{William~Lu~and~Anuran~Makur
    \thanks{The author ordering is alphabetical. An earlier version of this work was presented in part at the IEEE International Symposium on Information Theory (ISIT) 2023 \cite{LuMakurISIT2023}.}
    \thanks{W. Lu is with the Department of Computer Science, Purdue University, West Lafayette, IN 47907, USA (e-mail: lu909@purdue.edu).}
    \thanks{A. Makur is with the Department of Computer Science and the Elmore Family School of Electrical and Computer Engineering, Purdue University, West Lafayette, IN 47907, USA (e-mail: amakur@purdue.edu).}
}

\maketitle

\begin{abstract}
    Point-to-point permutation channels are useful models of communication networks and biological storage mechanisms and have received theoretical attention in recent years. Propelled by relevant advances in this area, we analyze the \emph{permutation adder multiple-access channel} (PAMAC) in this work. In the PAMAC network model, $d$ senders communicate with a single receiver by transmitting $p$-ary codewords through an adder multiple-access channel whose output is subsequently shuffled by a random permutation block. We define a suitable notion of \emph{permutation capacity region} $\Cregp$ for this model, and establish that $\Cregp$ is the simplex consisting of all rate $d$-tuples that sum to $\lfrac{d(p - 1)}{2}$ or less. We achieve this sum-rate by encoding messages as i.i.d. samples from categorical distributions with carefully chosen parameters, and we derive an inner bound on $\Cregp$ by extending the concept of time sharing to the permutation channel setting. Our proof notably illuminates various connections between mixed-radix numerical systems and coding schemes for multiple-access channels. Furthermore, we derive an alternative inner bound on $\Cregp$ for the binary PAMAC by analyzing the root stability of the probability generating function of the adder's output distribution. Using eigenvalue perturbation results, we obtain error bounds on the spectrum of the probability generating function's companion matrix, providing quantitative estimates of decoding performance. Finally, we obtain a converse bound on $\Cregp$ matching our achievability result.
\end{abstract}


\section{Introduction} \label{section:introduction}

The \emph{noisy permutation channel} model introduced in \cite{Makur2018} is a natural abstraction of point-to-point communication through networks in which packets arrive out-of-order. It consists of a discrete memoryless channel (DMC) followed by a random permutation block that permutes the output codeword of the DMC. Several recent advances have been made to understand this model, including the original capacity bounds in \cite{Makur2020a, Makur2020b}, the subsequent covering-number-based bounds in \cite{TangPolyanskiy2023}, and the coding schemes for related models in \cite{KovacevicTan2018b, ShomoronyHeckel2021, ShomoronyVahid2021, TamirMerhav2021}.

While the capacity of point-to-point permutation channels has been established, no prior work has generalized these results to a network setting. Motivated by intrinsic theoretical interest in the mathematical formulations underpinning permutation channels, and auxiliary applications in communication networks to boot, we initiate an information-theoretic study of \emph{permutation networks} by analyzing the permutation capacity region of the $p$-ary PAMAC in this paper. Before presenting our formal model, we briefly outline some applications of permutation channels and multiple-access channels (MACs) in coding theory, communication networks, and molecular data storage systems.

\subsection{Motivation and Related Literature}

A classical model in coding theory is the \emph{random deletion channel}, wherein symbols or packets are successfully transferred from the sender to the receiver with some probability and silently dropped otherwise. This contrasts with the \emph{erasure channel}, in which the receiver is notified of dropped symbols. Constructing error-correcting codes that achieve capacity of the random deletion channel is a notable problem in coding theory. As discussed in \cite{DiggaviGrossglauser2001, Metzner2009}, attaching sequence numbers to packets reduces the problem to the well-understood task of coding for the erasure channel. On the other hand, \emph{low density parity check codes} allow the decoder to implement verifications that are robust to shuffling of the packets \cite{Mitzenmacher2006}. These results may be interpreted as preliminary steps towards analyzing the capacity of an erasure channel followed by a random permutation block.

Analogous results exist in the communication networks literature. For example, \emph{Reed-Solomon codes} \cite{XuZhang2002} are comparable to using sequence numbers to convert a random deletion channel into a random erasure channel. Packet impairment errors can be corrected by using a code designed for a permutation channel, even if sequence numbers are not included \cite{GadouleauGoupil2010}. These results may be interpreted as constructing codes for specific types of noisy permutation channels.

In general, noisy permutation channels are suitable models of \emph{multipath routed networks}, wherein each packet takes one of several possible routes from the sender to the receiver. The random permutation block models the effect of differing route latencies causing the packets to arrive at the receiver out-of-order. Mobile networks whose topologies change over time, and load-balanced networks where packets are often re-routed, are concrete examples of multipath routed networks. Past works have analyzed rate-delay tradeoffs for multipath routed networks without accounting for packet impairments such as insertions, substitutions, erasures, and deletions \cite{WalshWeberMaina2008, WalshWeberMaina2009}. More recent work takes packet impairments into account and uses message encodings that are invariant under packet permutation \cite{KovacevicVukobratovic2013}. In this vein, \emph{multiset codes} analyzed in \cite{KovacevicVukobratovic2015} encode messages as samples from some probability distribution and decode by analyzing the empirical distribution of output samples. Other multiset codes based on \emph{Sidon sets} are investigated in \cite{KovacevicTan2018b}.

Another motivation for studying permutation channels arises from their relevance to DNA storage systems, which are attractive mediums for archival storage due to their high density and reliability over long periods of time \cite{YazdiKiahGarciaruizMaZhaoMilenkovic2015, ErlichZielinski2016}. Implementations of such systems store data in relatively short DNA molecules as strings of a few hundred nucleotides \cite{HeckelShomoronyRamchandranTse2017}. Each DNA molecule can be interpreted as a codeword whose alphabet is the nucleobases $\brc{A, C, G, T}$. The receiver uses shotgun sequencing to randomly sample short fragments of codewords from the DNA pool in an unordered fashion. This setting assumes that the DNA molecules are not corrupted and are read by the receiver without noise. Alternatively, in the \emph{noisy shuffling channel} \cite{ShomoronyHeckel2019}, a DMC models the potential for DNA molecules to be corrupted during synthesis or storage. We refer readers to \cite{YazdiKiahGarciaruizMaZhaoMilenkovic2015} for an overview of DNA storage systems, and \cite{KiahPuleoMilenkovic2016, KovacevicTan2018a} for further examples of coding for DNA storage. Recent work \cite{ShomoronyHeckel2021} introduces the \emph{noisy shuffling-sampling channel}, which reflects practical constraints in the current implementations of DNA storage systems. The authors characterize the capacity of this channel using a coding scheme based on simple indexing.

We remark that the permutation channel model is similar to two other models from the recent information theory literature. In the \emph{torn-paper coding} setting \cite{ShomoronyVahid2021}, message codewords are split into chunks at random indices, and the chunks are permuted as they pass through the torn-paper channel. The receiver must recover the original message from these shuffled chunks. This model lacks a DMC since even the noise-free version of this problem is non-trivial. In the \emph{bee-identification problem} \cite{TamirMerhav2021}, the sender and receiver have access to a ground-truth codebook containing a list of codewords. The codebook is randomly shuffled by a random permutation block and each codeword is passed through a DMC. Here, the receiver's goal is to recover the permutation sampled by the block.

In a different vein, classical analyses of adder MACs include results pertaining to their capacity regions, sum-capacities, and optimal coding schemes \cite{ChangWeldon1979,ElGamalKim2011}. For example, \cite{KasamiLin1976} investigated multiple-access binary erasure channels in both the noiseless and noisy settings, and recent work \cite{GyorfiLaczay2004} analyzed \emph{random-access channels}, wherein only a subset of users are active and codes must identify the active users on top of recovering the sent messages. In general, adder MACs are suitable models for wireless networks where senders transmit messages using orthogonal signals, such as in the \emph{frequency shift keying} scheme, and the signal energy detected by the receiver is the sum of each user's signal energy \cite{Chevillat1981}. Prior works have incorporated adder MACs into networks of point-to-point channels to model hybrid wired/wireless networks \cite{NazerGastpar2006}. Adder MACs also find application in modeling communication satellites, which are multiple-access by nature and use \emph{frequency-division multiplexing} for channelization \cite{Weldon1978}.

As noted earlier, notwithstanding the relevance of permutation networks to the communications and molecular storage domains, our primary objective in this paper is to commence a rigorous theoretical study of permutation networks to elucidate new insights on coding for such models. Previous work in information theory has studied specializations of MACs (including orthogonal MACs, multiplier MACs, and adder MACs) due to their simplicity being conducive to theoretical analysis \cite{PolyanskiyWu2017Notes}. Following this precedent, we adopt the adder MAC with addition over $\mathbb{Z}$ (cf. \cite{Gu2018}) in our initial investigation of permutation networks. Moreover, our PAMAC model can be interpreted as an abstraction of a system in which multiple senders communicate (usually wirelessly) with a base station that is in turn connected to a receiver through a network. The adder captures the effect of signal addition (in the analog domain) before the signal is sent through the network, and the ensuing noisy permutation channel captures the effects of noise and the network. As a concrete example, cellular mobile networks are divided into geographic regions, each of which is serviced by a \emph{base transceiver station} connected to a public telephone network \cite[Figure 7.18]{KuroseRoss2016}. In fact, recent work proposes a \emph{non-orthogonal multiple-access} concept for cellular mobile communications, wherein signals from multiple users are superposed in the power domain \cite{SaitoKishiyamaBenjebbourNakamuraLiHiguchi2013}. Our work thus presents a simple model that captures some high-level characteristics of such systems from the networking literature and is amenable to rigorous information-theoretic analysis.

\subsection{Notation}

Let $[n] = \mathbb{Z} \cap [1, n]$, $\bktz{n} = \mathbb{Z} \cap [0, n]$, and $\bktz{a, b} = \mathbb{Z} \cap [a, b]$ denote integer intervals. Let $\mathbb{N}$ denote the natural numbers starting from $1$. Let $\mathbb{B} = \brc{0, 1}$ be the Boolean domain. Let $\mathbb{R}_+ = [0, \infty)$ denote the non-negative real numbers. Let $\mathcal{S}_{n-1}$ denote the probability simplex in $\mathbb{R}^n$. Let $\Sg_n$ denote the symmetric group of all permutations on $n$ elements. Let $\multibinom{n}{k} = \binom{n + k - 1}{k}$ denote multiset coefficients \cite[Section 1.2]{Stanley2011} representing the number of $k$-multisubsets of $[n]$.\footnote{A $k$-multisubset of a set $\mathcal{T}$ is a multiset with cardinality $k$, counting multiplicities, where each element is a member of $\mathcal{T}$.} Given a set $A \subseteq \mathbb{R}^n$, let $\cl A$ denote its closure. Let $\Iv{\cdot}$ denote the Iverson bracket. Given a non-negative integer $x \in \mathbb{N} \cup \brc{0}$ and a natural number $y \in \mathbb{N}$, let $x \mod{y}$ be the remainder of $x$ divided by $y$. We assume empty sums are $0$ and empty products are $1$, namely for all $a > b$,
\begin{align}
    \sum_{i=a}^b f(i) = 0 \enspace \text{and} \enspace \prod_{i=a}^b f(i) = 1 \, .
\end{align}

In the context of Landau notation, let $\poly(n)$ be an unspecified polynomial function of $n$. For functions $f, g: \mathbb{N} \rightarrow \mathbb{R}$, we write $f(n) \lesssim g(n)$ or $g(n) \gtrsim f(n)$ to mean that $f$ is asymptotically dominated by $g$, namely
\begin{align}
    \exists n_0 \in \mathbb{N}, \, \forall n \geq n_0, \, f(n) \leq g(n) \, .
\end{align}

Bold letters denote column vectors or matrices unless otherwise stated. Given a vector $\mathbf{x}$, let $\bkt{\mathbf{x}}_i$ denote its $i$th entry. For any $p \in \bkt{1, \infty}$, let $\norm{\mathbf{x}}_p$ denote the $\ell^p$-norm of $\mathbf{x}$. Let $\mathbf{1}_n$ be the length-$n$ vector of all ones.

Given a matrix $\mathbf{A} \in \mathbb{R}^{m \times n}$, let $\bkt{\mathbf{A}}_{\ang{i}}$ denote its $i$th row represented as a column vector, let $\bkt{\mathbf{A}}_j$ denote its $j$th column, and let $\bkt{\mathbf{A}}_{i,j}$ denote the entry at row $i$ and column $j$. Let $\frnorm{\mathbf{A}}$ denote the matrix's Frobenius norm, let $\mathbf{A}^{-1}$ denote its inverse (when $m = n$ and $\mathbf{A}$ is non-singular), let $\mathbf{A}^\top$ denote its transpose, let $\sigma_j(\mathbf{A})$ denote its $j$th greatest singular value, and let $\sigma_{\min}(\mathbf{A})$ denote its $\min \brc{m, n}$th greatest singular value. For any $p \in \bkt{1, \infty}$, let $\norm{\mathbf{A}}_p$ denote the induced $\ell^p$-operator norm of $\mathbf{A}$. Let $\mathbf{I}_n$ be the $n \times n$ identity matrix. Let $\boldsymbol{\diag}(x_1, \dots, x_n) \in \mathbb{R}^{n \times n}$ be the diagonal matrix with diagonal entries $x_1, \dots, x_n$.

Given alphabet sets $\mathcal{X}$ and $\mathcal{Y}$ and random variables $X \in \mathcal{X}$ and $Y \in \mathcal{Y}$, denote the probability mass function of $X$ as the \emph{row} vector $\mathbf{p}_X \in \mathcal{S}_{\abs{\mathcal{X}} - 1}$, or equivalently, the function
\begin{align}
    p_X: \mathcal{X} \rightarrow [0, 1], \enspace p_X(x) = \bkt{\mathbf{p}_X}_x \, .
\end{align}
Denote the conditional probability distribution of $Y$ given $X$ as the \emph{row} stochastic matrix $\mathbf{P}_{Y|X} \in \mathbb{R}^{\abs{\mathcal{X}} \times \abs{\mathcal{Y}}}$, or equivalently, the kernel
\begin{align}
    p_{Y|X}: \mathcal{Y} \times \mathcal{X} \rightarrow [0, 1], \enspace p_{Y|X}(y|x) = \bkt{\mathbf{P}_{Y|X}}_{x,y} \, .
\end{align}

We occasionally use ``sequence-builder'' notation for tuples or multisets of subscripted variables:
\begin{align}
    \brc{x_i}_{i=1}^n &= (x_1, \dots, x_n) \, , \\
    \brc{x_{i,j}}_{i=1,j=1}^{m,n} &= (x_{1,1}, \dots, x_{1,n}, \dots, x_{m,1}, \dots, x_{m,n}) \, .
\end{align}
In the context of codewords, let $y_1^n = (y_1, \dots, y_n)$ and $\prn{x_i}_1^n = (x_{i,1}, \dots, x_{i,n})$. Throughout this paper, we measure Shannon entropy $\H{\cdot}$ and mutual information $\I{\cdot; \cdot}$ using bits.

\subsection{Formal Model} \label{subsection:formal-model}

\begin{figure*}[ht]
    \centering
    \includegraphics[trim = 10mm 66mm 10mm 68mm, clip, width=0.95\linewidth]{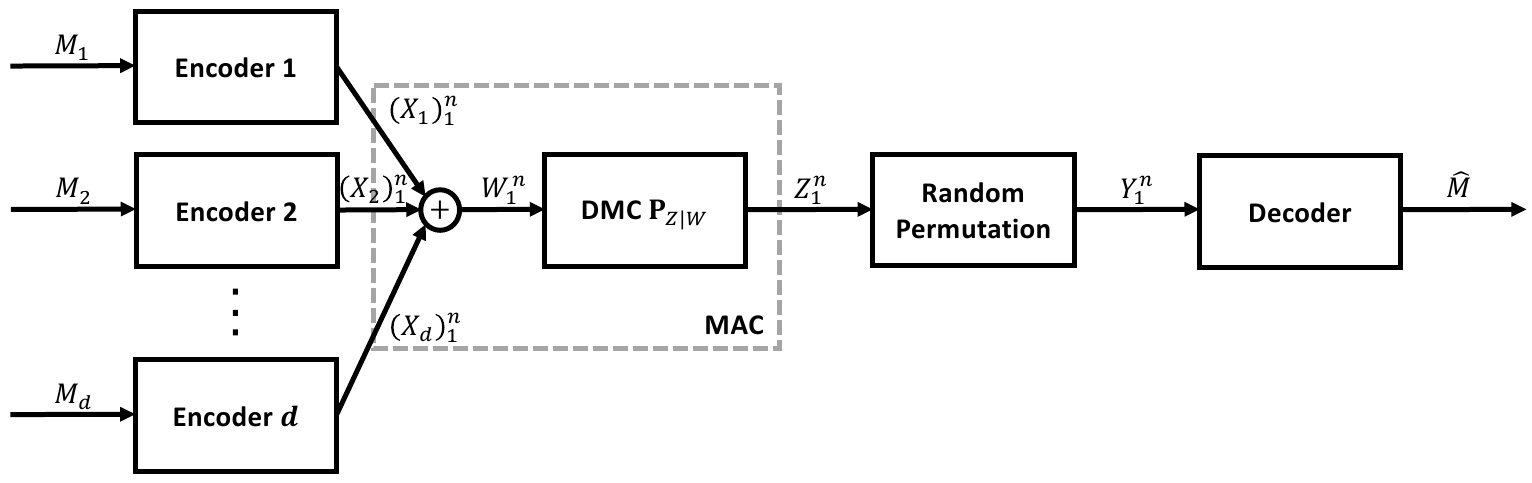}
    \caption{Illustration of a $d$-user PAMAC communication system.}
    \label{figure:formal-model}
\end{figure*}

Let $d \geq 2$ denote a fixed number of senders and $n \in \mathbb{N}$ denote the blocklength. We consider a simple single-hop network model consisting of the noisy $d$-user adder MAC followed by a random permutation block. We refer to this network model as the \emph{permutation adder multiple-access channel} (PAMAC).

In the PAMAC model, the senders (indexed by $i \in [d]$) seek to transmit independent \emph{messages} $M = (M_1, \dots, M_d)$, respectively. Each $M_i$ is uniformly distributed on a finite \emph{message set} $\mathcal{M}_i$ with $\abs{\mathcal{M}_i} \geq 2$. For notational simplicity, let $\mathcal{M} = \mathcal{M}_1 \times \cdots \times \mathcal{M}_d$.

Let $p \geq 2$ be the alphabet size of each sender's codeword. Let $\mathcal{X} = \bktz{p - 1}$ and $\mathcal{Y} = \bktz{d(p - 1)}$ be the input and output alphabets of the adder MAC, respectively. Let $f_{i,n}: \mathcal{M}_i \rightarrow \mathcal{X}^n$ be a (possibly randomized) \emph{encoder} for sender $i$, where the auxiliary randomness in the encoder functions $f_{1,n}, \dots, f_{d,n}$ are mutually independent. Let $g_n: \mathcal{Y}^n \rightarrow \mathcal{M} \cup \brc{\mathsf{error}}$ be a (possibly randomized) \emph{decoder} at the receiver, which may return an ``error''.

Each sender $i$ uses its encoder $f_{i,n}$ to encode its message $M_i$ into an $n$-length $p$-ary codeword $f_{i,n}(M_i) = \prn{X_i}_1^n \in \mathcal{X}^n$. These codewords are transmitted through a noiseless \emph{adder MAC} in a memoryless fashion:
\begin{align}
    \forall j \in [n], \enspace W_j = \sum_{i=1}^d X_{i,j} \, .
\end{align}
The output $W_1^n \in \mathcal{Y}^n$ then passes through a stationary DMC to produce $Z_1^n \in \mathcal{Y}^n$. This DMC is defined by an invertible and entry-wise strictly positive row stochastic matrix $\mathbf{P}_{Z|W} \in \mathbb{R}^{\abs{\mathcal{Y}} \times \abs{\mathcal{Y}}}$.\footnote{We impose the invertibility and strict positivity assumptions for analytical convenience, although they could be weakened in principle.} In particular, we have
\begin{align}
    \forall w_1^n, z_1^n \in \mathcal{Y}^n, \enspace p_{Z_1^n|W_1^n}(z_1^n|w_1^n) = \prod_{j=1}^n p_{Z|W}(z_j|w_j) \, .
\end{align}

Finally, $Z_1^n$ passes through an independent \emph{random permutation block} to generate the output codeword $Y_1^n \in \mathcal{Y}^n$ at the receiver. The random permutation block operates as follows. First, a random permutation $\sigma \in \Sg_n$ is drawn uniformly and independently of all other variables; the permutation $\sigma:[n] \rightarrow [n]$ is unknown to the senders and the receiver. Then, $Y_1^n$ is generated by permuting $Z_1^n$ according to $Y_j = Z_{\sigma(j)}$ for $j \in [n]$. At the end of transmission, the receiver decodes $Y_1^n$ using its decoder $g_n$ to produce \emph{estimates of the messages}
\begin{align}
    \hat{M} = g_n(Y_1^n) \in \mathcal{M} \cup \brc{\mathsf{error}} \, .
\end{align}
This communication system is illustrated in \cref{figure:formal-model}. We represent the true and empirical distributions of $W$, $Z$, and $Y$ as zero-indexed \emph{row} vectors in the $d(p - 1)$-dimensional probability simplex, e.g.,
\begin{align}
    \mathbf{p}_W, \hat{\mathbf{p}}_W \in \mathcal{S}_{d(p-1)} \subset \mathbb{R}^{d(p-1)+1} \, ,
\end{align}
and we zero-index the channel matrix $\mathbf{P}_{Z|W}$. In this context, let $\brc{\mathbf{e}_t}_{t=0}^{d(p-1)}$ be the standard basis \emph{column} vectors of $\mathbb{R}^{d(p-1)+1}$.

For any given code $(f_{1,n}, \dots, f_{d,n}, g_n)$, we let the \emph{average probability of error} be
\begin{align}
    \Perror{n} \triangleq \P{M \neq \hat{M}} \, ,
\end{align}
where the probability is computed with respect to the independent sources of randomness in the messages, the encoders, the DMC, the random permutation block, and the decoder. We define the \emph{rate $d$-tuple} $R = (R_1, \dots, R_d) \in \mathbb{R}_+^d$ of the code $(f_{1,n}, \dots, f_{d,n}, g_n)$ as (cf. \cite{Makur2020b})
\begin{align}
    (R_1, \dots, R_d) \triangleq \prn{\frac{\log_2 \abs{\M_1}}{\log_2 n}, \dots, \frac{\log_2 \abs{\M_d}}{\log_2 n}} \, .
\end{align}
We say that $R \in \mathbb{R}_+^d$ is \emph{achievable} if there exists a sequence of codes $\brc{(f_{1,n}, \dots, f_{d,n}, g_n)}_{n \in \mathbb{N}}$ with rate $d$-tuple $R$ such that $\lim_{n \rightarrow \infty} \Perror{n} = 0$. Lastly, we operationally define the \emph{permutation capacity region} $\Cregp$ of the PAMAC as the closure of the set of all achievable rate $d$-tuples:
\begin{align}
    \Cregp \triangleq \cl \brc{R \in \mathbb{R}_+^d: \text{$R$ is achievable}} \, ,
\end{align}
and the \emph{permutation sum-capacity} $\Csum$ of the PAMAC as the supremum of all achievable sum-rates (cf. \cite[Section 4.1]{ElGamalKim2011}):
\begin{align}
    \Csum \triangleq \sup \brc{\sum_{i=1}^d R_i: \text{$(R_1, \dots, R_d)$ is achievable}} \, .
\end{align}

\subsection{Main Contributions} \label{subsection:main-contributions}

In closing \cref{section:introduction}, we briefly enumerate the main results of our paper. We present three achievability bounds on the permutation capacity region of the PAMAC:
\begin{itemize}
    \item In \cref{theorem:binary-achievability-root-stability}, we restrict our analysis to the case of binary input alphabets and frame decoding as a polynomial root-finding problem by examining the probability generating function of the adder's output distribution. Using eigenvalue perturbation results from the literature, we analyze the spectral stability of the probability generating function's companion matrix, yielding a preliminary inner bound on $\Cregp$ for the binary PAMAC.
    \item In \cref{theorem:binary-achievability-time-sharing}, we devise a somewhat different coding scheme that allows individual senders to achieve greater rates than the approach presented in \cref{theorem:binary-achievability-root-stability}. Using formulations reminiscent of mixed-radix numerical systems, we extend the notion of time sharing to the permutation channel setting, yielding a tight inner bound on $\Cregp$ for the binary PAMAC.
    \item In \cref{theorem:general-achievability}, we adapt the message sets, encoders, and decoder from \cref{theorem:binary-achievability-time-sharing} to general input alphabets, and generalize the central techniques of \cref{theorem:binary-achievability-time-sharing} to derive a tight inner bound on $\Cregp$ for the $p$-ary PAMAC.
\end{itemize}
Next, we present two converse bounds on the permutation capacity region of the PAMAC that match our strongest achievability results:
\begin{itemize}
    \item In \cref{theorem:general-converse}, we adapt some arguments from the prior literature on permutation channels to derive a tight outer bound on $\Cregp$ for the $p$-ary PAMAC.
    \item In \cref{corollary:binary-converse}, we instantiate \cref{theorem:general-converse} on the $p = 2$ case, yielding a tight outer bound on $\Cregp$ for the binary PAMAC.
\end{itemize}
By juxtaposing the bounds listed above, we obtain explicit characterizations of the permutation capacity region and permutation sum-capacity:
\begin{itemize}
    \item \cref{theorem:general-permutation-capacity-region} is the principal contribution of our paper, and characterizes the permutation capacity region of the $p$-ary PAMAC by combining \cref{theorem:general-achievability,theorem:general-converse}.
    \item \cref{corollary:binary-permutation-capacity-region} specializes \cref{theorem:general-permutation-capacity-region} to the $p = 2$ case, characterizing the permutation capacity region of the binary PAMAC by combining \cref{theorem:binary-achievability-time-sharing,corollary:binary-converse}.
    \item \cref{corollary:general-permutation-sum-capacity} characterizes the permutation sum-capacity of the $p$-ary PAMAC as a direct consequence of \cref{theorem:general-permutation-capacity-region}.
\end{itemize}
We remark that the preliminary bound in \cref{theorem:binary-achievability-root-stability} is sufficient to characterize $\Csum$ for the binary PAMAC when combined with \cref{corollary:binary-converse}, and may thus be interpreted as an alternative achievability scheme for an extremal point on $\Cregp$ in light of the stronger result in \cref{theorem:binary-achievability-time-sharing}.

\subsection{Outline}

Finally, we provide a brief synopsis of the remaining sections of our paper. In \cref{section:main-results-and-discussion}, we present the formal mathematical statements of the results listed in \cref{subsection:main-contributions} and provide proof sketches that holistically elucidate the high-level intuition behind our main contributions. In \cref{section:binary-achievability-root-stability,section:general-achievability}, we prove \cref{theorem:binary-achievability-root-stability,theorem:general-achievability}, respectively. We defer the proof of \cref{theorem:binary-achievability-time-sharing} to \cref{appendix:binary-achievability-time-sharing}, as this result is a specialization of \cref{theorem:general-achievability} to the binary alphabet case; we include its proof nonetheless as a concise alternative for interested readers. Lastly, we prove \cref{theorem:general-converse} in \cref{section:converse-proof}, and provide proofs of auxiliary lemmas in \cref{appendix:auxiliary-results}.

\section{Main Results and Discussion} \label{section:main-results-and-discussion}

\subsection{Achievability Bound Using Root Stability} \label{subsection:binary-achievability-root-stability}

Our first main result is a preliminary inner bound on the permutation capacity region of the binary PAMAC (i.e., $p = 2$):

\begin{theorem}[Binary Achievability Using Root Stability] \label{theorem:binary-achievability-root-stability}
    The permutation capacity region of the binary PAMAC satisfies
    \begin{align}
        \Cregp \supseteq \bkt{0, \frac{1}{2}}^d \, .
    \end{align}
\end{theorem}

This theorem directly establishes the binary PAMAC's permutation sum-capacity as $\Csum = \frac{d}{2}$ when combined with the converse bound in \cref{corollary:binary-converse}. Thus, \cref{theorem:binary-achievability-root-stability} may be interpreted as providing an alternative coding scheme achieving the extremal rate $d$-tuple $R = \prn{\frac{1}{2}, \dots, \frac{1}{2}}$ in light of the tighter bound in \cref{theorem:binary-achievability-time-sharing}.

We defer the technical details of our proof to \cref{section:binary-achievability-root-stability}. Below, we describe our encoding and decoding schemes to elucidate the key insights in our proof. By definition of permutation capacity region, it suffices to show that $R = \prn{\frac{1}{2} - \alpha, \dots, \frac{1}{2} - \alpha}$ is achievable for any $\alpha > 0$. Thus, for any fixed $\alpha > 0$, consider the following message sets, encoders, and decoder.

\textbf{Message sets.} By definition of rate $d$-tuple, $\abs{\mathcal{M}_i} = n^{\frac{1}{2} - \alpha}$ for sender $i \in [d]$. Without loss of generality, assume $\mathcal{M}_i = \brc{\theta_{i,\ell}}_{\ell=1}^{\abs{\mathcal{M}_i}}$ where
\begin{align}
    \theta_{i,\ell} = \frac{2i - 1}{2d + 1} + \prn{\frac{\ell - 1}{\abs{\mathcal{M}_i} - 1}} \frac{1}{2d + 1} \, . \label{eq:root-stability-message-sets}
\end{align}
This is well-defined, i.e., the denominators $\abs{\mathcal{M}_i} - 1$ are non-zero, because our formal model stipulates that $\abs{\mathcal{M}_i} \geq 2$.

\textbf{Encoders.} Let $\brc{\theta_i}_{i=1}^d \in \mathcal{M}$ denote the messages to send. Given a message $\theta_i \in \mathcal{M}_i$, sender $i$'s randomized encoder outputs $n$ independent and identically distributed (i.i.d.) samples from a Bernoulli distribution with mean $\theta_i$. Formally,
\begin{align}
    \forall \ell \in [\abs{\mathcal{M}_i}], \enspace f_{i,n}(\theta_{i,\ell}) = \prn{X_i}_1^n \stackrel{\text{i.i.d.}}{\sim} \Bernoulli{\theta_{i,\ell}} \, .
\end{align}
Intuitively, we split the Bernoulli parameter space (the unit interval $[0, 1]$) into equal-length \emph{subintervals} and assign every other subinterval to a sender. This maintains padding between consecutive senders' subintervals and around the boundary points $0$ and $1$. We partition each sender's subinterval into equal-length \emph{slices} and assign each message to a slice boundary. Without loss of generality, we elide the difference between messages and Bernoulli parameters in our definition of message sets, for notational simplicity. \cref{figure:encoder} illustrates this setup.

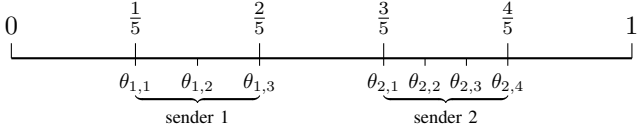
\begin{figure}[t] \label{figure:encoder}
    \begin{tikzpicture}
        \def\l{8.25}   
        \def\R{0.2}    
        \def\r{0.1}    
        
        \draw[thick] (0,1) -- (\l,1);
        
        \draw[thin] (0, 1) -- (0, 1 + \R) node [above] {$0$};
        \draw[thin] (1/5 * \l, 1 - \r) node [below, scale=0.8] {$\theta_{1,1}$} -- (1/5 * \l, 1 + \R) node [above] {$\frac{1}{5}$};
        \draw[thin] (2/5 * \l, 1 - \r) node [below, scale=0.8] {$\theta_{1,3}$} -- (2/5 * \l, 1 + \R) node [above] {$\frac{2}{5}$};
        \draw[thin] (3/5 * \l, 1 - \r) node [below, scale=0.8] {$\theta_{2,1}$} -- (3/5 * \l, 1 + \R) node [above] {$\frac{3}{5}$};
        \draw[thin] (4/5 * \l, 1 - \r) node [below, scale=0.8] {$\theta_{2,4}$} -- (4/5 * \l, 1 + \R) node [above] {$\frac{4}{5}$};
        \draw[thin] (\l, 1) -- (\l, 1 + \R) node [above] {$1$};
        
        \draw[thin] (3/10 * \l, 1 - \r) node [below, scale=0.8] {$\theta_{1,2}$} -- (3/10 * \l, 1);
        \draw[thin] (10/15 * \l, 1 - \r) node [below, scale=0.8] {$\theta_{2,2}$} -- (10/15 * \l, 1);
        \draw[thin] (11/15 * \l, 1 - \r) node [below, scale=0.8] {$\theta_{2,3}$} -- (11/15 * \l, 1);
        
        \draw[thick, decorate, decoration={calligraphic brace, mirror}] (1/5 * \l, 0.6 - \r) -- (2/5 * \l, 0.6 - \r)
        node [pos=0.5, anchor=north, yshift=-0.1cm, scale=0.7] {sender 1};
        \draw[thick, decorate, decoration={calligraphic brace, mirror}] (3/5 * \l, 0.6 - \r) -- (4/5 * \l, 0.6 - \r)
        node [pos=0.5, anchor=north, yshift=-0.1cm, scale=0.7] {sender 2};
    \end{tikzpicture}
    \caption{Visualization of our encoding scheme in the case $d = 2, \abs{\mathcal{M}_1} = 3, \abs{\mathcal{M}_2} = 4$. The $[0, 1]$ number line represents the Bernoulli parameter space. Each tick annotated with $\theta_{i,\ell}$ under the number line is a message in $\mathcal{M}_i$. Note the padding around each sender's subinterval.}
\end{figure}

\textbf{Decoder.} Given the output codeword $y_1^n$, the decoder $g_n: \bktz{d}^n \rightarrow \M$ executes the following:

\begin{enumerate}
    \item Compute the empirical distribution $\hat{\mathbf{p}}_Z \in \mathcal{S}_d$, given by
    \begin{align}
        \forall t \in \bktz{d}, \enspace \bkt{\hat{\mathbf{p}}_Z}_t = \frac{1}{n} \sum_{j=1}^n \Iv{y_j = t} \, .
    \end{align}
    \item Compute the estimated distribution
    \begin{align}
        \tilde{\mathbf{p}}_W = \hat{\mathbf{p}}_Z \mathbf{P}_{Z|W}^{-1} \in \mathbb{R}^{d+1} \, ,
    \end{align}
    represented as a zero-indexed \emph{row} vector. (Note that $\tilde{\mathbf{p}}_W$ is not a PMF, in general.)
    \item Form the estimated \emph{probability generating function} $\tilde{G}_W: \mathbb{C} \rightarrow \mathbb{C}$, given by
    \begin{align}
        \tilde{G}_W(\xi) = \sum_{t=0}^d \bkt{\tilde{\mathbf{p}}_W}_t \xi^t \, .
    \end{align}
    \item Compute the roots $\ibrc{\tilde{\xi}_i}_{i=1}^d \subset \mathbb{C}$ of $\tilde{G}_W$.\footnote{Technically, there are less than $d$ roots if $\tilde{G}_W$ has degree less than $d$, i.e., $\bkt{\tilde{\mathbf{p}}_W}_d = 0$. To be fully rigorous, in this exception the decoder may use dummy values outside $[0, 1]$ for the remaining estimated Bernoulli parameters. The specific way this exception is handled is immaterial, because the proof of \cref{theorem:binary-achievability-root-stability} restricts to the case where $\bkt{\tilde{\mathbf{p}}_W}_d \geq \frac{1}{2} \bkt{\mathbf{p}_W}_d > 0$.}
    \item Transform the roots into estimated Bernoulli parameters $\ibrc{\tilde{\theta}_i}_{i=1}^d \subset \mathbb{R}$, using the relations
    \begin{align}
        \forall i \in [d], \enspace \tilde{\theta}_i = \frac{1}{1 - \Re \brc{\tilde{\xi}_i}} \, . \label{eq:estimated-root-to-parameter-transform}
    \end{align}
    \item Sort the estimated Bernoulli parameters in ascending order $\tilde{\theta}_1 \leq \cdots \leq \tilde{\theta}_d$.
    \item Return the predicted messages $\ibrc{\hat{\theta}_i}_{i=1}^d \in \M$ given by $\hat{\theta}_i = \arg \min_{\theta \in \mathcal{M}_i} \iabs{\tilde{\theta}_i - \theta}$.
\end{enumerate}

In our model, each sender independently samples a message from its message set, with no collusion between senders possible when sending messages. However, the senders collude to establish the protocol described above for sharing the multiple-access channel, wherein each sender uses a disjoint subinterval of the Bernoulli parameter space. In this regard, our model is similar to standard regimes described in \cite[Section 4.1]{ElGamalKim2011}, where different senders' codewords $\prn{X_i}_1^n$ and $\prn{X_{i'}}_1^n$ are independent but each sender uses a different alphabet.

Next, we provide a high-level overview of the key techniques used in \cref{section:binary-achievability-root-stability}. Because the random permutation destroys the ordering of the output letters, we encode a message $\theta_i \in \mathcal{M}_i$ as samples from a Bernoulli distribution parameterized by $\theta_i$, since recovering the parameter from the samples is agnostic to the ordering of the samples. (In this sense, our strategy is comparable to the notion of \emph{multiset codes}, cf. \cite{KovacevicVukobratovic2015,Makur2020b}.) The decoder correctly rounds off the noise in a predicted parameter $\tilde{\theta}_i$ if $\theta_i$ is the closest of sender $i$'s messages to $\tilde{\theta}_i$. Thus, $\Perror{n}$ is upper bounded by the probability that for some sender $i$, the error in $\tilde{\theta}_i$ is no less than half the gap length between adjacent messages in sender $i$'s subinterval. In our encoding scheme, sender $i$'s messages are evenly spaced over a subinterval of length $\frac{1}{2d + 1}$. Consequently, sender $i$'s subinterval is split into $\abs{\mathcal{M}_i} - 1$ slices, and we want to upper bound
\begin{align}
    \P{\exists i \in [d], \, \abs{\tilde{\theta}_i - \theta_i} \geq \frac{1}{2(2d + 1)(\abs{\mathcal{M}_i} - 1)}} \, .
\end{align}

Due to the random permutation, the empirical probability vector $\hat{\mathbf{p}}_Y = \hat{\mathbf{p}}_Z$ is a sufficient statistic of $Z_1^n$. As a warmup, assume the empirical and true probability vectors of $Z$ match, i.e., $\hat{\mathbf{p}}_Z = \mathbf{p}_Z$. We can simulate running the DMC backwards by inverting its stochastic matrix, thus obtaining the true probability vector of $W$:
\begin{align}
    \mathbf{p}_W = \mathbf{p}_Z \mathbf{P}_{Z|W}^{-1} = \mathbf{\hat{p}}_Z \mathbf{P}_{Z|W}^{-1} \, .
\end{align}
Since $W$ is the sum of $d$ independent random variables $\brc{X_i}_{i=1}^d$, its probability generating function $G_W$ is the product of the probability generating functions $\brc{G_{X_i}}_{i=1}^d$ of the summands. Since $\brc{X_i}_{i=1}^d$ are Bernoulli random variables, $\brc{G_{X_i}}_{i=1}^d$ are linear functions. By the fundamental theorem of algebra, computing the factorization $G_W(\xi) = \prod_{i=1}^d G_{X_i}(\xi)$ reduces to finding the roots of $G_W$. Each root of $G_W$ corresponds to one of the $G_{X_i}$, so finding the roots of $G_W$ is sufficient to exactly recover the Bernoulli parameters. This intuition is formalized in \cref{prop:binary-adder-mac-output-distribution} below.

\begin{proposition}[Binary Adder MAC Output Distribution {\cite[Proposition 4.9]{Makur2019}}] \label{prop:binary-adder-mac-output-distribution}
    Let $W \in \bktz{d}$ be a random variable with an entry-wise strictly positive probability vector $\mathbf{p}_W > \mathbf{0}$. Then $W$ is the sum of $d$ independent Bernoulli random variables, i.e.,
    \begin{align}
        W = \sum_{i=1}^d X_i \enspace \text{with} \enspace X_i \sim \Bernoulli{p_i} \enspace \text{and} \enspace p_i \in (0, 1) \, ,
    \end{align}
    iff its probability generating function $G_W: \mathbb{C} \rightarrow \mathbb{C}$ given by
    \begin{align}
        G_W(\xi) = \E{\xi^W} = \sum_{t=0}^d \bkt{\mathbf{p}_W}_t \xi^t
    \end{align}
    has all real roots. Furthermore, the roots $\brc{\xi_i}_{i=1}^d$ of $G_W$, counted with multiplicity, determine the Bernoulli parameters $\brc{p_i}_{i=1}^d$ up to permutations of the indices, via the relations
    \begin{align}
        \forall i \in [d], \enspace \xi_i = \frac{p_i - 1}{p_i} \, .
    \end{align}
\end{proposition}

\begin{proof}
    Suppose $W = \sum_{i=1}^d X_i$, where $X_i \sim \Bernoulli{p_i}$ with $p_i \in (0, 1)$ are independent. Then,
    \begin{align}
        G_W(\xi) = \prod_{i=1}^d G_{X_i}(\xi) = \prn{\prod_{i=1}^d p_i} \prod_{i=1}^d \prn{\xi + \frac{1 - p_i}{p_i}} \, .
    \end{align}
    Thus, $G_W$ has all real roots: $\xi_i = \frac{p_i - 1}{p_i}$ for $i \in [d]$.
    
    Conversely, suppose $G_W$ has all real roots $\brc{\xi_i}_{i=1}^d \subset \mathbb{R}$. By the fundamental theorem of algebra,
    \begin{align}
        G_W(\xi) = \sum_{t=0}^d \bkt{\mathbf{p}_W}_t \xi^t = \beta \prod_{i=1}^d (\xi - \xi_i) \, .
    \end{align}
    By Descartes' rule of signs, $\brc{\xi_i}_{i=1}^d < 0$. Furthermore, none of the roots are zero because $\bkt{\mathbf{p}_W}_0 > 0$. Therefore, we may define $\brc{p_i}_{i=1}^d \subset (0, 1)$ via the relations $\xi_i = \frac{p_i - 1}{p_i}$ for $i \in [d]$. This yields
    \begin{align}
        G_W(\xi) &= \beta \prn{\prod_{i=1}^d p_i}^{-1} \prod_{i=1}^d \prn{1 - p_i + p_i \xi} \\
        &\stackclap{(a)}{=} \prod_{i=1}^d \prn{1 - p_i + p_i \xi} \\
        &\stackclap{(b)}{=} \prod_{i=1}^d G_{X_i}(\xi) \, ,
    \end{align}
    where (a) follows because $G_W(1) = \sum_{t=0}^d \bkt{\mathbf{p}_W}_t = 1$ and therefore $\beta = \prod_{i=1}^d p_i$, and we define independent $X_i \sim \Bernoulli{p_i}$ in (b). Hence, $W = \sum_{i=1}^d X_i$ as desired.
\end{proof}

\cref{prop:binary-adder-mac-output-distribution} generalizes the result in \cite[Lemma 1]{AjjanagaddePolyanskiy2015}, which proves the $d = 2$ case using a somewhat different approach. We also note that \cref{prop:binary-adder-mac-output-distribution} can be easily extended to include the edge cases where some $p_i \in \brc{0, 1}$.

Now consider the general case where $\hat{\mathbf{p}}_Z \neq \mathbf{p}_Z$ due to sampling noise. By Hoeffding's inequality, we can derive a high-probability upper bound on the infinity-norm error in $\hat{\mathbf{p}}_Z$. This incurs an error bound on the estimated probability vector $\tilde{\mathbf{p}}_W = \hat{\mathbf{p}}_Z \mathbf{P}_{Z|W}^{-1}$ and thus an error bound on the coefficients of $\tilde{G}_W$. Consequently, our problem reduces to bounding the stability of the roots of $\tilde{G}_W$, or equivalently, bounding the error in the eigenvalues of the Frobenius companion matrix of $\tilde{G}_W$ \cite[Definition 3.3.13]{HornJohnson2013}.

At a high level, the \emph{Bauer-Fike theorem} \cite[Theorem 3.3, Chapter IV]{StewartSun1990} from matrix perturbation theory states that the spectral stability of a diagonalizable matrix depends on the conditioning of the matrix's eigenbasis. By inspection, companion matrices are diagonalizable by Vandermonde matrices \cite[Section 0.9.11]{HornJohnson2013}, and the stability of Vandermonde systems has been studied in the prior literature \cite[Theorem 2.1]{Gautschi1990}. Our coding scheme maintains padding between consecutive senders' message sets, thereby enforcing sufficient separation between the roots of $G_W$ to ensure that the companion matrix's Vandermonde eigenbasis is well-conditioned. Hence, we can invoke the Bauer-Fike theorem to bound the error in the companion eigenvalues as desired, concluding our proof sketch.

We finish with several pertinent remarks. Firstly, although local Lipschitz continuity of polynomial roots is a well-studied phenomenon \cite{Brink2010}, extracting explicit Lipschitz constants in closed-form from these works is difficult. This motivates the spectral stability analysis in our proof of achievability. Secondly, while our proof focuses on establishing that the extremal rate $d$-tuple $R = \prn{\frac{1}{2} - \alpha, \dots, \frac{1}{2} - \alpha}$ is achievable for arbitrary $\alpha > 0$, it is straightforward to extend our analysis to any rate $d$-tuple in $\bigl[ 0, \frac{1}{2} \bigr)^d$, which yields \cref{theorem:binary-achievability-root-stability}.

Finally, our randomized encoders and decoder have polynomial time complexity with respect to $n$. The encoders each run in $\bigO{n}$ time since they take $n$ samples from a $\Bernoulli{\theta}$ distribution; each sample can be done in $\bigO{1}$ time by sampling $U \sim \Uniform{0}{1}$ to a fixed precision and computing $X = \Iv{U \geq 1 - \theta}$. The first decoding step (computing $\hat{\mathbf{p}}_Z$) costs $\bigO{n + d} = \bigO{n}$ time, steps 2 through 6 cost $\bigO{\poly(d)} = \bigO{1}$, and step 7 costs $O \iprn{\sum_{i=1}^d n^{R_i}}$. Hence, our randomized coding scheme does not suffer from intractable decoding complexity.

\subsection{Achievability Bounds Using Time Sharing} \label{subsection:general-achievability}

Our second main result is a tight inner bound on the permutation capacity region of the binary PAMAC (i.e., $p = 2$), which matches our converse result presented later in \cref{subsection:converse-bounds}:

\begin{theorem}[Binary Achievability Using Time Sharing] \label{theorem:binary-achievability-time-sharing}
    The permutation capacity region of the binary PAMAC satisfies
    \begin{align}
        \Cregp \supseteq \brc{R \in \mathbb{R}_+^d: \sum_{i=1}^d R_i \leq \frac{d}{2}} \, .
    \end{align}
\end{theorem}

Our third main result is an extension of this inner bound to the general $p$-ary PAMAC:

\begin{theorem}[General Achievability] \label{theorem:general-achievability}
    The permutation capacity region of the $p$-ary PAMAC satisfies
    \begin{align}
        \Cregp \supseteq \brc{R \in \mathbb{R}_+^d: \sum_{i=1}^d R_i \leq \frac{d(p - 1)}{2}} \, .
    \end{align}
\end{theorem}

We defer the technical details of the general achievability proof to \cref{section:general-achievability}, and provide the binary achievability proof as a more concise and intuitive alternative in \cref{appendix:binary-achievability-time-sharing}. We organize our proof sketch into three steps. Call a sender \emph{active} when it is sending randomly-generated letters that encode its message, and \emph{passive} when it is sending hard-coded letters that do not encode its message.

\emph{Step 1: Achieving the desired permutation sum-capacity for binary alphabets.} In this step, we fix $p = 2$ and describe a simple coding scheme which achieves $\Csum = \frac{d}{2}$ and hence matches the converse bound's permutation sum-capacity. We consider a de facto ``single-access'' setting where only sender $1$ actively sends messages at a positive rate; all other senders passively transmit a deterministic code at rate zero on each iteration.\footnote{Essentially, each sender $2$ through $d$ has a singleton message set and repeatedly sends the sole message in its set. This interpretation violates the stipulation in our formal model that $\abs{\mathcal{M}_i} \geq 2$, but this discrepancy is immaterial in the context of our preliminary analysis and will be remedied in Step 2 and our formal proofs.} In this scenario, the receiver merely has to recover sender $1$'s message from the output codeword $Y_1^n$. The crux of our proof is to cleverly construct deterministic codes for the passive senders that maximally aid the receiver in this task, allowing sender $1$ to achieve rate $\frac{d}{2}$.

We split the indices $[n] = \brc{1, \dots, n}$ of the codeword letters into $d$ equal-length \emph{segments} indexed by $c \in [d]$. Each segment is a contiguous integer interval of $\frac{n}{d}$ indices. Without loss of generality, let each message in sender $1$'s message set be a $d$-tuple of numbers, where each number is chosen from a linearly spaced grid in $[0, 1]$. Given a message $(\theta_1, \dots, \theta_d) \in \mathcal{M}_1$, sender $1$ encodes the $c$th \emph{component} $\theta_c$ within segment $c$ using i.i.d. samples from a Bernoulli distribution parameterized by $\theta_c$. (Essentially, sender $1$ adopts the randomized code described in \cref{subsection:binary-achievability-root-stability}, but treats each segment as an independent subset of letters to apply the randomized code within.)

To motivate the crucial insight underpinning our argument, we recap the intuition behind how the rates $R_i = \frac{1}{2}$ were achieved in \cref{theorem:binary-achievability-root-stability}. Under the root stability approach, successful recovery of the messages $\brc{\theta_i}_{i=1}^d$ was contingent on accurately estimating the true output distribution $\mathbf{p}_Y = \mathbf{p}_Z$. By Hoeffding's inequality, given $n$ samples, the empirical distribution $\hat{\mathbf{p}}_Z$ approximates the true distribution within error
\begin{align}
    \norm{\mathbf{p}_Z - \hat{\mathbf{p}}_Z}_\infty \leq \sqrt{\frac{\log_e n}{n}} = \littleO{n^{-\frac{1}{2} + \alpha}} \, ,
\end{align}
with high probability. By Lipschitz continuity and root stability, the same asymptotic error bound holds on the estimates $\ibrc{\tilde{\theta}_i}_{i=1}^d$. Thus, the decoder correctly rounds $\tilde{\theta}_i$ to the true message $\theta_i$ if the gaps between adjacent messages in $\mathcal{M}_i$ are at least $\Omega \iprn{n^{-\frac{1}{2} + \alpha}}$ in length, or equivalently, if $\mathcal{M}_i$ has at most $O \iprn{n^{\frac{1}{2} - \alpha}}$ messages. In short, the rate $R_i = \frac{1}{2}$ was a consequence of the Hoeffding bound on an empirical distribution vector computed from $n$ samples.

Next, we apply this intuition to our present setting. Since each of the $d$ segments contains $\frac{n}{d}$ letters, the randomized Bernoulli code allows each segment, in isolation, to encode a value from a set of size
\begin{align}
    \prn{\frac{n}{d}}^{\frac{1}{2} - \alpha} \approx \frac{n^{\frac{1}{2}}}{\sqrt{d}} = \bigO{n^{\frac{1}{2}}} \, .
\end{align}
The crucial step in our proof is to define the following deterministic code, which ensures the receiver can recover \emph{each} segment's encoded value (matched with the segment's index) from $Y_1^n$: Each passive sender $i \in \bktz{2, d}$ transmits all ones in segments $c \geq i$ and all zeros otherwise. Equivalently, in any segment $c \in [d]$, senders $i \in \bktz{2, c}$ transmit ones and senders $i \in \bktz{c + 1, d}$ transmit zeros. It follows that the output codeword of the adder MAC is equal to sender $1$'s codeword with a \emph{domain shift} of $c - 1$ in segment $c$, i.e.,
\begin{align}
    W_j = X_{1,j} + (c - 1)
\end{align}
for all indices $j$ in segment $c$. (Note that the alphabets of $W$ from different segments intersect at no more than one value, namely value $c$ for a pair of adjacent segments $c$ and $c + 1$ and no value otherwise.)

Combined with the full-rank assumption on the DMC, this imposes a unique structure on the distributions of $Z$ from different segments, allowing segment-specific information to be recovered even after the random permutation block. In essence, our coding scheme enables sender $1$ to outperform the rate achieved in \cref{theorem:binary-achievability-root-stability} by probabilistically overcoming the destructive effect of the random permutation on ordering information, leading to a combinatorial increase in the number of messages the decoder can distinguish between:
\begin{align}
    \abs{\mathcal{M}_1} \approx \prn{\frac{n^{\frac{1}{2}}}{\sqrt{d}}}^d = \bigO{n^{\frac{d}{2}}} \, .
\end{align}
Our segmentation procedure thus incurs an immaterial \emph{multiplicative} penalty of $\sqrt{d}$ within each segment, but introduces an \emph{exponent} of $d$ in the number of distinguishable messages when all segments are taken into consideration.

Next, we describe how the decoder recovers $(\theta_1, \dots, \theta_d)$ from $Y_1^n$ with high probability. For the remainder of this step, let $\brc{\mathbf{v}_t}_{t=0}^d$ denote the rows of the channel $\mathbf{P}_{Z|W}$, namely the conditional output distributions for each input value. We use the terms \emph{input distribution} and \emph{output distribution} with respect to the DMC, i.e., to respectively refer to $\mathbf{p}_W$ and $\mathbf{p}_Z$. Let $\mathbf{p}_{W,c}$ and $\mathbf{p}_{Z,c}$ denote the input and output distribution of segment $c$, respectively.

Since sender $1$ transmits a binary codeword, the alphabet of $W$ in each segment has size $2$. Equivalently, the input distribution of a segment $c \in [d]$ is the \emph{two-hot} vector
\begin{align}
    \mathbf{p}_{W,c} = (1 - \theta_c) \mathbf{e}_{c-1} + \theta_c \mathbf{e}_c \, .
\end{align}
The output distribution of segment $c$ is thus a \emph{convex combination} of two consecutive channel rows, with the convex coefficients in $\mathbf{p}_{W,c}$ corresponding to the Bernoulli probabilities used by sender $1$ in segment $c$:
\begin{align}
    \mathbf{p}_{Z,c} = (1 - \theta_c) \mathbf{v}_{c-1} + \theta_c \mathbf{v}_c \, .
\end{align}
The overall output distribution is the equally-weighted mean of the segment-specific output distributions, since all segments have equal length:
\begin{align}
    \mathbf{p}_Z = \frac{1}{d} \sum_{c=1}^d \mathbf{p}_{Z,c} = \frac{1}{d} \sum_{c=1}^d \brc{(1 - \theta_c) \mathbf{v}_{c-1} + \theta_c \mathbf{v}_c} \, .
\end{align}

Since the DMC $\mathbf{P}_{Z|W}$ is full-rank, its rows $\brc{\mathbf{v}_t}_{t=0}^d$ form a basis of $\mathbb{R}^{d+1}$. Each segment's output distribution lies in the \emph{span} of two consecutive basis vectors $\mathbf{v}_{c-1}$ and $\mathbf{v}_c$. Equivalently, each basis vector (except the first and last vectors $\mathbf{v}_0$ and $\mathbf{v}_d$) in the above representation of $\mathbf{p}_Z$ is \emph{weighted} by two consecutive Bernoulli parameters, as evidenced by simple rearrangement:
\begin{align} \label{eq:consecutive}
    \mathbf{p}_Z = \frac{1}{d} \prn{(1 - \theta_1) \mathbf{v}_0 + \sum_{c=1}^{d-1} (1 + \theta_c - \theta_{c+1}) \mathbf{v}_c + \theta_d \mathbf{v}_d} \, .
\end{align}

As a warmup, assume sampling noise is absent and thus $\hat{\mathbf{p}}_Y = \hat{\mathbf{p}}_Z = \mathbf{p}_Z$. The decoder computes and represents $\hat{\mathbf{p}}_Z$ with respect to the basis induced by the DMC, obtaining coefficients $\brc{\gamma_t}_{t=0}^d$ such that $\hat{\mathbf{p}}_Z = \sum_{t=0}^d \gamma_t \mathbf{v}_t$, and automatically recovering $\theta_d = \gamma_d$. Next, the decoder recovers $\brc{\theta_c}_{c=1}^{d-1}$ in reverse order by back-substituting the known variables into the equations $\gamma_c = 1 + \theta_c - \theta_{c+1}$ for $c \in [d - 1]$. Since $\hat{\mathbf{p}}_Z = \mathbf{p}_Z$ possesses the structure of \cref{eq:consecutive}, the value of $\theta_1$ recovered in this manner does not contradict the extra equation $\gamma_0 = 1 - \theta_1$.

Finally, consider the general case where $\hat{\mathbf{p}}_Z \neq \mathbf{p}_Z$ due to sampling noise. After expressing $\hat{\mathbf{p}}_Z$ with respect to the basis $\brc{\mathbf{v}_t}_{t=0}^d$, the decoder computes the least squares solution to the overdetermined mapping from $\brc{\gamma_t}_{t=0}^d$ to $\brc{\theta_c}_{c=1}^d$. Using various bounds on matrix norms, our achievability proof relates the concentration bound on $\hat{\mathbf{p}}_Z$ to the least squares approximation error, providing quantitative guarantees on decoding performance as desired. This concludes Step 1.

We remark that our segmentation procedure is fundamentally capable of creating at most $d$ segments, because there are $d - 1$ passive senders and thus $d$ unique domain shifts. Nonetheless, the permutation sum-capacity attained using $d$ segments matches the converse bound, so no additional segments are needed.

\emph{Step 2: Achieving the desired permutation capacity region by time sharing.} In this step, we extend the notion of time sharing to the permutation channel setting to arbitrarily distribute the PAMAC's sum-capacity among the $d$ senders. Combined with the characterization of $\Csum$ from Step 1, this implies
\begin{align}
    \Cregp \supseteq \brc{R \in \mathbb{R}_+^d: \sum_{i=1}^d R_i \leq \frac{d}{2}} \, ,
\end{align}
thus proving \cref{theorem:binary-achievability-time-sharing}.

Under our time sharing strategy, we partition each segment into $d$ \emph{subsegments} (indexed by $b \in [d]$) with carefully chosen lengths. In each segment, sender $i \in [d]$ actively encodes the respective component of its message $(\theta_{i,1}, \dots, \theta_{i,d}) \in \mathcal{M}_i$ with the aforementioned Bernoulli coding scheme in subsegment $b = i$, and passively transmits a deterministic code in all other subsegments to contribute to domain shifting. Hence, each segment contains active letter indices for each sender, and only one sender is active at any given letter index. Our subsegmentation scheme is thus comparable to classic notions of \emph{time division}, wherein only one sender transmits in each time slot \cite[Section 4.4]{ElGamalKim2011}.

Note that our segmentation and subsegmentation procedures serve the orthogonal purposes of \emph{increasing} a single sender's rate by maintaining ordering information and \emph{distributing} this increased rate arbitrarily among the senders, respectively. \cref{table:binary-encoder} visualizes the encoders' behavior in the case of $d = 3$ senders.
\begin{table*}[t]
    \caption{Behavior of the encoders in the case of $d = 3$ senders. Each column corresponds to a subsegment in a segment. Each row corresponds to a sender. Each cell contains the value returned by the encoder, which is a sample from a Bernoulli distribution (in active phase) or a constant (in passive phase). Observe that the deterministic codes in segment $c$ produce a domain shift of $c - 1$, commensurate with the approach for a single active sender discussed in Step 1.}
    \centering
    \begin{tabular}{c|c|c|c|c|c|c|c|c|c|}
        \cline{2-10}
        & \multicolumn{3}{c|}{segment 1} & \multicolumn{3}{c|}{segment 2} & \multicolumn{3}{c|}{segment 3} \\
        \cline{2-10}
        & sub 1 & sub 2 & sub 3 & sub 1 & sub 2 & sub 3 & sub 1 & sub 2 & sub 3 \\
        \hline
        \multicolumn{1}{|c|}{sender 1} & $\mathsf{Ber}(\theta_{1,1})$ & $0$ & $0$ & $\mathsf{Ber}(\theta_{1,2})$ & $1$ & $1$ & $\mathsf{Ber}(\theta_{1,3})$ & $1$ & $1$ \\
        \multicolumn{1}{|c|}{sender 2} & $0$ & $\mathsf{Ber}(\theta_{2,1})$ & $0$ & $1$ & $\mathsf{Ber}(\theta_{2,2})$ & $0$ & $1$ & $\mathsf{Ber}(\theta_{2,3})$ & $1$ \\
        \multicolumn{1}{|c|}{sender 3} & $0$ & $0$ & $\mathsf{Ber}(\theta_{3,1})$ & $0$ & $0$ & $\mathsf{Ber}(\theta_{3,2})$ & $1$ & $1$ & $\mathsf{Ber}(\theta_{3,3})$ \\
        \hline
    \end{tabular}
    \label{table:binary-encoder}
\end{table*}

Next, we briefly distill the essence of Step 1 to contextualize the intuition behind our achievability proof. In Step 1, the decoder recovered least-squares estimates of $d$ scalars $\brc{\theta_c}_{c=1}^d$, where each $\theta_c$ was the true proportion of ones among the letters actively sent in segment $c$. Since only sender $1$ was active, $\theta_c$ was precisely the parameter of the Bernoulli distribution generating sender $1$'s letters in segment $c$. Thus, each recovered scalar contributed rate $\frac{1}{2}$ for sender $1$, but contributed nothing for all other senders.

The crux of Step 2 is to \emph{share each recovered scalar among all the senders} by constructing a bijection $h$ between the Bernoulli parameters $\brc{\theta_b}_{b=1}^d$ used in a segment and the active proportion of ones $\phi$ in that segment.\footnote{For notational simplicity, we elide the subscript $c$ from $\phi$ and $\theta_b$ for the remainder of this step, since our time sharing strategy is identical for, and self-contained within, each segment. We emphasize that $\theta_c$ in Step 1 is subscripted by \emph{segment} index $c$, and $\theta_b$ in Step 2 is subscripted by \emph{subsegment} index $b$.} Then, accurate recovery of $\phi$ by the least-squares decoding procedure contributes information about each sender's message. The bijection $h$ controls how the information content of $\phi$ is distributed among the senders: By changing $h$, we alter the granularity at which each $\theta_b$ is encoded, allowing different rate $d$-tuples to be achieved. Implementing such a bijection in our coding scheme entails choosing the subsegment lengths to satisfy
\begin{align}
    \sum_{b=1}^d \rho_b \theta_{b} = h(\theta_{1}, \dots, \theta_{d}) \, , \label{eq:bijection-canonical-form}
\end{align}
where $\rho_b$ is the proportion of letter indices assigned to subsegment $b$. The left-hand side of this equation denotes that $\phi$ is a convex combination of the Bernoulli parameters, weighted by the relative sizes of the corresponding subsegments. The right-hand side asserts that the desired value of $\phi$ is specified by the bijection $h$.

As a warmup, assume $n$ and $R$ are such that each sender's message set has size $\abs{\mathcal{M}_i} = 10^d$, and so each $\theta_b$ lies on a grid of $10$ linearly spaced points between $0$ and $1$. We use the closed interval $[0, 1]$ for $b < d$ and the half-open interval $[0, 1)$ for $b = d$:
\begin{align}
    \theta_b = \begin{cases}
        \frac{\ell_b}{9} \, , & \text{if $b < d$} \\
        \frac{\ell_b}{10} \, , & \text{if $b = d$}
    \end{cases} \quad \text{where} \enspace \ell_b \in \bktz{9} \, . \label{eq:theta-b-base-10}
\end{align}
In this setting, a natural encoding of $\brc{\theta_b}_{b=1}^d$ within a scalar variable $\phi$ arises from taking each $\ell_b$ to be the $b$th most significant digit in the decimal representation of $\phi$:
\begin{align}
    h(\theta_1, \dots, \theta_d) = 0.\ell_1 \ell_2 \cdots \ell_d = \sum_{b=1}^d \frac{\ell_b}{10^b} \, . \label{eq:phi-base-10}
\end{align}
Combining \cref{eq:theta-b-base-10} and \cref{eq:phi-base-10} yields a characterization of $\phi$ as a weighted sum of the Bernoulli parameters $\brc{\theta_b}_{b=1}^d$:
\begin{align}
    \phi = \sum_{b=1}^{d-1} \underbrace{\frac{9}{10^b}}_{\rho_b} \theta_b + \underbrace{\frac{10}{10^d}}_{\rho_d} \theta_d \, .
\end{align}
By induction, the last $t$ weights evidently sum to $10^{-d+t}$. Hence, the entire set of weights $\brc{\rho_b}_{b=1}^d$ sums to $1$ and we may interpret $\rho_b$ as the proportional size of subsegment $b$, in accordance with \cref{eq:bijection-canonical-form}.\footnote{This motivates the use of closed and half-open intervals in \cref{eq:theta-b-base-10}.} This allocation of indices to senders results in $\phi = 0.\ell_1 \ell_2 \cdots \ell_d$ being the proportion of ones among the letters actively sent in segment $c$, as desired.

Next, we extrapolate this approach to the general case involving arbitrary rate $d$-tuples with sum-rate $\frac{d}{2} - \alpha$. Let $m_i = n^{\frac{R_i}{d}}$ be the size of the domain of $\theta_i$. (The warmup above corresponds to the special case where $m_i = 10$ for all $i \in [d]$.) Choosing linearly spaced points between $0$ and $1$ for the Bernoulli parameters, we have
\begin{align}
    \theta_b = \begin{cases}
        \frac{\ell_b}{m_b - 1} \, , & \text{if $b < d$} \\
        \frac{\ell_b}{m_b} \, , & \text{if $b = d$}
    \end{cases} \quad \text{where} \enspace \ell_b \in \bktz{m_b - 1} \, . \label{eq:theta-b-mixed-radix}
\end{align}
We encode $\brc{\theta_b}_{b=1}^d$ into $\phi$ by taking each $\ell_b$ to be the $b$th most significant digit in the \emph{mixed-radix} representation \cite[p. 208-209]{Knuth1997}, \cite{Cantor1869} of $\phi$, where the $b$th most significant digit has base $m_b$:
\begin{align}
    h(\theta_1, \dots, \theta_d) &= 0.(\ell_1)_{m_1} (\ell_2)_{m_2} \cdots (\ell_d)_{m_d} \\
    &= \sum_{b=1}^d \frac{\ell_b}{\prod_{i=1}^b m_i} \, . \label{eq:phi-mixed-radix}
\end{align}
Finally, the characterization of $\phi$ as a weighted sum of the Bernoulli parameters $\brc{\theta_b}_{b=1}^d$, obtained by combining \cref{eq:theta-b-mixed-radix} and \cref{eq:phi-mixed-radix}, is
\begin{align}
    \phi = \sum_{b=1}^{d-1} \underbrace{\frac{m_b - 1}{\prod_{i=1}^b m_i}}_{\rho_b} \theta_b + \underbrace{\frac{m_d}{\prod_{i=1}^d m_i}}_{\rho_d} \theta_d \, . \label{eq:phi-characterization-mixed-radix}
\end{align}
For a rigorous derivation of these equations, refer to the discussion of time sharing in \cref{section:general-achievability}.

At a high level, this mapping interprets $\brc{\theta_{b}}_{b=1}^d$ as normalized axis-wise indices into a multi-dimensional array of message components, and $\phi$ as the corresponding normalized index for the flattened one-dimensional view of the array. The domain of $\phi$ is a sufficiently fine grid (whose gap width is determined by sum-rate) to encode the $\brc{\theta_{b}}_{b=1}^d$ values lying on coarser grids (whose gap widths are determined by the individual senders' rates). \cref{figure:time-sharing} visualizes this interpretation in the case of $d = 2$ senders. The single-dimensional array along each axis is the grid of linearly spaced points from which the corresponding sender's message components $\theta_i$ are chosen; we label this domain as $\sqrt[d]{\mathcal{M}_i}$ with some abuse of notation. Each cell in the two-dimensional array represents a possible combination of transmitted message components, labeled with the corresponding $\phi$ value that encodes the combination.
\begin{figure}[b]
    \centering
    \begin{tblr}{
        colspec = {Q[c,m]Q[c,m]Q[c,m]Q[c,m]Q[c,m]Q[c,m]Q[c,m]},
        row{3} = {0mm, rowsep=8pt},
        column{3} = {0mm, colsep=8pt},
        stretch = 0,
        hline{2} = {4-7}{},
        hline{3} = {4-7}{},
        hline{4} = {2-2, 4-7}{},
        hline{5} = {2-2, 4-7}{},
        hline{6} = {2-2, 4-7}{},
        hline{7} = {2-2, 4-7}{},
        vline{2} = {4-6}{},
        vline{3} = {4-6}{},
        vline{4} = {2-2, 4-6}{},
        vline{5} = {2-2, 4-6}{},
        vline{6} = {2-2, 4-6}{},
        vline{7} = {2-2, 4-6}{},
        vline{8} = {2-2, 4-6}{}
    }
        & & & \SetCell[c=4]{} $\sqrt[d]{\mathcal{M}_2}$ & & & \\
        & & & {$\ell_2 = 0$ \\ $\theta_2 = \frac{0}{4}$} & {$\ell_2 = 1$ \\ $\theta_2 = \frac{1}{4}$} & {$\ell_2 = 2$ \\ $\theta_2 = \frac{2}{4}$} & {$\ell_2 = 3$ \\ $\theta_2 = \frac{3}{4}$} \\
        & & & & & & \\
        \SetCell[r=3]{} \rotatebox[origin=c]{90}{$\sqrt[d]{\mathcal{M}_1}$} & {$\ell_1 = 0$ \\ $\theta_1 = \frac{0}{2}$} & & $\phi = \frac{0}{12}$ & $\phi = \frac{1}{12}$ & $\phi = \frac{2}{12}$ & $\phi = \frac{3}{12}$ \\
        & {$\ell_1 = 1$ \\ $\theta_1 = \frac{1}{2}$} & & $\phi = \frac{4}{12}$ & $\phi = \frac{5}{12}$ & $\phi = \frac{6}{12}$ & $\phi = \frac{7}{12}$ \\
        & {$\ell_1 = 2$ \\ $\theta_1 = \frac{2}{2}$} & & $\phi = \frac{8}{12}$ & $\phi = \frac{9}{12}$ & $\phi = \frac{10}{12}$ & $\phi = \frac{11}{12}$
    \end{tblr}
    \caption{Visual interpretation of our time sharing bijection in the case $d = 2$, $m_1 = 3$, and $m_2 = 4$ (hence, $\abs{\mathcal{M}_1} = 9$ and $\abs{\mathcal{M}_2} = 16$).}
    \label{figure:time-sharing}
\end{figure}

We remark that the lengths of the $d$ subsegments within each segment are monotonically non-increasing, regardless of the rate distribution among the senders. In the analysis below, assume $n$ is sufficiently large such that each $m_b \geq 2$. Then, for every $b \in [d - 1]$,
\begin{align}
    \rho_b &\stackclap{(a)}{=} \frac{m_b - 1}{\prod_{i=1}^b m_i} \\
    &\geq \frac{1}{\prod_{i=1}^b m_i} \\
    &\geq \frac{m_{b+1} - \Iv{b + 1 < d}}{\prod_{i=1}^{b+1} m_i} \\
    &\stackclap{(b)}{=} \rho_{b+1} \, ,
\end{align}
where (a) and (b) follow from \cref{eq:phi-characterization-mixed-radix}. Thus, irrespective of the rate $d$-tuple, sender $1$ sends the most active letters and sender $d$ sends the least active letters in each of the $d$ segments. We note that it is unnecessary to rotate the assignment of senders to subsegments to prevent any one sender from consistently receiving the shortest subsegment. Under the current time sharing procedure, the size of the $b$th subsegment is
\begin{align}
    \rho_b \, \frac{n}{d} &\stackclap{(a)}{\geq} \frac{m_b}{2 \prod_{i=1}^b m_i} \cdot \frac{n}{d} \\
    &\stackclap{(b)}{=} \frac{1}{2 \prod_{i=1}^{b-1} n^{\frac{R_i}{d}}} \cdot \frac{n}{d} \\
    &= \frac{1}{2d} \, n^{1 - \sum_{i=1}^{b-1} \frac{R_i}{d}} \\
    &\stackclap{(c)}{=} \frac{1}{2d} \, n^{\frac{1}{2} + \frac{\alpha}{d} + \sum_{i=b}^d \frac{R_i}{d}} \\
    &\stackclap{(d)}{\geq} \frac{1}{2d} \, n^{\frac{1}{2} + \frac{R_b}{d} + \frac{\alpha}{d}} \\
    &\stackclap{(e)}{\geq} \frac{1}{2d} \, n^{\frac{2R_b}{d} + \frac{\alpha}{d}} \, ,
\end{align}
where (a) follows from \cref{eq:phi-characterization-mixed-radix}, (b) follows because $m_i = n^{\frac{R_i}{d}}$, (c) follows because we consider rate $d$-tuples with sum-rate $\frac{d}{2} - \alpha$, (d) follows because rates are non-negative, and (e) follows because $\frac{1}{2} > \sum_{i=1}^d \frac{R_i}{d} \geq \frac{R_b}{d}$. Thus, the $b$th subsegment can encode
\begin{align}
    \frac{1}{\sqrt{2d}} \, n^{\frac{R_b}{d} + \frac{\alpha}{2d}} \gtrsim n^{\frac{R_b}{d}}
\end{align}
messages, as desired.

We finish Step 2 with a cautionary remark. As prior work \cite{Makur2018} has shown that the permutation capacity of the binary symmetric channel is $\frac{1}{2}$, it is tempting (but incorrect) to conclude that $\Cregp \subseteq \bkt{0, \frac{1}{2}}^d$ for the binary PAMAC, since sharing a communication channel with additional senders should not increase any individual sender's rate. This intuition is flawed because the dimensionality of the DMC in the PAMAC increases with the number of senders, so any conjecture about the PAMAC's permutation capacity region should consider the alphabet size of the DMC's input codeword $W_1^n$ instead of each encoder's output codeword $X_1^n$. (Indeed, prior work \cite{Makur2020b} establishes the permutation capacity of a full-rank, strictly-positive $(d + 1) \times (d + 1)$ DMC as $C_\mathsf{perm} = \frac{d}{2}$.)

\emph{Step 3: Generalizing to $p$-ary alphabets.} The techniques described in Steps 1 and 2 naturally extend to general $p$-ary alphabets, with modifications to exploit the expanded alphabet size. Active letters are generated by sampling from a categorical distribution over $\bktz{p - 1}$ instead of a Bernoulli distribution. Each component of a message $(\boldsymbol{\theta}_{i,1}, \dots, \boldsymbol{\theta}_{i,d}) \in \mathcal{M}_i$ is a categorical PMF chosen from a lattice $\mathcal{L}_i$ embedded in $\mathcal{S}_{p-1}$. The lattice $\mathcal{L}_i$ is the higher-dimensional analogue of the grid of linearly spaced points between $0$ and $1$ defined in \cref{eq:theta-b-mixed-radix} and visualized along the $i$th axis of \cref{figure:time-sharing}. A visualization of $\mathcal{L}_i$ for $p = 3$ is provided in \cref{figure:message-set}. To maintain the property that the alphabets of $W$ from adjacent segments overlap at exactly one value, passive senders transmit $p - 1$ values to contribute to a domain shift of $(c - 1)(p - 1)$ in segment $c$.

The sent messages $M \in \mathcal{M}$ are comprised of scalar variables $\brc{\theta_{i,c,k}}_{i=1,c=1,k=0}^{d,d,p-1}$ subscripted by indices denoting the sender, segment, and alphabet symbol, respectively. The time sharing bijection computes the weighted sum in \cref{eq:phi-characterization-mixed-radix} along the sender axis, flattening this three-dimensional array of scalars into $\brc{\phi_{c,k}}_{c=1,k=0}^{d,p-1}$. The decoder recovers least-squares estimates of these $dp$ scalars from an overdetermined system of $dp + 1$ equations:
\begin{itemize}
    \item Analogously to \cref{eq:consecutive}, the output distribution $\mathbf{p}_Z \in \mathbb{R}^{d(p-1)+1}$ is a convex combination of the rows of $\mathbf{P}_{Z|W}$, whose convex coefficients are given by a simple rearrangement of $\brc{\phi_{c,k}}_{c=1,k=0}^{d,p-1}$. This contributes $d(p - 1) + 1$ equations.
    \item By \cref{eq:phi-characterization-mixed-radix}, for each $c \in [d]$, the vector $\boldsymbol{\phi}_c = \brc{\phi_{c,k}}_{k=0}^{p-1}$ is a convex combination of the sender-wise categorical PMFs $\boldsymbol{\theta}_{i,c} = \brc{\theta_{i,c,k}}_{k=0}^{p-1}$. Hence, each $\boldsymbol{\phi}_c$ is a categorical PMF which sums to $1$. This contributes $d$ equations.
\end{itemize}

We finish with two pertinent remarks. Firstly, our randomized coding scheme has polynomial time complexity with respect to $n$, since most decoding steps are operations on matrices and vectors of size $\bigO{\poly(d, p)} = \bigO{1}$. (Refer to the proof of \cref{theorem:general-achievability} in \cref{section:general-achievability} for specific details.) Secondly, the high probability bound in \cref{eq:high-probability-bound} holds conditioned on \emph{any} message values. Hence, although $\Perror{n}$ is the average probability of error over all possible messages, a standard \emph{expurgation argument} similar to \cite[Section 7.7, p. 204]{CoverThomas2006} shows that the inner bound on $\Cregp$ remains the same under a maximal probability of error criterion.

\subsection{Converse Bounds} \label{subsection:converse-bounds}

Our fourth main result is an outer bound on the permutation capacity region of the general $p$-ary PAMAC:

\begin{theorem}[General Converse] \label{theorem:general-converse}
    The permutation capacity region of the $p$-ary PAMAC satisfies
    \begin{align}
        \Cregp \subseteq \brc{R \in \mathbb{R}_+^d: \sum_{i=1}^d R_i \leq \frac{d(p - 1)}{2}} \, .
    \end{align}
\end{theorem}

For convenience, we state the direct specialization of this result to the binary ($p = 2$) case:

\begin{corollary}[Binary Converse] \label{corollary:binary-converse}
    The permutation capacity region of the binary PAMAC satisfies
    \begin{align}
        \Cregp \subseteq \brc{R \in \mathbb{R}_+^d: \sum_{i=1}^d R_i \leq \frac{d}{2}} \, .
    \end{align}
\end{corollary}

We defer the technical details of our proof to \cref{section:converse-proof}. We adapt the proof technique in \cite[Section III-C]{Makur2020b} based on Fano's inequality \cite[Theorem 2.10.1]{CoverThomas2006} and the standard argument in \cite[Section 7.9]{CoverThomas2006}. Our derivations make use of the independence between messages and the fact that each message is uniformly distributed. By the Fisher-Neyman factorization theorem \cite[Theorem 3.6]{Keener2010}, $\hat{\mathbf{p}}_Y$ is a sufficient statistic of $Y_1^n$, since the probability mass function $p_{Y_1^n | M}(y_1^n | m)$ depends on $y_1^n$ through $\hat{\mathbf{p}}_Y$. Thus, it suffices to upper-bound $\I{W_1^n; \hat{\mathbf{p}}_Y}$. Upper-bounding the mutual information between the inputs and outputs of a noisy permutation channel has been studied in the prior literature, and we adopt the analysis in the proof of \cite[Theorem 2]{Makur2020b} based on analyzing the Shannon entropy of binomial random variables.

\subsection{Permutation Capacity Region}

Lastly, we combine the achievability and converse bounds discussed above to obtain explicit characterizations of the PAMAC's permutation capacity region. The following result for the $p$-ary PAMAC incorporates \cref{theorem:general-achievability,theorem:general-converse}:

\begin{theorem}[General Permutation Capacity Region] \label{theorem:general-permutation-capacity-region}
    The permutation capacity region of the $p$-ary PAMAC is
    \begin{align}
        \Cregp = \brc{R \in \mathbb{R}_+^d: \sum_{i=1}^d R_i \leq \frac{d(p - 1)}{2}} \, .
    \end{align}
\end{theorem}

Combining \cref{theorem:binary-achievability-time-sharing,corollary:binary-converse} gives rise to the permutation capacity region of the binary PAMAC:

\begin{corollary}[Binary Permutation Capacity Region] \label{corollary:binary-permutation-capacity-region}
    The permutation capacity region of the binary PAMAC is
    \begin{align}
        \Cregp = \brc{R \in \mathbb{R}_+^d: \sum_{i=1}^d R_i \leq \frac{d}{2}} \, .
    \end{align}
\end{corollary}

An immediate corollary of \cref{theorem:general-permutation-capacity-region} is the following characterization of the PAMAC's permutation sum-capacity:

\begin{corollary}[General Permutation Sum-Capacity] \label{corollary:general-permutation-sum-capacity}
    The permutation sum-capacity of the $p$-ary PAMAC is
    \begin{align}
        \Csum = \frac{d(p - 1)}{2} \, .
    \end{align}
\end{corollary}

\section{Proof of Binary Achievability Using Root Stability} \label{section:binary-achievability-root-stability}

In this section, we prove \cref{theorem:binary-achievability-root-stability}. Our argument makes use of two auxiliary results (\cref{lemma:lipschitz-continuity-of-roots,lemma:difference-of-quotients}), which we prove at the end of this section.

\begin{proof}[Proof of \cref{theorem:binary-achievability-root-stability}]
    Recall that we utilize the message sets, randomized encoders, and decoder outlined in \cref{subsection:binary-achievability-root-stability} with rate $d$-tuple $R = \iprn{\frac{1}{2} - \alpha, \dots, \frac{1}{2} - \alpha}$ for some arbitrary, fixed $\alpha > 0$. For notational simplicity, let
    \begin{align}
        \gamma = (2d + 1)^{d+1} d^2 \prn{2 + (2d + 1)^d}^{2d + 2}
    \end{align}
    throughout this proof. By definition of achievable rate tuples, we want to show that
    \begin{align}
        \forall \epsilon > 0, \, \exists n_0 \in \mathbb{N}, \, \forall n \geq n_0, \, \Perror{n} \leq \epsilon \, .
    \end{align}
    
    Fix any $\epsilon > 0$. Since logarithms are asymptotically dominated by polynomials, there exists an $n_1 \in \mathbb{N}$ such that
    \begin{align}
        \sqrt{\frac{\log_e n}{n}} \leq \prn{4 \gamma \norm{\mathbf{P}_{Z|W}^{-1}}_1}^{-1} n^{-\frac{1}{2} + \alpha}
    \end{align}
    for all $n \geq n_1$. (Recall that $d$ is a constant in our formal model.) Choose
    \begin{align}
        n_0 = \max \brc{\prn{6 \gamma \norm{\mathbf{P}_{Z|W}^{-1}}_1}^4, n_1, \sqrt{\frac{2(d + 1)}{\epsilon}}} \, ,
    \end{align}
    and fix any $n \geq n_0$.
    
    \emph{Step 0: Upper-bounding $\norm{\mathbf{p}_W - \tilde{\mathbf{p}}_W}_\infty$ with high probability.} In this proof, all probabilities $\P{\cdot}$ are conditioned on sending the fixed message values $M = (\theta_1, \dots, \theta_d)$. By definition of $\tilde{\mathbf{p}}_W$,
    \begin{align}
        \norm{\mathbf{p}_W - \tilde{\mathbf{p}}_W}_\infty \leq \norm{\mathbf{P}_{Z|W}^{-1}}_1 \norm{\mathbf{p}_Z - \hat{\mathbf{p}}_Z}_\infty \, ,
    \end{align}
    where $\mathbf{p}_W$ and $\mathbf{p}_Z$ are the true marginal probability distributions of $W$ and $Z$, respectively. Since $\mathbf{p}_Z, \hat{\mathbf{p}}_Z \in \mathcal{S}_d$,
    \begin{align}
        &\equad \P{\norm{\mathbf{p}_W - \tilde{\mathbf{p}}_W}_\infty \leq \norm{\mathbf{P}_{Z|W}^{-1}}_1 \sqrt{\frac{\log_e n}{n}}} \\
        &\geq \P{\norm{\mathbf{p}_Z - \hat{\mathbf{p}}_Z}_\infty \leq \sqrt{\frac{\log_e n}{n}}} \\
        &\stackclap{(a)}{\geq} 1 - \sum_{t=0}^d \P{\abs{\bkt{\mathbf{p}_Z}_t - \bkt{\hat{\mathbf{p}}_Z}_t} > \sqrt{\frac{\log_e n}{n}}} \\
        &\stackclap{(b)}{\geq} 1 - \frac{2(d + 1)}{n^2} \, ,
    \end{align}
    where (a) follows from the union bound. The $Z_1^n$ are conditionally independent given the messages, since the letters $\prn{X_i}_1^n$ are independently generated and the $W_1^n$ are independently passed through the DMC. Hence, (b) follows by applying Hoeffding's inequality (\cref{lemma:hoeffdings-inequality}) with $\tau = \sqrt{\frac{\log_e n}{n}}$.
    
    For Steps 1 to 3 of this proof, we restrict to the subset of the sample space where
    \begin{align}
        \norm{\mathbf{p}_W - \tilde{\mathbf{p}}_W}_\infty \leq \norm{\mathbf{P}_{Z|W}^{-1}}_1 \sqrt{\frac{\log_e n}{n}} \, . \label{eq:sample-space-restriction}
    \end{align}
    
    \emph{Step 1: Upper-bounding stability of roots in unspecified order.} Let $G_W: \mathbb{C} \rightarrow \mathbb{C}$ be the true probability generating function of $W$, given by
    \begin{align}
        G_W(\xi) = \E{\xi^W} = \sum_{t=0}^d \bkt{\mathbf{p}_W}_t \xi^t \, .
    \end{align}
    Let $\brc{\xi_i}_{i=1}^d$ be the roots of $G_W$. In this step, we will show that
    \begin{align} \label{eq:step-1-result}
        \min_{\pi \in \Sg_d} \max_{i \in [d]} \abs{\Re \brc{\tilde{\xi}_{\pi(i)}} - \xi_i} \leq \frac{2 \gamma}{2d + 1} \norm{\mathbf{p}_W - \tilde{\mathbf{p}}_W}_\infty \, .
    \end{align}
    
    By definition of the adder and encoder, $W = \sum_{i=1}^d X_i$ and $X_i \sim \Bernoulli{\theta_i}$. Since the $\brc{X_i}_{i=1}^d$ are independent,
    \begin{align}
        G_W(\xi) = \prod_{i=1}^d G_{X_i}(\xi) = \prod_{i=1}^d (1 - \theta_i + \theta_i \xi) \, .
    \end{align}
    Since $G_W$ is a degree-$d$ polynomial, $G_W(\xi) \propto \prod_{i=1}^d (\xi - \xi_i)$ by the fundamental theorem of algebra. Hence,
    \begin{align} \label{eq:root-to-parameter-transform}
        \xi_i = \frac{\theta_i - 1}{\theta_i} \enspace \text{and therefore} \enspace \theta_i = \frac{1}{1 - \xi_i} \, .
    \end{align}
    
    Define two monic complex-valued polynomials $F$ and $\tilde{F}$ with the same roots as $G_W$ and $\tilde{G}_W$, respectively:
    \begin{align}
        F(\xi) &= \frac{G_W(\xi)}{\bkt{\mathbf{p}_W}_d} = \xi^d + \sum_{t=0}^{d-1} \frac{\bkt{\mathbf{p}_W}_t}{\bkt{\mathbf{p}_W}_d} \, \xi^t \, , \\
        \tilde{F}(\xi) &= \frac{\tilde{G}_W(\xi)}{\bkt{\tilde{\mathbf{p}}_W}_d} = \xi^d + \sum_{t=0}^{d-1} \frac{\bkt{\tilde{\mathbf{p}}_W}_t}{\bkt{\tilde{\mathbf{p}}_W}_d} \, \xi^t \, .
    \end{align}
    Observe that
    \begin{align}
        \bkt{\mathbf{p}_W}_d \stackclap{(a)}{=} \prod_{i=1}^d \theta_i \stackclap{(b)}{\geq} \frac{1}{(2d + 1)^d} \, , \label{eq:probability-w-d-lower-bound}
    \end{align}
    where (a) follows from the definition of $W$ and independence of $\brc{X_i}_{i=1}^d$, and (b) holds because the message sets are padded away from $0$ and so each $\theta_i \geq \frac{1}{2d + 1}$ by \cref{eq:root-stability-message-sets}. Thus, $\tilde{F}$ is well-defined (namely, $\bkt{\tilde{\mathbf{p}}_W}_d \neq 0$) because
    \begin{align}
        \bkt{\tilde{\mathbf{p}}_W}_d &\geq \bkt{\mathbf{p}_W}_d - \abs{\bkt{\tilde{\mathbf{p}}_W}_d - \bkt{\mathbf{p}_W}_d} \\
        &\stackclap{(a)}{\geq} \bkt{\mathbf{p}_W}_d - \norm{\mathbf{P}_{Z|W}^{-1}}_1 \sqrt{\frac{\log_e n}{n}} \\
        &\geq \bkt{\mathbf{p}_W}_d - \norm{\mathbf{P}_{Z|W}^{-1}}_1 \frac{1}{\sqrt[4]{n}} \\
        &\stackclap{(b)}{\geq} \bkt{\mathbf{p}_W}_d - \frac{1}{6 \gamma} \\
        &\stackclap{(c)}{\geq} \bkt{\mathbf{p}_W}_d - \frac{1}{2(2d + 1)^d} \\
        &\stackclap{(d)}{\geq} \frac{\bkt{\mathbf{p}_W}_d}{2} \\
        &> 0 \, ,
    \end{align}
    where (a) follows from the sample space restriction \cref{eq:sample-space-restriction}, (b) follows from the choice of $n_0$, (c) follows from the definition of $\gamma$, and (d) follows from \cref{eq:probability-w-d-lower-bound}.
    
    Since the messages $\brc{\theta_i}_{i=1}^d$ are distinct and in $(0, 1)$, the roots $\brc{\xi_i}_{i=1}^d$ are distinct and negative. Since $\theta(\xi) = \frac{1}{1 - \xi}$ is $1$-Lipschitz for $\xi < 0$,
    \begin{align} \label{eq:separation-of-roots}
        \min_{i \neq j} \abs{\xi_i - \xi_j} \geq \min_{i \neq j} \abs{\theta_i - \theta_j} \stackclap{(a)}{\geq} \frac{1}{2d + 1} \, ,
    \end{align}
    where (a) holds due to the padding between message sets. Invoking \cref{lemma:lipschitz-continuity-of-roots} with $\delta = \frac{1}{2d + 1}$, we have
    \begin{align}
        &\equad \min_{\pi \in \Sg_d} \max_{i \in [d]} \abs{\Re \brc{\tilde{\xi}_{\pi(i)}} - \xi_i} \\
        &\leq (2d - 1) d^2 (2d + 1)^{d-1} \prn{2 + \max_{t \in \bktz{d-1}} \frac{\bkt{\mathbf{p}_W}_t}{\bkt{\mathbf{p}_W}_d}}^{2d} \\
        &\qquad \cdot \max_{t \in \bktz{d-1}} \abs{\frac{\bkt{\mathbf{p}_W}_t}{\bkt{\mathbf{p}_W}_d} - \frac{\bkt{\tilde{\mathbf{p}}_W}_t}{\bkt{\tilde{\mathbf{p}}_W}_d}} \\
        &\leq (2d + 1)^d d^2 \prn{2 + \frac{1}{\bkt{\mathbf{p}_W}_d}}^{2d} \underbrace{\max_{t \in \bktz{d-1}} \abs{\frac{\bkt{\mathbf{p}_W}_t}{\bkt{\mathbf{p}_W}_d} - \frac{\bkt{\tilde{\mathbf{p}}_W}_t}{\bkt{\tilde{\mathbf{p}}_W}_d}}}_{\circled{1}} \, .
    \end{align}
    Next, we upper-bound $\circled{1}$:
    \begin{align}
        \circled{1} &\stackclap{(a)}{\leq} \max_{t \in \bktz{d-1}} \brc{\frac{\abs{\bkt{\mathbf{p}_W}_t - \bkt{\tilde{\mathbf{p}}_W}_t}}{\bkt{\tilde{\mathbf{p}}_W}_d} + \frac{\bkt{\mathbf{p}_W}_t \abs{\bkt{\mathbf{p}_W}_d - \bkt{\tilde{\mathbf{p}}_W}_d}}{\bkt{\mathbf{p}_W}_d \bkt{\tilde{\mathbf{p}}_W}_d}} \\
        &\leq \frac{\norm{\mathbf{p}_W - \tilde{\mathbf{p}}_W}_\infty}{\bkt{\tilde{\mathbf{p}}_W}_d} + \frac{\norm{\mathbf{p}_W - \tilde{\mathbf{p}}_W}_\infty}{\bkt{\mathbf{p}_W}_d \bkt{\tilde{\mathbf{p}}_W}_d} \\
        &\stackclap{(b)}{\leq} \prn{1 + \frac{1}{\bkt{\mathbf{p}_W}_d}} \frac{2}{\bkt{\mathbf{p}_W}_d} \norm{\mathbf{p}_W - \tilde{\mathbf{p}}_W}_\infty \\
        &\leq 2 \prn{2 + \frac{1}{\bkt{\mathbf{p}_W}_d}}^2 \norm{\mathbf{p}_W - \tilde{\mathbf{p}}_W}_\infty \, ,
    \end{align}
    where (a) follows from \cref{lemma:difference-of-quotients} and (b) follows from the fact that $\bkt{\tilde{\mathbf{p}}_W}_d \geq \frac{\bkt{\mathbf{p}_W}_d}{2}$. Therefore,
    \begin{align}
        \min_{\pi \in \Sg_d} \max_{i \in [d]} \abs{\Re \brc{\tilde{\xi}_{\pi(i)}} - \xi_i} \leq \frac{2 \gamma}{2d + 1} \norm{\mathbf{p}_W - \tilde{\mathbf{p}}_W}_\infty
    \end{align}
    as desired, using the definition of $\gamma$ and \cref{eq:probability-w-d-lower-bound}.
    
    \emph{Step 2: Upper-bounding stability of roots in sorted order.} By definition of the encoders, $\theta_1 < \cdots < \theta_d$ and therefore $\xi_1 < \cdots < \xi_d$. Consider an ascending order on $\ibrc{\Re \ibrc{\tilde{\xi}_i}}_{i=1}^d$ so that $\Re \ibrc{\tilde{\xi}_1} \leq \cdots \leq \Re \ibrc{\tilde{\xi}_d}$. In this step, we will strengthen the result from Step 1 to
    \begin{align} \label{eq:step-2-result}
        \max_{i \in [d]} \abs{\Re \brc{\tilde{\xi}_i} - \xi_i} \leq \frac{2 \gamma}{2d + 1} \norm{\mathbf{p}_W - \tilde{\mathbf{p}}_W}_\infty \, .
    \end{align}
    
    Let $\sigma$ be a permutation achieving the minimum in \cref{eq:step-1-result}:
    \begin{align}
        \sigma \in \arg \min_{\pi \in \Sg_d} \max_{i \in [d]} \abs{\Re \brc{\tilde{\xi}_{\pi(i)}} - \xi_i} \, .
    \end{align}
    We have
    \begin{align}
        &\equad \max_{i \in [d]} \abs{\Re \brc{\tilde{\xi}_{\sigma(i)}} - \xi_i} \\
        &\stackclap{(a)}{\leq} \frac{2 \gamma}{2d + 1} \norm{\mathbf{p}_W - \tilde{\mathbf{p}}_W}_\infty \\
        &\stackclap{(b)}{\leq} \frac{2 \gamma}{2d + 1} \norm{\mathbf{P}_{Z|W}^{-1}}_1 \sqrt{\frac{\log_e n}{n}} \\
        &\leq \frac{2 \gamma}{2d + 1} \norm{\mathbf{P}_{Z|W}^{-1}}_1 \frac{1}{\sqrt[4]{n}} \\
        &\stackclap{(c)}{\leq} \frac{2 \gamma}{2d + 1} \norm{\mathbf{P}_{Z|W}^{-1}}_1 \prn{6 \gamma \norm{\mathbf{P}_{Z|W}^{-1}}_1}^{-1} \\
        &= \frac{1}{3(2d + 1)} \\
        &< \frac{\delta}{2} \, , \label{eq:root-error-ub-delta}
    \end{align}
    where (a) follows from Step 1, (b) follows from the sample space restriction \cref{eq:sample-space-restriction}, and (c) follows from the choice of $n_0$. Therefore, for any $i \in [d - 1]$,
    \begin{align}
        \Re \brc{\tilde{\xi}_{\sigma(i)}} &\stackclap{(a)}{<} \xi_i + \frac{\delta}{2} \\
        &\stackclap{(b)}{\leq} \xi_{i+1} - \frac{\delta}{2} \\
        &\stackclap{(c)}{<} \Re \brc{\tilde{\xi}_{\sigma(i+1)}} \, ,
    \end{align}
    where (a) and (c) follow from \cref{eq:root-error-ub-delta} and (b) follows from \cref{eq:separation-of-roots}. Thus, $\sigma$ is the identity permutation as desired.
    
    \emph{Step 3: Upper-bounding error in predicted parameters.} In this step, we will show that for every $i \in [d]$,
    \begin{align}
        \abs{\tilde{\theta}_i - \theta_i} < \frac{1}{2(2d + 1)(\abs{\mathcal{M}_i} - 1)} \, .
    \end{align}
    For each $i \in [d]$, we have
    \begin{align}
        \Re \brc{\tilde{\xi}_i} &\leq \xi_i + \abs{\Re \brc{\tilde{\xi}_i} - \xi_i} \\
        &\stackclap{(a)}{=} \prn{1 - \frac{1}{\theta_i}} + \abs{\Re \brc{\tilde{\xi}_i} - \xi_i} \\
        &\stackclap{(b)}{\leq} -\frac{1}{2d} + \max_{j \in [d]} \abs{\Re \brc{\tilde{\xi}_j} - \xi_j} \\
        &\stackclap{(c)}{<} -\frac{1}{2d} + \frac{\delta}{2} \\
        &\stackclap{(d)}{<} 0 \, ,
    \end{align}
    where (a) follows from \cref{eq:root-to-parameter-transform}, (b) holds due to the message sets being padded away from $1$, (c) follows from \cref{eq:root-error-ub-delta}, and (d) follows from the fact that $\delta = \frac{1}{2d + 1}$. 
    
    The ascending order on $\ibrc{\Re \ibrc{\tilde{\xi}_i}}_{i=1}^d$ coincides with the sorted ordering $\tilde{\theta}_1 \leq \cdots \leq \tilde{\theta}_d$ of the estimated Bernoulli parameters in the decoder, because the mapping \cref{eq:estimated-root-to-parameter-transform} is monotonically increasing for $\Re \ibrc{\tilde{\xi}_i} < 0$. Thus,
    \begin{align}
        \abs{\tilde{\theta}_i - \theta_i} &\stackclap{(a)}{\leq} \abs{\Re \brc{\tilde{\xi}_i} - \xi_i} \\
        &\stackclap{(b)}{\leq} \frac{2 \gamma}{2d + 1} \norm{\mathbf{p}_W - \tilde{\mathbf{p}}_W}_\infty \\
        &\stackclap{(c)}{\leq} \frac{2 \gamma}{2d + 1} \norm{\mathbf{P}_{Z|W}^{-1}}_1 \sqrt{\frac{\log_e n}{n}} \\
        &\stackclap{(d)}{\leq} \frac{n^{-\frac{1}{2} + \alpha}}{2(2d + 1)} \\
        &\stackclap{(e)}{<} \frac{1}{2(2d + 1)(\abs{\mathcal{M}_i} - 1)} \, ,
    \end{align}
    where (a) follows from the fact that $u \mapsto \frac{1}{1 - u}$ is $1$-Lipschitz for $u < 0$, (b) follows from \cref{eq:step-2-result}, (c) follows from the sample space restriction \cref{eq:sample-space-restriction}, (d) follows because $n_0 \geq n_1$, and (e) follows because $\abs{\mathcal{M}_i} = n^{\frac{1}{2} - \alpha}$.
    
    \emph{Step 4: Upper-bounding probability of decoding error.} By definition of the message sets and decoder, $\hat{\theta}_i \neq \theta_i$ only if $\iabs{\tilde{\theta}_i - \theta_i} \geq \frac{1}{2(2d + 1)(\abs{\mathcal{M}_i} - 1)}$. Therefore,
    \begin{align}
        &\equad \P{\exists i \in [d], \, \hat{\theta}_i \neq \theta_i} \\
        &\leq \P{\exists i \in [d], \, \abs{\tilde{\theta}_i - \theta_i} \geq \frac{1}{2(2d + 1)(\abs{\mathcal{M}_i} - 1)}} \\
        &\stackclap{(a)}{\leq} \frac{2(d + 1)}{n^2} \\
        &\stackclap{(b)}{\leq} \epsilon \, ,
    \end{align}
    where (a) follows from the results of Steps 0 and 3 and (b) holds because $n \geq n_0 \geq \sqrt{\frac{2(d + 1)}{\epsilon}}$. Finally, taking expectation with respect to the messages yields $\Perror{n} \leq \epsilon$ as desired. 
\end{proof}

Below, we present the technical lemmas used in the proof of \cref{theorem:binary-achievability-root-stability}. The first lemma provides an upper bound on the root stability of a monic polynomial:

\begin{lemma}[Lipschitz Continuity of Roots] \label{lemma:lipschitz-continuity-of-roots}
    Let $f, g: \mathbb{C} \rightarrow \mathbb{C}$ be monic polynomials of degree $d$ with real coefficients. Write:
    \begin{align}
        f(x) &= x^d + \sum_{t=0}^{d-1} a_t x^t = \prod_{i=1}^d (x - \lambda_i) \, , \\
        g(x) &= x^d + \sum_{t=0}^{d-1} b_t x^t = \prod_{i=1}^d (x - \mu_i) \, .
    \end{align}
    Assume $f$ has all distinct, real, negative roots (and thus all positive coefficients). Assume $\min_{i \neq j} \abs{\lambda_i - \lambda_j} \geq \delta$ for some constant $\delta > 0$. Then
    \begin{align}
        &\equad \min_{\pi \in \Sg_d} \max_{i \in [d]} \abs{\Re \brc{\mu_{\pi(i)}} - \lambda_i} \\
        &\leq \frac{(2d - 1) d^2}{\delta^{d-1}} \prn{2 + \max_{t \in \bktz{d-1}} a_t}^{2d} \max_{t \in \bktz{d-1}} \abs{a_t - b_t} \, .
    \end{align}
\end{lemma}

\begin{proof}
    Let $\mathbf{C}(f) \in \R^{d \times d}$ and $\mathbf{C}(g) \in \R^{d \times d}$ be the Frobenius companion matrices of $f$ and $g$, respectively. For example,
    \begin{align}
        \mathbf{C}(f) = \begin{bmatrix}
            0 & 0 & \cdots & 0 & -a_0 \\
            1 & 0 & \cdots & 0 & -a_1 \\
            0 & 1 & \cdots & 0 & -a_2 \\
            \vdots & \vdots & \ddots & \vdots & \vdots \\
            0 & 0 & \cdots & 1 & -a_{d-1}
        \end{bmatrix} .
    \end{align}
    The eigenvalues of $\mathbf{C}(f)$ and $\mathbf{C}(g)$ are $\brc{\lambda_i}_{i=1}^d$ and $\brc{\mu_i}_{i=1}^d$, respectively. Let $\mathbf{V}  \in \R^{d \times d}$ be the Vandermonde matrix with parameters $(\lambda_1, \dots, \lambda_d)$:
    \begin{align}
        \mathbf{V} = \begin{bmatrix}
            1 & \lambda_1 & \lambda_1^2 & \cdots & \lambda_1^{d-1} \\
            1 & \lambda_2 & \lambda_2^2 & \cdots & \lambda_2^{d-1} \\
            1 & \lambda_3 & \lambda_3^2 & \cdots & \lambda_3^{d-1} \\
            \vdots & \vdots & \vdots & \ddots & \vdots \\
            1 & \lambda_d & \lambda_d^2 & \cdots & \lambda_d^{d-1}
        \end{bmatrix} .
    \end{align}
    Let $\mathbf{D} = \mathbf{diag}(\lambda_1, \dots, \lambda_d) \in \R^{d \times d}$. We have the eigendecomposition $\mathbf{C}(f) = \mathbf{V}^{-1} \mathbf{D} \mathbf{V}$:
    \begin{align}
        \mathbf{V} \cdot \mathbf{C}(f) &= \begin{bmatrix}
            \lambda_1 & \lambda_1^2 & \cdots & \lambda_1^{d-1} & -\sum_{t=0}^{d-1} a_t \lambda_1^t \\
            \lambda_2 & \lambda_2^2 & \cdots & \lambda_2^{d-1} & -\sum_{t=0}^{d-1} a_t \lambda_2^t \\
            \vdots & \vdots & \ddots & \vdots & \vdots \\
            \lambda_d & \lambda_d^2 & \cdots & \lambda_d^{d-1} & -\sum_{t=0}^{d-1} a_t \lambda_d^t
        \end{bmatrix} \\
        &\stackclap{(a)}{=} \begin{bmatrix}
            \lambda_1 & \lambda_1^2 & \cdots & \lambda_1^{d-1} & \lambda_1^d - f(\lambda_1) \\
            \lambda_2 & \lambda_2^2 & \cdots & \lambda_2^{d-1} & \lambda_2^d - f(\lambda_2) \\
            \vdots & \vdots & \ddots & \vdots & \vdots \\
            \lambda_d & \lambda_d^2 & \cdots & \lambda_d^{d-1} & \lambda_d^d - f(\lambda_d)
        \end{bmatrix} \\
        &\stackclap{(b)}{=} \begin{bmatrix}
            \lambda_1 & \lambda_1^2 & \cdots & \lambda_1^{d-1} & \lambda_1^d \\
            \lambda_2 & \lambda_2^2 & \cdots & \lambda_2^{d-1} & \lambda_2^d \\
            \vdots & \vdots & \ddots & \vdots & \vdots \\
            \lambda_d & \lambda_d^2 & \cdots & \lambda_d^{d-1} & \lambda_d^d \\
        \end{bmatrix} \\
        &= \mathbf{D} \mathbf{V} \, ,
    \end{align}
    where (a) follows from rearranging $f$ and (b) follows because $\brc{\lambda_i}_{i=1}^d$ are roots of $f$.
    Our proof uses three standard results from the matrix analysis literature, which we restate below for convenience.
    
    \begin{lemma}[Bauer-Fike Theorem {\cite[Theorem 3.3, Chapter IV]{StewartSun1990}}] \label{lemma:bauer-fike}
        Let $\mathbf{A} \in \mathbb{R}^{d \times d}$ have eigenvalues $\brc{\lambda_i}_{i=1}^d$ and eigendecomposition $\mathbf{A} = \mathbf{X} \boldsymbol{\Lambda} \mathbf{X}^{-1}$. Let $\mathbf{B} \in \mathbb{R}^{d \times d}$ have eigenvalues $\brc{\mu_i}_{i=1}^d$. Define the matching distance \cite[Definition 1.2, Chapter IV]{StewartSun1990} between the eigenvalues of $\mathbf{A}$ and $\mathbf{B}$ as
        \begin{align}
            \mathrm{md}(\mathbf{A}, \mathbf{B}) = \min_{\pi \in \Sg_d} \max_{i \in [d]} \abs{\mu_{\pi(i)} - \lambda_i} \, .
        \end{align}
        Then the matching distance satisfies the upper bound
        \begin{align}
            \mathrm{md}(\mathbf{A}, \mathbf{B}) \leq (2d - 1) \norm{\mathbf{X}}_1 \norm{\mathbf{X}^{-1}}_1 \norm{\mathbf{B} - \mathbf{A}}_1 \, .
        \end{align}
    \end{lemma}
    
    \begin{lemma}[Inverse Vandermonde Norm {\cite[Theorem 2.1]{Gautschi1990}}] \label{lemma:vandermonde-l1-norm}
        Let $\brc{x_i}_{i=1}^d$ be distinct real numbers. Let $\mathbf{V} \in \mathbb{R}^{d \times d}$ be the Vandermonde matrix with parameters $(x_1, \dots, x_d)$. Then
        \begin{align}
            \norm{\mathbf{V}^{-1}}_1 \leq \max_{i \in [d]} \prod_{j \neq i} \frac{1 + \abs{x_j}}{\abs{x_i - x_j}} \, .
        \end{align}
    \end{lemma}
    
    \begin{lemma}[Gershgorin Circle Theorem {\cite[Theorem 6.1.1]{HornJohnson2013}}] \label{lemma:gershgorin-circle-theorem}
        For any $\mathbf{A} \in \mathbb{R}^{d \times d}$, the eigenvalues of $\mathbf{A}$ are in the union of Gershgorin discs
        \begin{align}
            \bigcup_{i=1}^d \brc{z \in \mathbb{C}: \abs{z - \bkt{\mathbf{A}}_{i,i}} \leq \sum_{j \neq i} \abs{\bkt{\mathbf{A}}_{i,j}}} \, .
        \end{align}
    \end{lemma}
    
    Invoking \cref{lemma:bauer-fike} on $\mathbf{C}(f)$ and $\mathbf{C}(g)$, we have
    \begin{align}
        &\equad \min_{\pi \in \Sg_d} \max_{i \in [d]} \abs{\mu_{\pi(i)} - \lambda_i} \\
        &\leq (2d - 1) \norm{\mathbf{V}^{-1}}_1 \norm{\mathbf{V}}_1 \norm{\mathbf{C}(g) - \mathbf{C}(f)}_1 \, .
    \end{align}
    Next, we upper-bound $\norm{\mathbf{V}^{-1}}_1$. The Gershgorin discs $\brc{D_i}_{i=1}^d$ of $\mathbf{C}(f)$ are
    \begin{align}
        D_i &= \brc{z \in \mathbb{C}: \abs{z - \bkt{\mathbf{C}(f)}_{i,i}} \leq \sum_{j \neq i} \abs{\bkt{\mathbf{C}(f)}_{i,j}}} \\
        &= \begin{cases}
            \brc{z \in \mathbb{C}: \abs{z} \leq \abs{a_0}} \, , & \text{if $i = 1$} \\
            \brc{z \in \mathbb{C}: \abs{z} \leq 1 + \abs{a_{i-1}}} \, , & \text{if $i \in \bktz{2, d - 1}$} \\
            \brc{z \in \mathbb{C}: \abs{z + a_{d-1}} \leq 1} \, , & \text{if $i = d$}
        \end{cases} \\
        &\stackclap{(a)}{\subseteq} \brc{z \in \mathbb{C}: \abs{z} \leq 1 + \abs{a_{i-1}}} \, , \label{eq:gershgorin-discs}
    \end{align}
    where (a) follows in the $i = d$ case by the triangle inequality:
    \begin{align}
        \abs{z} &= \abs{z + a_{d-1} - a_{d-1}} \\
        &\leq \abs{z + a_{d-1}} + \abs{a_{d-1}} \leq 1 + \abs{a_{d-1}} \, .
    \end{align}
    Therefore,
    \begin{align}
        \norm{\mathbf{V}^{-1}}_1 &\stackclap{(a)}{\leq} \max_{i \in [d]} \prod_{j \neq i} \frac{1 + \abs{\lambda_j}}{\abs{\lambda_i - \lambda_j}} \\
        &\stackclap{(b)}{\leq} \max_{i \in [d]} \prod_{j \neq i} \frac{2 + \max_{t \in \bktz{d-1}} \abs{a_t}}{\abs{\lambda_i - \lambda_j}} \\
        &\stackclap{(c)}{\leq} \max_{i \in [d]} \prod_{j \neq i} \frac{2 + \max_{t \in \bktz{d-1}} a_t}{\delta} \\
        &= \frac{\prn{2 + \max_{t \in \bktz{d-1}} a_t}^{d-1}}{\delta^{d-1}} \, ,
    \end{align}
    where (a) follows from \cref{lemma:vandermonde-l1-norm}, (b) follows from \cref{lemma:gershgorin-circle-theorem,eq:gershgorin-discs}, and (c) follows from the assumption that $f$ has all positive coefficients and roots separated by at least $\delta$.

    Next, we upper-bound $\norm{\mathbf{V}}_1$:
    \begin{align}
        \norm{\mathbf{V}}_1 &= \max_{j \in [d]} \sum_{i=1}^d \abs{\lambda_i}^{j-1} \\
        &\stackclap{(a)}{=} \max \brc{d, \sum_{i=1}^d \abs{\lambda_i}^{d-1}} \\
        &\stackclap{(b)}{\leq} \max \brc{d, \sum_{i=1}^d \prn{1 + \max_{t \in \bktz{d-1}} \abs{a_t}}^{d-1}} \\
        &= d \prn{1 + \max_{t \in \bktz{d-1}} \abs{a_t}}^{d-1} \\
        &\stackclap{(c)}{=} d \prn{1 + \max_{t \in \bktz{d-1}} a_t}^{d-1},
    \end{align}
    where (a) follows because a sum of exponentials is convex and the maximum of a convex function is attained at an endpoint of its domain, (b) follows from \cref{lemma:gershgorin-circle-theorem,eq:gershgorin-discs}, and (c) follows from the assumption that $f$ has all positive coefficients.
    
    By inspection, $\norm{\mathbf{C}(g) - \mathbf{C}(f)}_1 \leq d \max_{t \in \bktz{d-1}} \abs{a_t - b_t}$. Therefore,
    \begin{align}
        &\equad \min_{\pi \in \Sg_d} \max_{i \in [d]} \abs{\Re \brc{\mu_{\pi(i)}} - \lambda_i} \\
        &\stackclap{(a)}{\leq} \min_{\pi \in \Sg_d} \max_{i \in [d]} \abs{\mu_{\pi(i)} - \lambda_i} \\
        &\leq \frac{(2d - 1) d^2}{\delta^{d-1}} \prn{2 + \max_{t \in \bktz{d-1}} a_t}^{2d} \max_{t \in \bktz{d-1}} \abs{a_t - b_t}
    \end{align}
    as desired, where (a) follows because $\brc{\lambda_i}_{i=1}^d \subset \mathbb{R}$.
\end{proof}

The second lemma is an upper bound on the absolute difference between two quotients:

\begin{lemma}[Difference of Quotients] \label{lemma:difference-of-quotients}
    For all $x_1, x_2 \in \mathbb{R}$ and $y_1, y_2 \neq 0$,
    \begin{align}
        \abs{\frac{x_1}{y_1} - \frac{x_2}{y_2}} \leq \frac{\abs{x_1 - x_2}}{\abs{y_2}} + \frac{\abs{x_1} \abs{y_1 - y_2}}{\abs{y_1} \abs{y_2}} \, .
    \end{align}
\end{lemma}

\begin{proof}
    Observe via the triangle inequality that
    \begin{align}
        \abs{\frac{x_1}{y_1} - \frac{x_2}{y_2}} &= \abs{\frac{x_1}{y_2} - \frac{x_2}{y_2} + \frac{x_1}{y_1} - \frac{x_1}{y_2}} \\
        &\leq \abs{\frac{x_1}{y_2} - \frac{x_2}{y_2}} + \abs{\frac{x_1}{y_1} - \frac{x_1}{y_2}} \\
        &= \frac{\abs{x_1 - x_2}}{\abs{y_2}} + \abs{x_1} \abs{\frac{1}{y_1} - \frac{1}{y_2}} \\
        &= \frac{\abs{x_1 - x_2}}{\abs{y_2}} + \frac{\abs{x_1} \abs{y_1 - y_2}}{\abs{y_1} \abs{y_2}} \, .
    \end{align}
\end{proof}

\section{Proof of General Achievability Using Time Sharing} \label{section:general-achievability}

In this section, we prove \cref{theorem:general-achievability}. Our argument makes use of three auxiliary results (\cref{lemma:y-marginal,lemma:c-tilde-inverse,lemma:minimum-singular-value}), which we prove at the end of this section and in \cref{appendix:auxiliary-results}. We also provide alternative justifications (\cref{lemma:cardinality,lemma:subsegment-well-definedness}) in \cref{appendix:auxiliary-results} for some parts of our proof.

\begin{proof}[Proof of \cref{theorem:general-achievability}]
    By definition of permutation capacity region, it suffices to show that for all $\alpha > 0$, any rate $d$-tuple satisfying
    \begin{align}
        \sum_{i=1}^d R_i = \frac{d(p - 1)}{2} - \alpha \enspace \text{and} \enspace \forall i \in [d], \, R_i > 0
    \end{align}
    is achievable. Fix $\alpha > 0$ and $(R_1, \dots, R_d) \in \mathbb{R}_+^d$ satisfying the above. By definition of achievable rate tuples, we want to show that
    \begin{align}
        \forall \epsilon > 0, \, \exists n_0 \in \mathbb{N}, \, \forall n \geq n_0, \, \Perror{n} \leq \epsilon \, .
    \end{align}
    Fix $\epsilon > 0$. Choose
    \begin{align}
        n_0 = \max \Biggl\{ &(p - 1)^{\frac{2d^2 (p - 1)}{\alpha}}, p^{\frac{d(p - 1)^2}{R_d}}, 2^{\frac{d(p - 1)}{\min_{i \in [d]} R_i}}, \\
        &n_1, \sqrt{\frac{2(d(p - 1) + 1)}{\epsilon}} \Biggr\} \, ,
    \end{align}
    where $n_1 \in \mathbb{N}$ is sufficiently large such that
    \begin{align}
        \forall n \geq n_1, \, \frac{\sigma_{\min}^2 \prn{\mathbf{P}_{Z|W}}}{2 \sqrt{2} \prn{dp + 1}^{\frac{7}{2}} d} \, n^{-\frac{1}{2} + \frac{\alpha}{2d(p - 1)}} \geq \sqrt{\frac{\log_e n}{n}} \, .
    \end{align}
    Such an $n_1$ exists because
    \begin{align}
        \lim_{n \rightarrow \infty} \frac{n^{-\frac{1}{2} + \frac{\alpha}{2d(p - 1)}}}{\sqrt{\frac{\log_e n}{n}}} = \lim_{n \rightarrow \infty} \frac{n^{\frac{\alpha}{2d(p - 1)}}}{\sqrt{\log_e n}} = \infty \, .
    \end{align}
    Fix $n \geq n_0$. For notational simplicity, let
    \begin{align}
        m_i = n^{\frac{1}{d(p - 1)} \prn{R_i + \frac{\alpha}{2d}}}
    \end{align}
    for each $i \in [d]$, and assume without loss of generality that $m_i \in \mathbb{N}$. Consider the following message sets, encoders, and decoder.
    
    \textbf{Message sets.} Without loss of generality, let each message in $\mathcal{M}_i$ be a $d$-tuple of points, where each point lies on a lattice $\mathcal{L}_i$ embedded in the $(p - 1)$-dimensional probability simplex, as visualized in \cref{figure:message-set}. (Contrary to the situation in the proof of \cref{theorem:binary-achievability-root-stability}, there is no need to specifically incorporate padding around the boundary of the simplex.) Formally,
    \begin{align}
        \forall i \in [d], \enspace \mathcal{M}_i &= \mathcal{L}_i^d \, , \\
        \forall i \in [d], \enspace \mathcal{L}_i &= \brc{\brc{\theta_k}_{k=0}^{p-1} \in \Theta_i^p: \sum_{k=0}^{p-1} \theta_k = 1} \, , \\
        \forall i \in [d - 1], \enspace \Theta_i &= \brc{\frac{\ell_i}{m_i - 1}: \ell_i \in \bktz{m_i - 1}} \, , \\
        \Theta_d &= \brc{\frac{\ell_d}{m_d}: \ell_d \in \bktz{m_d - 1}} \, . \label{eq:message-sets}
    \end{align}
    The denominators $m_i - 1$ are positive, and thus the message sets are well-defined, because
    \begin{align}
        m_i &= n^{\frac{1}{d(p - 1)} \prn{R_i + \frac{\alpha}{2d}}} \\
        &\geq n_0^{\frac{1}{d(p - 1)} \prn{R_i + \frac{\alpha}{2d}}} \\
        &\geq \prn{2^{\frac{d(p - 1)}{\min_{i \in [d]} R_i}}}^{\frac{1}{d(p - 1)} \prn{R_i + \frac{\alpha}{2d}}} \\
        &\geq 2
    \end{align}
    for each $i \in [d]$. We represent a message $\boldsymbol{\mu}_i \in \mathcal{M}_i$ sent by sender $i$ as a flattened $(dp)$-tuple of triple-subscripted variables, whose latter two subscripts arise from the definitions $\mathcal{M}_i = \mathcal{L}_i^d$ and $\mathcal{L}_i \subset \Theta_i^p$:
    \begin{align}
        \boldsymbol{\mu}_i &= \brc{\theta_{i,c,k}}_{c=1, k=0}^{d, p-1} \\
        &= (\theta_{i,1,0}, \dots, \theta_{i,1,p-1}, \dots, \theta_{i,d,0}, \dots, \theta_{i,d,p-1}) \, .
    \end{align}
    
    \begin{figure}
        \centering
        \includegraphics[width=0.75\linewidth]{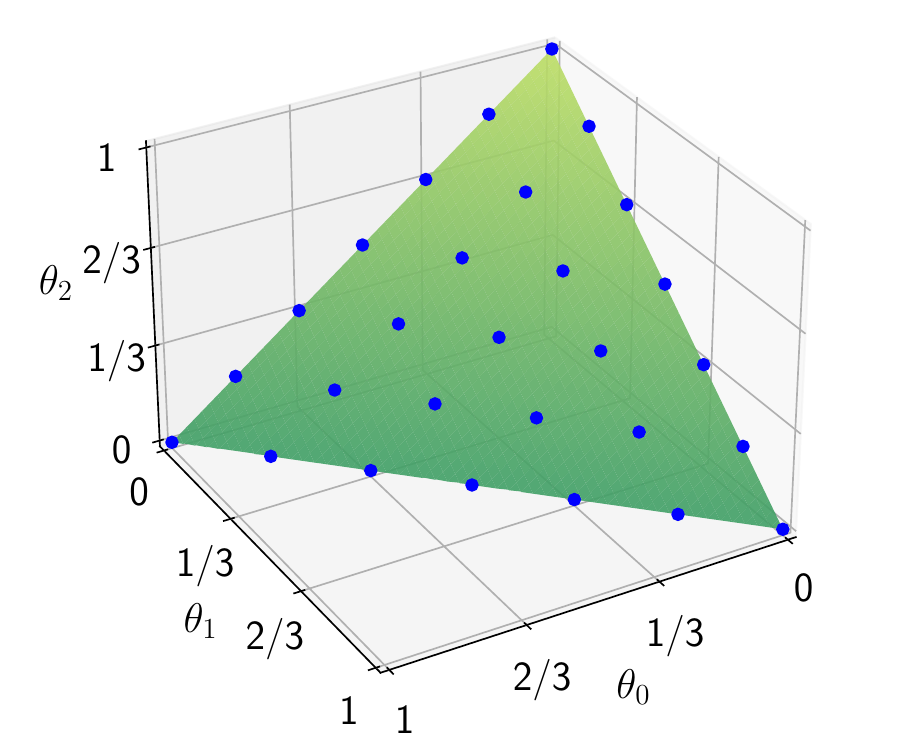}
        \caption{Visualization of the lattice $\mathcal{L}_i \subset \mathcal{S}_{p-1}$ in the case $i < d$, $p = 3$, and $m_i = 7$. The green surface is $\mathcal{S}_2$ and the blue dots are $\mathcal{L}_i$.}
        \label{figure:message-set}   
    \end{figure}
    
    \textbf{Proof of cardinality.} In this part, we verify that $\brc{\mathcal{M}_i}_{i=1}^d$ satisfies the definition of rate $d$-tuple, namely
    \begin{align}
        \forall i \in [d], \enspace \abs{\mathcal{M}_i} \geq n^{R_i} \, .
    \end{align}

    \emph{Case 1: $i < d$.} Each $\brc{\theta_k}_{k=0}^{p-1} \in \mathcal{L}_i$ corresponds to a \emph{weak $p$-composition} \cite[p. 18]{Stanley2011} of $m_i - 1$. Therefore,
    \begin{align}
        \abs{\mathcal{M}_i} = \abs{\mathcal{L}_i}^d \stackclap{(a)}{=} \multibinom{m_i}{p - 1}^d \, ,
    \end{align}
    where (a) holds by the \emph{stars and bars} method \cite[p. 19]{Stanley2011}. (\cref{lemma:cardinality} provides an alternative proof of this fact using a \emph{prefix sum} argument.) The message set size satisfies the lower bound
    \begin{align}
        \abs{\mathcal{M}_i} \stackclap{(a)}{\geq} \binom{m_i}{p - 1}^d \stackclap{(b)}{\geq} \frac{m_i^{d(p-1)}}{(p - 1)^{d(p-1)}} = \frac{n^{R_i + \frac{\alpha}{2d}}}{(p - 1)^{d(p-1)}} \, , \label{eq:mi-lower-bound}
    \end{align}
    where (a) holds because all subsets are trivially multisubsets and (b) follows from the lower bound on binomial coefficients in \cite[Eq. C.5]{CormenLeisersonRivestStein2022}. The message set size also asymptotically satisfies the upper bound
    \begin{align}
        \abs{\mathcal{M}_i} &= \binom{m_i + p - 2}{p - 1}^d \\
        &\stackclap{(a)}{\leq} e^{d(p-1)} \frac{(m_i + p - 2)^{d(p-1)}}{(p - 1)^{d(p-1)}} \\
        &\stackclap{(b)}{\lesssim} e^{d(p-1)} \frac{(2 m_i)^{d(p-1)}}{(p - 1)^{d(p-1)}} \\
        &\stackclap{(c)}{=} \prn{\frac{2e}{p - 1}}^{d(p-1)} n^{R_i + \frac{\alpha}{2d}} \, ,
    \end{align}
    where (a) follows from the upper bound on binomial coefficients in \cite[Eq. C.6]{CormenLeisersonRivestStein2022}, (b) holds because $p$ is constant with respect to $n$, and (c) follows from the definition of $m_i$. Hence, we remark that the message set size satisfies $\abs{\mathcal{M}_i} = \bigTheta{n^{R_i + \frac{\alpha}{2d}}}$. Lastly, $\abs{\mathcal{M}_i} \geq n^{R_i}$ follows from \cref{eq:mi-lower-bound} as desired because
    \begin{align}
        n \geq n_0 \geq (p - 1)^{\frac{2d^2 (p - 1)}{\alpha}} \, .
    \end{align}

    \emph{Case 2: $i = d$.} Each $\brc{\theta_k}_{k=0}^{p-1} \in \mathcal{L}_d$ corresponds to a weak $p$-composition of $m_d$, with the property $(*)$ that each part in the composition is less than $m_d$. Therefore,
    \begin{align}
        \abs{\mathcal{M}_d} = \abs{\mathcal{L}_d}^d \stackclap{(a)}{=} \prn{\multibinom{m_d + 1}{p - 1} - p}^d \, ,
    \end{align}
    where (a) holds because there are $p$ weak $p$-compositions of $m_d$ which do not satisfy $(*)$.\footnote{A weak $p$-composition of $m_d$ does not satisfy $(*)$ iff it contains one $m_d$ part and all other parts $0$; the $p$ such compositions place the $m_d$ part at different indices within the composition.} The message set size satisfies the lower bound
    \begin{align}
        \abs{\mathcal{M}_d} &\stackclap{(a)}{\geq} \prn{\binom{m_d + 1}{p - 1} - p}^d \\
        &\stackclap{(b)}{\geq} \prn{\frac{(m_d + 1)^{p-1}}{(p - 1)^{p-1}} - p}^d \\
        &\stackclap{(c)}{\geq} \prn{\frac{m_d^{p-1}}{(p - 1)^{p-1}} + \frac{m_d}{(p - 1)^{p-2}} - p}^d \\
        &\stackclap{(d)}{\geq} \frac{m_d^{d(p - 1)}}{(p - 1)^{d(p - 1)}} \\
        &= \frac{n^{R_d + \frac{\alpha}{2d}}}{(p - 1)^{d(p - 1)}} \, , \label{eq:md-lower-bound}
    \end{align}
    where (a) holds because all subsets are trivially multisubsets, (b) follows from the lower bound on binomial coefficients in \cite[Eq. C.5]{CormenLeisersonRivestStein2022}, (c) follows from the binomial theorem, and (d) holds because $n \geq p^{\frac{d(p - 1)^2}{R_d}}$ and thus
    \begin{align}
        \frac{m_d}{(p - 1)^{p-2}} &= \frac{n^{\frac{1}{d(p - 1)} \prn{R_d + \frac{\alpha}{2d}}}}{(p - 1)^{p-2}} \geq \frac{p^{p-1}}{(p - 1)^{p-2}} \geq p \, .
    \end{align}
    The message set size also asymptotically satisfies the upper bound
    \begin{align}
        \abs{\mathcal{M}_d} &= \prn{\binom{m_d + p - 1}{p - 1} - p}^d \\
        &\stackclap{(a)}{\leq} e^{d(p-1)} \frac{(m_d + p - 1)^{d(p-1)}}{(p - 1)^{d(p-1)}} \\
        &\stackclap{(b)}{\lesssim} e^{d(p-1)} \frac{(2 m_d)^{d(p-1)}}{(p - 1)^{d(p-1)}} \\
        &\stackclap{(c)}{=} \prn{\frac{2e}{p - 1}}^{d(p-1)} n^{R_d + \frac{\alpha}{2d}} \, ,
    \end{align}
    where (a) follows from the upper bound on binomial coefficients in \cite[Eq. C.6]{CormenLeisersonRivestStein2022}, (b) holds because $p$ is constant with respect to $n$, and (c) follows from the definition of $m_d$. Hence, we remark that the message set size satisfies $\abs{\mathcal{M}_d} = \bigTheta{n^{R_d + \frac{\alpha}{2d}}}$. Lastly, $\abs{\mathcal{M}_d} \geq n^{R_d}$ follows from \cref{eq:md-lower-bound} as desired because
    \begin{align}
        n \geq n_0 \geq (p - 1)^{\frac{2d^2 (p - 1)}{\alpha}} \, .
    \end{align}
    
    \textbf{Codeword segmentation.} For simplicity, assume $n$ is a multiple of $d$. In our analysis, we partition the codeword indices $1$ through $n$ into $d$ contiguous segments, each of length $\frac{n}{d}$. The $c$th segment consists of indices
    \begin{align}
        \mathcal{J}_c = \bktz{(c - 1) \frac{n}{d} + 1, \frac{cn}{d}} \, .
    \end{align}
    We further partition each segment into $d$ subsegments, and denote the indices in the $b$th subsegment by $\mathcal{J}_{c,b}$. The $b$th subsegment has cardinality $\abs{\mathcal{J}_{c,b}} = \rho_b \, \frac{n}{d}$, where
    \begin{align}
        \rho_b = \begin{cases}
            \frac{m_b - 1}{\prod_{i=1}^b m_i} \, , & \text{if $b \in [d - 1]$} \, , \\
            \frac{m_b}{\prod_{i=1}^b m_i} \, , & \text{if $b = d$} \, .
        \end{cases} \label{eq:subsegment-proportions}
    \end{align}
    
    \textbf{Proof of subsegment well-definedness.} In this part, we verify that $\sum_{b=1}^d \abs{\mathcal{J}_{c,b}} = \abs{\mathcal{J}_c}$, and so the subsegments are well-defined. Observe that
    \begin{align}
        \sum_{b=1}^d \rho_b &= \sum_{b=1}^{d-1} \frac{m_b - 1}{\prod_{i=1}^b m_i} + \frac{m_d}{\prod_{i=1}^d m_i} \\
        &= \prn{\prod_{i=1}^d m_i}^{-1} \prn{\sum_{b=1}^{d-1} (m_b - 1) \prod_{i=b+1}^d m_i + m_d} \\
        &= \prn{\prod_{i=1}^d m_i}^{-1} \prn{\sum_{b=1}^{d-1} \prod_{i=b}^d m_i - \sum_{b=1}^{d-1} \prod_{i=b+1}^d m_i + m_d} \\
        &= \prn{\prod_{i=1}^d m_i}^{-1} \prn{\sum_{b=1}^d \prod_{i=b}^d m_i - \sum_{b=2}^d \prod_{i=b}^d m_i} \\
        &= 1 \, , \label{eq:sum-of-subsegment-proportions}
    \end{align}
    so therefore
    \begin{align}
        \sum_{b=1}^d \abs{\mathcal{J}_{c,b}} = \frac{n}{d} \sum_{b=1}^d \rho_b = \frac{n}{d} = \abs{\mathcal{J}_c}
    \end{align}
    as desired. (\cref{lemma:subsegment-well-definedness} provides an alternative proof of this fact using recurrences.)
    
    \textbf{Encoders.} For notational simplicity, let
    \begin{align}
        \boldsymbol{\theta}_{i,c} = (\theta_{i,c,0}, \dots, \theta_{i,c,p-1})
    \end{align}
    for each $i, c \in [d]$, where $(\boldsymbol{\theta}_{i,1}, \dots, \boldsymbol{\theta}_{i,d}) = \boldsymbol{\mu}_i$ is the message sent by sender $i$. Each encoder $f_{i,n}: \mathcal{M}_i \rightarrow \bktz{p - 1}^n$ is randomized and outputs a codeword of length $n$ with alphabet size $p$. In each segment $c$, sender $i$ is active in subsegment $i$ and passive in all other subsegments. During its active phase, sender $i$'s encoder generates letters by sampling i.i.d. from $\Categorical{\boldsymbol{\theta}_{i,c}}$. During its passive phase, sender $i$ outputs all values $p - 1$ if it is one of the first $c - 1$ passive senders, and all zeros otherwise. Formally, for $j \in [n]$, let $c(j)$ and $b(j)$ be the indices of the segment and subsegment containing the $j$th letter, respectively. Then,
    \begin{align}
        &f_{i,n}(\boldsymbol{\mu}_i) = \prn{x_i}_1^n \enspace \text{where} \\
        &\forall j \in [n], \, \begin{cases}
            x_{i,j} \sim \Categorical{\boldsymbol{\theta}_{i,c(j)}} \, , & \text{if $i = b(j)$} \, , \\
            x_{i,j} = \prn{p - 1} \Iv{i \leq c(j) - 1} \, , & \text{if $i < b(j)$} \, , \\
            x_{i,j} = \prn{p - 1} \Iv{i \leq c(j)} \, , & \text{if $i > b(j)$} \, .
        \end{cases}
    \end{align}
    This coding scheme ensures the alphabets of $W_j$ from adjacent segments overlap at only one value, namely $c(p - 1)$ for segments $c$ and $c + 1$. \cref{table:general-encoder} visualizes the encoders' behavior in the case of $d = 3$ senders.
    \begin{table*}[t]
        \caption{Behavior of the encoders in the case of $d = 3$ senders. Each column corresponds to a subsegment in a segment. Each row corresponds to a sender. Each cell contains the value returned by the encoder, which is a sample from a categorical distribution (in active phase) or a constant (in passive phase).}
        \centering
        \begin{tabular}{c|c|c|c|c|c|c|c|c|c|}
            \cline{2-10}
            & \multicolumn{3}{c|}{segment 1} & \multicolumn{3}{c|}{segment 2} & \multicolumn{3}{c|}{segment 3} \\
            \cline{2-10}
            & sub 1 & sub 2 & sub 3 & sub 1 & sub 2 & sub 3 & sub 1 & sub 2 & sub 3 \\
            \hline
            \multicolumn{1}{|c|}{sender 1} & $\mathsf{Cat}(\boldsymbol{\theta}_{1,1})$ & $0$ & $0$ & $\mathsf{Cat}(\boldsymbol{\theta}_{1,2})$ & $p - 1$ & $p - 1$ & $\mathsf{Cat}(\boldsymbol{\theta}_{1,3})$ & $p - 1$ & $p - 1$ \\
            \multicolumn{1}{|c|}{sender 2} & $0$ & $\mathsf{Cat}(\boldsymbol{\theta}_{2,1})$ & $0$ & $p - 1$ & $\mathsf{Cat}(\boldsymbol{\theta}_{2,2})$ & $0$ & $p - 1$ & $\mathsf{Cat}(\boldsymbol{\theta}_{2,3})$ & $p - 1$ \\
            \multicolumn{1}{|c|}{sender 3} & $0$ & $0$ & $\mathsf{Cat}(\boldsymbol{\theta}_{3,1})$ & $0$ & $0$ & $\mathsf{Cat}(\boldsymbol{\theta}_{3,2})$ & $p - 1$ & $p - 1$ & $\mathsf{Cat}(\boldsymbol{\theta}_{3,3})$ \\
            \hline
        \end{tabular}
        \label{table:general-encoder}
    \end{table*}
    
    \textbf{Time sharing.} In this part, we present several definitions pertaining to the time sharing scheme used in our proof. Let $m = \prod_{i=1}^d m_i$. Let $\Phi$ be the set
    \begin{align}
        \Phi = \brc{\frac{\ell}{m}: \ell \in \bktz{m - 1}} \, .
    \end{align}
    Let $h: \Theta_1 \times \cdots \times \Theta_d \rightarrow \Phi$ be the mapping
    \begin{align}
        \phi = h(\theta_1, \dots, \theta_d) = \sum_{b=1}^d \rho_b \theta_b \, . \label{eq:weighted-sum-of-subsegment-thetas}
    \end{align}
    Observe that $h$ is a bijection, because
    \begin{align}
        \phi &\stackclap{(a)}{=} \sum_{b=1}^{d-1} \frac{m_b - 1}{\prod_{i=1}^b m_i} \, \theta_b + \frac{m_d}{\prod_{i=1}^d m_i} \, \theta_d \\
        &= \frac{1}{m} \prn{\sum_{b=1}^{d-1} \theta_b (m_b - 1) \prod_{i=b+1}^d m_i + \theta_d m_d} \\
        &\stackclap{(b)}{=} \frac{1}{m} \sum_{b=1}^d \ell_b \prod_{i=b+1}^d m_i \, , \label{eq:phi-characterization}
    \end{align}
    where (a) follows from substituting \cref{eq:subsegment-proportions} into \cref{eq:weighted-sum-of-subsegment-thetas}, and $\ell_b$ in (b) arises from the definition of message sets, namely
    \begin{align}
        \ell_b = \begin{cases}
            \theta_b (m_b - 1) \, , & \text{if $b \in [d - 1]$} \, , \\
            \theta_b m_b \, , & \text{if $b = d$} \, .
        \end{cases} \label{eq:ell-characterization}
    \end{align}
    Since $\ell_b \in \bktz{m_b - 1}$ for each $b \in [d]$, the mapping $h$ can be interpreted as converting a multi-dimensional array index into the corresponding one-dimensional index for the flattened view of the array. More explicitly, the inverse mapping $h^{-1}: \Phi \rightarrow \Theta_1 \times \cdots \times \Theta_d$ is given by
    \begin{align}
        (\theta_1, \dots, \theta_d) &= h^{-1}(\phi) \, , \\
        \forall b \in [d - 1], \enspace \theta_b &= \frac{1}{m_b - 1} \prn{\floor{\frac{m \phi}{\prod_{i=b+1}^d m_i}} \mod{m_b}} \, , \\
        \theta_d &= \frac{1}{m_d} \prn{m \phi \mod{m_d}} \, ,
    \end{align}
    thus showing injectivity of $h$.\footnote{Note that $\phi$ is not an integer in general, so $m \phi \mod{m_d}$ cannot be simplified to $0$.} The inverse mapping is readily justified by simple algebra: For $b < d$, we have
    \begin{align}
        \prn{m_b - 1} \theta_b &\stackclap{(a)}{=} \ell_b \\
        &\stackclap{(b)}{=} \prn{\sum_{b'=1}^{b-1} \ell_{b'} \prod_{i=b'+1}^b m_i + \ell_b} \mod{m_b} \\
        &\stackclap{(c)}{=} \Biggl\lfloor \sum_{b'=1}^{b-1} \ell_{b'} \prod_{i=b'+1}^b m_i + \ell_b \\
        &\quad + \frac{\sum_{b'=b+1}^d \ell_{b'} \prod_{i=b'+1}^d m_i}{\prod_{i=b+1}^d m_i} \Biggr\rfloor \mod{m_b} \\
        &= \floor{\frac{\sum_{b'=1}^d \ell_{b'} \prod_{i=b'+1}^d m_i}{\prod_{i=b+1}^d m_i}} \mod{m_b} \\
        &\stackclap{(d)}{=} \floor{\frac{m \phi}{\prod_{i=b+1}^d m_i}} \mod{m_b} \, ,
    \end{align}
    where (a) follows from \cref{eq:ell-characterization}, (b) holds because $\ell_b \in \bktz{m_b - 1}$, (c) holds because
    \begin{align}
        \sum_{b'=b+1}^d \ell_{b'} \prod_{i=b'+1}^d m_i &\leq \sum_{b'=b+1}^d (m_{b'} - 1) \prod_{i=b'+1}^d m_i \\
        &= \sum_{b'=b+1}^d \prn{\prod_{i=b'}^d m_i - \prod_{i=b'+1}^d m_i} \\
        &= \sum_{b'=b+1}^d \prod_{i=b'}^d m_i - \sum_{b'=b+2}^{d+1} \prod_{i=b'}^d m_i \\
        &= \prod_{i=b+1}^d m_i - 1 \\
        &< \prod_{i=b+1}^d m_i \, ,
    \end{align}
    and (d) follows from \cref{eq:phi-characterization}. For $b = d$, we have
    \begin{align}
        m_d \theta_d &\stackclap{(a)}{=} \ell_d \\
        &\stackclap{(b)}{=} \prn{\sum_{b'=1}^{d-1} \ell_{b'} \prod_{i=b'+1}^d m_i + \ell_d} \mod{m_d} \\
        &\stackclap{(c)}{=} m \phi \mod{m_d} \, ,
    \end{align}
    where (a) follows from \cref{eq:ell-characterization}, (b) holds because $\ell_d \in \bktz{m_d - 1}$, and (c) follows from \cref{eq:phi-characterization}. Lastly, we have
    \begin{align}
        \abs{\Theta_1 \times \cdots \times \Theta_d} = \abs{\Theta_1} \times \cdots \times \abs{\Theta_d} = \prod_{i=1}^d m_i = m = \abs{\Phi} \, ,
    \end{align}
    and so $h$ is bijective as mentioned.
    
    \textbf{Decoder.} Given the output codeword $y_1^n$, the decoder $g_n: \bktz{d(p - 1)}^n \rightarrow \mathcal{M} \cup \brc{\mathsf{error}}$ executes the following:
    \begin{enumerate}
        \item Form the matrix
        \begin{align}
            \mathbf{B} = \begin{bmatrix}
                \mathbf{P}_{Z|W}^\top & \mathbf{0} \\
                \mathbf{0} & \mathbf{I}_d
            \end{bmatrix} \in \mathbb{R}^{(dp+1) \times (dp+1)} \, .
        \end{align}
        \item Form the matrix $\mathbf{C}_1 \in \mathbb{B}^{(d(p-1)+1) \times dp}$, given by:
        \begin{align}
            \bkt{\mathbf{C}_1}_{s,(c,k)} = \begin{cases}
                \uline{1} \, , & \text{if $(c, k) =$} \\
                & \prn{\floor{\frac{s}{p - 1}} + 1, \, s \mod{p - 1}} \, , \\
                \smashwave{1} \, , & \text{if $0 \equiv s \mod{p - 1}$} \\
                & \text{and $(c, k) = \prn{\frac{s}{p - 1}, \, p - 1} \, ,$} \\
                0 \, , & \text{otherwise} \, .
            \end{cases} \label{eq:c1-matrix}
        \end{align}
        We index the rows by $s \in \bktz{d(p - 1)}$, as in $\mathbf{P}_{Z|W}$. We index the columns by $(c,k) \in [d] \times \bktz{p-1}$, matching the latter two indices of the message variables $\theta_{i,c,k}$. The $1$s in the formula above are underlined to match \cref{figure:c-matrix}. Intuitively, $\mathbf{C}_1$ specifies the output value $s$ from the adder (and hence the input value to the DMC $\mathbf{P}_{Z|W}$) when an active sender in segment $c$ sends value $k$. We mention two further remarks:
        \begin{enumerate}
            \item For each row $s$, at most one column $(c,k)$ satisfies the \uline{first} case in \cref{eq:c1-matrix}, which uniquely determines $c$ and $k$ by the quotient and remainder of $s$ divided by $p - 1$, respectively.
            \item Any row of the form $s = c(p - 1)$ for some $c \in [d - 1]$ contains two $1$s, from the \uwave{second} and \uline{first} cases in \cref{eq:c1-matrix} respectively, due to the possibility of this value $s$ appearing in two adjacent segments $c$ and $c + 1$. All other rows contain one $1$. This is readily observed by rearranging the conditions in \cref{eq:c1-matrix} to find $s$ in terms of $(c, k)$:
            \begin{align}
                \bkt{\mathbf{C}_1}_{s,(c,k)} = \Iv{s = (c - 1)(p - 1) + k} \, . \label{eq:c1-matrix-characterization}
            \end{align}
        \end{enumerate}
        \item Form the matrix $\mathbf{C}_2 \in \mathbb{B}^{d \times dp}$, given by
        \begin{align}
            \bkt{\mathbf{C}_2}_{s,(c,k)} = \Iv{c = s} \, . \label{eq:c2-matrix}
        \end{align}
        We index the rows by $s \in [d]$.  We index the columns by $(c,k) \in [d] \times \bktz{p - 1}$, as in $\mathbf{C}_1$. The matrix $\mathbf{C}_2$ is used to encode the sum-to-one constraint on each $\boldsymbol{\phi}_c = \brc{\phi_{c,k}}_{k=0}^{p-1}$.
        \item Collect $\mathbf{C}_1$ and $\mathbf{C}_2$ into the block matrix
        \begin{align}
            \mathbf{C} = \begin{bmatrix}
                \mathbf{C}_1 \\
                \mathbf{C}_2
            \end{bmatrix} \in \mathbb{B}^{(dp+1) \times dp} \, .
        \end{align}
        An example of this matrix is visualized in \cref{figure:c-matrix}.
        \item Compute the matrix $\mathbf{A} = \mathbf{B} \mathbf{C} \in \mathbb{R}^{(dp+1) \times dp}$. By block matrix multiplication, observe that
        \begin{align}
            \mathbf{A} = \begin{bmatrix}
                \mathbf{P}_{Z|W}^\top & \mathbf{0} \\
                \mathbf{0} & \mathbf{I}_d
            \end{bmatrix} \begin{bmatrix}
                \mathbf{C}_1 \\
                \mathbf{C}_2
            \end{bmatrix} = \begin{bmatrix}
                \mathbf{P}_{Z|W}^\top \mathbf{C}_1 \\
                \mathbf{C}_2
            \end{bmatrix} \, .
        \end{align}
        At a high level, $\mathbf{A}$ encodes the system of $dp + 1$ equations which we compute the least-squares solution of, as discussed in Step 3 of \cref{subsection:general-achievability}. The first $d(p - 1) + 1$ equations are given by the submatrix $\mathbf{P}_{Z|W}^\top \mathbf{C}_1$, whose columns specify the conditional distributions in $\mathbf{P}_{Z|W}$ which $\mathbf{p}_Y$ is a linear combination of. The last $d$ equations are given by the submatrix $\mathbf{C}_2$, which encodes the sum-to-one constraint on each $\boldsymbol{\phi}_c$.
        \item Compute the empirical distribution $\hat{\mathbf{p}}_Y \in \mathcal{S}_{d(p - 1)}$, given by
        \begin{align}
            \forall t \in \bktz{d(p - 1)}, \enspace \bkt{\hat{\mathbf{p}}_Y}_t = \frac{1}{n} \sum_{j=1}^n \Iv{y_j = t} \, .
        \end{align}
        \item Form the vector
        \begin{align}
            \hat{\mathbf{b}} = \begin{bmatrix}
                d \hat{\mathbf{p}}_Y^\top \\
                \mathbf{1}_d
            \end{bmatrix} \in \mathbb{R}^{dp + 1} \, .
        \end{align}
        \item Using the normal equations, compute the least-squares solution (cf. \cite{Makur2020b})
        \begin{align}
            \tilde{\boldsymbol{\phi}} &= \brc{\tilde{\phi}_{c,k}}_{c=1,k=0}^{d,p-1} \\
            &= \arg \min_{\boldsymbol{\phi} \in \mathbb{R}^{dp}} \norm{\mathbf{A} \boldsymbol{\phi} - \hat{\mathbf{b}}}_2 \\
            &= \underbrace{\prn{\mathbf{A}^\top \mathbf{A}}^{-1} \mathbf{A}^\top}_{\mathbf{A}^\dagger} \hat{\mathbf{b}} \, ,
        \end{align}
        where $\mathbf{A}^\dagger$ is the \emph{Moore-Penrose pseudoinverse} of $\mathbf{A}$.
        \item Compute $\hat{\boldsymbol{\phi}} = \brc{\hat{\phi}_{c,k}}_{c=1,k=0}^{d,p-1} \in \Phi^{dp}$ by rounding each entry of $\tilde{\boldsymbol{\phi}}$ to the nearest element in $\Phi$ (cf. \cite{Makur2020b}):
        \begin{align}
            \forall (c,k) \in [d] \times \bktz{p - 1}, \enspace \hat{\phi}_{c,k} = \arg \min_{\phi \in \Phi} \abs{\tilde{\phi}_{c,k} - \phi} \, .
        \end{align}
        \item Convert each $\hat{\phi}_{c,k}$ into a $d$-tuple in $\Theta_1 \times \cdots \times \Theta_d$ using the inverse mapping $h^{-1}$:
        \begin{align}
            \forall (c,k) \in [d] \times \bktz{p - 1}, \enspace \brc{\hat{\theta}_{i,c,k}}_{i=1}^d = h^{-1} \mprn{\hat{\phi}_{c,k}} \, .
        \end{align}
        \item Form the predicted messages
        \begin{align}
            \forall i \in [d], \enspace \hat{\boldsymbol{\mu}}_i = \brc{\hat{\theta}_{i,c,k}}_{c=1,k=0}^{d,p-1} \, .
        \end{align}
        \item If $\hat{\boldsymbol{\mu}}_i \notin \mathcal{M}_i$ for any $i \in [d]$, return $\mathsf{error}$. Otherwise, return the predicted messages $\prn{\hat{\boldsymbol{\mu}}_1, \dots, \hat{\boldsymbol{\mu}}_d}$.
    \end{enumerate}

    \begin{figure}
        \begin{align}
            \mathbf{C} &= \bkt{\begin{array}{c@{}|c@{}}
                \begin{matrix}
                    \uline{1} & \zero & \zero \\
                    \zero & \uline{1} & \zero
                \end{matrix} & \begin{matrix}
                    \zero & \zero & \zero \\
                    \zero & \zero & \zero
                \end{matrix} \\
                \begin{matrix}
                    \zero & \zero & \smashwave{1} \\
                    \zero & \zero & \zero
                \end{matrix} & \begin{matrix}
                    \uline{1} & \zero & \zero \\
                    \zero & \uline{1} & \zero
                \end{matrix} \\
                \begin{matrix}
                    \zero & \zero & \zero
                \end{matrix} & \begin{matrix}
                    \zero & \zero & \smashwave{1}
                \end{matrix} \\
                \hline
                \begin{matrix}
                    1 & 1 & 1 \\
                    \zero & \zero & \zero
                \end{matrix} & \begin{matrix}
                    \zero & \zero & \zero \\
                    1 & 1 & 1
                \end{matrix}
            \end{array}} \in \mathbb{B}^{7 \times 6}
        \end{align}
        \caption{Visualization of $\mathbf{C}$ in the case $d = 2$ and $p = 3$. Zeros are omitted for clarity. The two block rows are $\mathbf{C}_1 \in \mathbb{B}^{5 \times 6}$ and $\mathbf{C}_2 \in \mathbb{B}^{2 \times 6}$ respectively. Each block column is a group of columns with the same $c$ index. The number of block columns is $d = 2$. The width of each block column is $p = 3$. The underlined $1$s were generated by the corresponding cases in \cref{eq:c1-matrix}.}
        \label{figure:c-matrix}
    \end{figure}
    
    \textbf{Proof of decoder well-definedness.} In this part, we verify that $\mathbf{A}^\top \mathbf{A}$ is indeed invertible, and so the decoder is well-defined. First, observe that $\mathbf{C}$ has full column rank because each column is not a linear combination of the columns to its left, i.e.,
    \begin{align}
        &\forall (c, k) \in [d] \times \bktz{p - 1}, \, \bkt{\mathbf{C}}_{(c,k)} \notin \spn \Bigl\{ \bkt{\mathbf{C}}_{(c',k')}: \\
        &\quad \text{$(c',k') \in [d] \times \bktz{p-1}$ such that $c'p + k' < cp + k$} \Bigr\} \, , \label{eq:full-column-rank}
    \end{align}
    due to column $(c, k)$ being the leftmost column that contains a non-zero entry in
    \begin{itemize}
        \item Row $c$ of $\mathbf{C}_2$, if $k = 0$; or
        \item Row $(c - 1)(p - 1) + k$ of $\mathbf{C}_1$, if $k > 0$. This is readily observed from \cref{eq:c1-matrix-characterization}.
    \end{itemize}
    By inspection, since $\mathbf{P}_{Z|W}$ has full rank, $\mathbf{B}$ has full rank. Since $\mathbf{A}$ is the product of two matrices with full column rank, $\mathbf{A}$ has full column rank. Since a Gramian matrix has the same rank as its vector realization,
    \begin{align}
        \rank{\mathbf{A}^\top \mathbf{A}} = \rank{\mathbf{A}} = dp \, ,
    \end{align}
    and thus $\mathbf{A}^\top \mathbf{A}$ is invertible as desired.
    
    \textbf{Proof of achievability.} In this proof, all probabilities $\P{\cdot}$ are conditioned on sending the true messages $(\boldsymbol{\mu}_1, \dots, \boldsymbol{\mu}_d)$.
    
    \emph{Step 1: Upper-bounding error in least-squares solution.} Let $\mathbf{p}_{W,c,b}$ and $\mathbf{p}_{Z,c,b}$ be the true distributions of $W_j$ and $Z_j$, respectively, given $j \in \mathcal{J}_{c,b}$. The dynamics of the formal model are governed by the equations:
    \begin{align}
        \forall c, b \in [d], \enspace \mathbf{p}_{W,c,b} &= \sum_{k=0}^{p-1} \theta_{b,c,k} \mathbf{e}_{(c-1)(p-1) + k}^\top \, , \label{eq:general-encoders-and-adder} \\
        \forall c, b \in [d], \enspace \mathbf{p}_{Z,c,b} &= \mathbf{p}_{W,c,b} \mathbf{P}_{Z|W} \, , \label{eq:general-dmc} \\
        \mathbf{p}_Y &= \sum_{c=1}^d \sum_{b=1}^d \frac{\rho_b}{d} \mathbf{p}_{Z,c,b} \, , \label{eq:general-random-permutation-block} \\
        \forall i, c \in [d], \enspace \sum_{k=0}^{p-1} \theta_{i,c,k} &= 1 \, , \label{eq:sum-to-one}
    \end{align}
    where \cref{eq:general-encoders-and-adder} models the encoders and adder, \cref{eq:general-dmc} models the DMC, \cref{eq:general-random-permutation-block} models the random permutation block by \cref{lemma:y-marginal}, and \cref{eq:sum-to-one} encodes the fact that $\brc{\theta_{i,c,k}}_{k=0}^{p-1} \in \mathcal{L}_i$.
    
    Define $\phi_{c,k} = h(\theta_{1,c,k}, \dots, \theta_{d,c,k})$ for each $(c, k) \in [d] \times \bktz{p - 1}$. Collect these variables into a vector $\boldsymbol{\phi}^* = \brc{\phi_{c,k}}_{c=1,k=0}^{d,p-1} \in \mathbb{R}^{dp}$. Combining equations \cref{eq:general-encoders-and-adder}-\cref{eq:general-random-permutation-block},
    \begin{align}
        d \mathbf{p}_Y^\top &\stackclap{(a)}{=} \sum_{c=1}^d \sum_{b=1}^d \rho_b \mathbf{P}_{Z|W}^\top \sum_{k=0}^{p-1} \theta_{b,c,k} \mathbf{e}_{(c-1)(p-1) + k} \\
        &\stackclap{(b)}{=} \sum_{c=1}^d \mathbf{P}_{Z|W}^\top \sum_{k=0}^{p-1} \phi_{c,k} \mathbf{e}_{(c-1)(p-1) + k} \\
        &= \sum_{c=1}^d \sum_{k=0}^{p-1} \phi_{c,k} \bkt{\mathbf{P}_{Z|W}}_{\ang{(c-1)(p-1) + k}} \\
        &\stackclap{(c)}{=} \sum_{s=0}^{d(p-1)} \bkt{\mathbf{P}_{Z|W}}_{\ang{s}} \begin{cases}
            \uwave{\phi_{\frac{s}{p-1}, \, p-1}} + \uline{\phi_{\frac{s}{p-1} + 1, \, 0}} \, , \\
            \qquad \text{if $0 \equiv s \mod{p - 1}$} \\
            \uline{\phi_{\floor{\frac{s}{p-1}} + 1, \, s \mod{p - 1}}} \, , \\
            \qquad \text{otherwise}
        \end{cases} \\
        &\stackclap{(d)}{=} \sum_{s=0}^{d(p-1)} \bkt{\mathbf{P}_{Z|W}}_{\ang{s}} \bkt{\mathbf{C}_1 \boldsymbol{\phi}^*}_s \\
        &= \mathbf{P}_{Z|W}^\top \mathbf{C}_1 \boldsymbol{\phi}^* \, ,
    \end{align}
    where (a) follows from substituting $\mathbf{p}_{W,c,b}$ and $\mathbf{p}_{Z,c,b}$ into $\mathbf{p}_Y$, (b) holds by \cref{eq:weighted-sum-of-subsegment-thetas}, (c) reindexes the double sum in terms of the rows of $\mathbf{P}_{Z|W}$, (d) follows from \cref{eq:c1-matrix}, and we zero-index $\mathbf{C}_1 \boldsymbol{\phi}^*$ in (d). To avoid creating special cases in (c) for $s = 0$ and $s = d(p - 1)$, we assume that $\phi_{c,k} = 0$ for any $c \notin [d]$. The terms in (c) are underlined to match their corresponding entries in $\mathbf{C}_1$, as defined in \cref{eq:c1-matrix}, to show why (d) follows. Note that the case structure in (c) mirrors remark (2b) from the decoder definition.
    
    For each $c \in [d]$, $\brc{\phi_{c,k}}_{k=0}^{p-1}$ sums to $1$, because
    \begin{align}
        \sum_{k=0}^{p-1} \phi_{c,k} &\stackclap{(a)}{=} \sum_{k=0}^{p-1} \sum_{b=1}^d \rho_b \theta_{b,c,k} \\
        &= \sum_{b=1}^d \rho_b \sum_{k=0}^{p-1} \theta_{b,c,k} \\
        &\stackclap{(b)}{=} \sum_{b=1}^d \rho_b \\
        &\stackclap{(c)}{=} 1 \, ,
    \end{align}
    where (a) follows from \cref{eq:weighted-sum-of-subsegment-thetas}, (b) holds because $\brc{\theta_{b,c,k}}_{k=0}^{p-1} \in \mathcal{L}_b$ sums to $1$, and (c) follows from \cref{eq:sum-of-subsegment-proportions}. This can be vectorized as $\mathbf{C}_2 \mathbf{\boldsymbol{\phi}}^* = \mathbf{1}_d$. Define a vector
    \begin{align}
        \mathbf{b} = \begin{bmatrix}
            d \mathbf{p}_Y^\top \\
            \mathbf{1}_d
        \end{bmatrix} \in \mathbb{R}^{dp + 1} \, .
    \end{align}
    Combining the above,
    \begin{align}
        \mathbf{A} \boldsymbol{\phi}^* = \begin{bmatrix}
            \mathbf{P}_{Z|W}^\top & \mathbf{0} \\
            \mathbf{0} & \mathbf{I}_d
        \end{bmatrix} \begin{bmatrix}
            \mathbf{C}_1 \\
            \mathbf{C}_2
        \end{bmatrix} \boldsymbol{\phi}^* = \begin{bmatrix}
            \mathbf{P}_{Z|W}^\top \mathbf{C}_1 \boldsymbol{\phi}^* \\
            \mathbf{C}_2 \boldsymbol{\phi}^*
        \end{bmatrix} = \mathbf{b} \, ,
    \end{align}
    and so $\boldsymbol{\phi}^* = \prn{\mathbf{A}^\top \mathbf{A}}^{-1} \mathbf{A}^\top \mathbf{b}$. Therefore,
    \begin{align}
        \norm{\tilde{\boldsymbol{\phi}} - \boldsymbol{\phi}^*}_\infty &\leq \norm{\prn{\mathbf{A}^\top \mathbf{A}}^{-1} \mathbf{A}^\top}_\infty \norm{\hat{\mathbf{b}} - \mathbf{b}}_\infty \\
        &\stackclap{(a)}{\leq} \sqrt{dp + 1} \norm{\prn{\mathbf{A}^\top \mathbf{A}}^{-1} \mathbf{A}^\top}_2 \norm{\hat{\mathbf{b}} - \mathbf{b}}_\infty \\
        &\stackclap{(b)}{\leq} \sqrt{dp + 1} \underbrace{\norm{\prn{\mathbf{A}^\top \mathbf{A}}^{-1}}_2}_{\circled{1}} \underbrace{\norm{\mathbf{A}^\top}_2}_{\circled{2}} \underbrace{\norm{\hat{\mathbf{b}} - \mathbf{b}}_\infty}_{\circled{3}} \, ,
    \end{align}
    where (a) holds by the equivalence of $\ell^\infty$ and $\ell^2$ norms and (b) holds by the submultiplicativity of the $\ell^2$ matrix norm.
    
    Next, we upper-bound $\circled{1}$:
    \begin{align}
        \circled{1} &= \frac{1}{\sigma_{\min}(\mathbf{A}^\top \mathbf{A})} \\
        &= \frac{1}{\sigma_{\min}^2(\mathbf{A})} \\
        &\leq \frac{1}{\sigma_{\min}^2(\mathbf{B}) \, \sigma_{\min}^2(\mathbf{C})} \\
        &\stackclap{(a)}{=} \frac{1}{\min \brc{\sigma_{\min} \mprn{\mathbf{P}_{Z|W}}, 1}^2 \sigma_{\min}^2(\mathbf{C})} \\
        &\stackclap{(b)}{=} \frac{1}{\sigma_{\min}^2 \mprn{\mathbf{P}_{Z|W}} \, \sigma_{\min}^2(\mathbf{C})} \, ,
    \end{align}
    where (a) holds by definition of $\mathbf{B}$, and (b) holds because $\mathbf{P}_{Z|W}$ is row stochastic and so
    \begin{align}
        \sigma_{\min} \mprn{\mathbf{P}_{Z|W}} &= \prn{\norm{\mathbf{p}}_1 \sigma_{\min}^t \mprn{\mathbf{P}_{Z|W}}}^{\frac{1}{t}} \\
        &\leq \prn{\sqrt{d(p - 1) + 1} \norm{\mathbf{p}}_2 \sigma_{\min}^t \mprn{\mathbf{P}_{Z|W}}}^{\frac{1}{t}} \\
        &\leq \prn{\sqrt{d(p - 1) + 1} \norm{\mathbf{p} \mathbf{P}_{Z|W}^t}_2}^{\frac{1}{t}} \\
        &\leq \prn{\sqrt{d(p - 1) + 1} \norm{\mathbf{p} \mathbf{P}_{Z|W}^t}_1}^{\frac{1}{t}} \\
        &= \sqrt[2t]{d(p - 1) + 1}
    \end{align}
    for any $t \in \mathbb{N}$ and $\mathbf{p} \in \mathcal{S}_{d(p-1)}$; taking the limit as $t \rightarrow \infty$, we have $\sigma_{\min} \mprn{\mathbf{P}_{Z|W}} \leq 1$. Define a matrix $\tilde{\mathbf{C}} \in \mathbb{B}^{(dp+1) \times (dp+1)}$ by prepending one column to $\mathbf{C}$:
    \begin{align}
        \tilde{\mathbf{C}} = \begin{bmatrix}
            \tilde{\mathbf{C}}_1 \\
            \tilde{\mathbf{C}}_2
        \end{bmatrix}, \, \tilde{\mathbf{C}}_1 &= \bkt{\begin{array}{@{}c@{}|@{}c@{}}
            \begin{matrix}
                1 \\
                0 \\
                \vdots \\
                0
            \end{matrix} & \makebox[2em]{$\mathbf{C}_1$}
        \end{array}}, \, \tilde{\mathbf{C}}_2 = \bkt{\begin{array}{@{}c@{}|@{}c@{}}
            \begin{matrix}
                0 \\
                \vdots \\
                0
            \end{matrix} & \makebox[2em]{$\mathbf{C}_2$}
        \end{array}} \, . \label{eq:c-tilde-matrix}
    \end{align}
    It follows that
    \begin{align}
        \frac{1}{\sigma_{\min}(\mathbf{C})} \stackclap{(a)}{\leq} \frac{1}{\sigma_{\min} \mprn{\tilde{\mathbf{C}}}} = \norm{\tilde{\mathbf{C}}^{-1}}_2 \leq \frnorm{\tilde{\mathbf{C}}^{-1}} \stackclap{(b)}{\leq} dp + 1 \, ,
    \end{align}
    where (a) holds because adding a column to a tall matrix does not increase its minimum singular value (\cref{lemma:minimum-singular-value}), and (b) holds because all non-zero entries of $\tilde{\mathbf{C}}^{-1}$ have unit magnitude by \cref{lemma:c-tilde-inverse}. Combining the above,
    \begin{align}
        \circled{1} \leq \frac{(dp + 1)^2}{\sigma_{\min}^2 \mprn{\mathbf{P}_{Z|W}}} \, .
    \end{align}
    
    Next, we upper-bound $\circled{2}$:
    \begin{align}
        \circled{2} &\stackclap{(a)}{\leq} \norm{\mathbf{B}}_2 \frnorm{\mathbf{C}} \\
        &\stackclap{(b)}{=} \max \brc{\norm{\mathbf{P}_{Z|W}}_2, \norm{\mathbf{I}_d}_2} \frnorm{\mathbf{C}} \\
        &\stackclap{(c)}{=} \norm{\mathbf{P}_{Z|W}}_2 \frnorm{\mathbf{C}} \\
        &\stackclap{(d)}{=} \norm{\mathbf{P}_{Z|W}}_2 \sqrt{2dp} \\
        &\stackclap{(e)}{\leq} \sqrt{d(p - 1) + 1} \norm{\mathbf{P}_{Z|W}}_\infty \sqrt{2dp} \\
        &\stackclap{(f)}{=} \sqrt{d(p - 1) + 1} \sqrt{2dp} \\
        &\leq \sqrt{2} \prn{dp + 1} \, ,
    \end{align}
    where (a) holds by the submultiplicativity of the $\ell^2$ matrix norm, (b) holds by the block diagonal structure of $\mathbf{B}$ \cite[Eq. 1.2]{LiMathias2003}, (c) holds because $\mathbf{P}_{Z|W}$ is row stochastic and so
    \begin{align}
        \norm{\mathbf{P}_{Z|W}}_2 \stackclap{(g)}{\geq} \rho \mprn{\mathbf{P}_{Z|W}} = 1
    \end{align}
    by \cite[Theorem 5.6.9]{HornJohnson2013} and \cite[Section 8.7]{HornJohnson2013}, (d) holds because each column of $\mathbf{C}$ contains two $1$s and all other entries $0$, (e) holds by the equivalence of $\ell^2$ and $\ell^\infty$ norms, (f) holds because each row of $\mathbf{P}_{Z|W}$ sums to $1$ and so the maximum row sum matrix norm \cite[Example 5.6.5]{HornJohnson2013} of $\mathbf{P}_{Z|W}$ is $1$, and the notation $\rho(\cdot)$ in (g) refers to the spectral radius.
    
    By definition of $\mathbf{b}$ and $\hat{\mathbf{b}}$, we have $\circled{3} = d \norm{\hat{\mathbf{p}}_Y - \mathbf{p}_Y}_\infty$. Combining the bounds on $\circled{1}$ to $\circled{3}$,
    \begin{align}
        \norm{\tilde{\boldsymbol{\phi}} - \boldsymbol{\phi}^*}_\infty \leq \frac{\sqrt{2} \prn{dp + 1}^{\frac{7}{2}} d}{\sigma_{\min}^2 \prn{\mathbf{P}_{Z|W}}} \norm{\hat{\mathbf{p}}_Y - \mathbf{p}_Y}_\infty \, .
    \end{align}
    
    \emph{Step 2: Concentration bound for empirical distribution of $Y$.} We have
    \begin{align}
        &\equad \P{g_n(Y_1^n) \neq (\boldsymbol{\mu}_1, \dots, \boldsymbol{\mu}_d)} \\
        &= \P{g_n(Y_1^n) = \mathsf{error} \lor \exists i \in [d], \, \hat{\boldsymbol{\mu}}_i \neq \boldsymbol{\mu}_i} \\
        &\stackclap{(a)}{=} \P{\exists i \in [d], \, \exists (c,k) \in [d] \times \bktz{p - 1}, \, \hat{\theta}_{i,c,k} \neq \theta_{i,c,k}} \\
        &\stackclap{(b)}{=} \P{\exists (c,k) \in [d] \times \bktz{p - 1}, \, \hat{\phi}_{c,k} \neq \phi_{c,k}} \\
        &\stackclap{(c)}{\leq} \P{\norm{\tilde{\boldsymbol{\phi}} - \boldsymbol{\phi}^*}_\infty \geq \frac{1}{2m}} \\
        &\stackclap{(d)}{=} \P{\norm{\tilde{\boldsymbol{\phi}} - \boldsymbol{\phi}^*}_\infty \geq \frac{1}{2} n^{-\frac{1}{d(p - 1)} \sum_{i=1}^d \brc{R_i + \frac{\alpha}{2d}}}} \\
        &\stackclap{(e)}{=} \P{\norm{\tilde{\boldsymbol{\phi}} - \boldsymbol{\phi}^*}_\infty \geq \frac{1}{2} n^{-\frac{1}{2} + \frac{\alpha}{2d(p - 1)}}} \\
        &\stackclap{(f)}{\leq} \P{\frac{\sqrt{2} \prn{dp + 1}^{\frac{7}{2}} d}{\sigma_{\min}^2 \prn{\mathbf{P}_{Z|W}}} \norm{\hat{\mathbf{p}}_Y - \mathbf{p}_Y}_\infty \geq \frac{1}{2} n^{-\frac{1}{2} + \frac{\alpha}{2d(p - 1)}}} \\
        &\stackclap{(g)}{\leq} \P{\norm{\hat{\mathbf{p}}_Y - \mathbf{p}_Y}_\infty \geq \sqrt{\frac{\log_e n}{n}}} \\
        &\stackclap{(h)}{\leq} \sum_{t=0}^{d(p-1)} \P{\abs{\bkt{\hat{\mathbf{p}}_Y}_t - \bkt{\mathbf{p}_Y}_t} \geq \sqrt{\frac{\log_e n}{n}}} \, ,
    \end{align}
    where (a) follows from the definitions of predicted and true messages, (b) holds because $h$ is a bijection, (c) follows from the distance between adjacent elements in $\Phi$, (d) follows from the definitions of $m$ and $m_i$, (e) holds because $\sum_{i=1}^d R_i = \frac{d(p - 1)}{2} - \alpha$, (f) holds due to the upper-bound in Step 1, (g) holds because $n \geq n_1$, and (h) follows from the union bound.
    
    Next, observe that
    \begin{align}
        \bkt{\hat{\mathbf{p}}_Y}_t - \bkt{\mathbf{p}_Y}_t &\stackclap{(a)}{=} \frac{1}{n} \sum_{j=1}^n \Iv{Z_j = t} - \bkt{\mathbf{p}_Y}_t \\
        &\stackclap{(b)}{=} \frac{1}{n} \sum_{j=1}^n \Iv{Z_j = t} - \sum_{c=1}^d \sum_{b=1}^d \frac{\rho_b}{d} \bkt{\mathbf{p}_{Z,c,b}}_t \\
        &\stackclap{(c)}{=} \frac{1}{n} \sum_{j=1}^n \Iv{Z_j = t} \\
        &\quad - \frac{1}{n} \sum_{c=1}^d \sum_{b=1}^d \sum_{j \in \mathcal{J}_{c,b}} \bkt{\mathbf{p}_{Z,c,b}}_t \\
        &= \frac{1}{n} \sum_{j=1}^n \Iv{Z_j = t} - \frac{1}{n} \sum_{j=1}^n \P{Z_j = t} \, ,
    \end{align}
    where (a) holds because $Y_1^n$ is a permutation of $Z_1^n$, (b) follows from substituting in \cref{eq:general-random-permutation-block}, and (c) holds because $\abs{\mathcal{J}_{c,b}} = \rho_b \, \frac{n}{d}$ for all $c, b \in [d]$. The $Z_1^n$ are conditionally independent given the messages, since the letters $\prn{X_i}_1^n$ are independently generated and the $W_1^n$ are independently passed through the DMC. Applying Hoeffding's inequality (\cref{lemma:hoeffdings-inequality}) with $\tau = \sqrt{\frac{\log_e n}{n}}$,\footnote{We cannot apply Hoeffding directly on $\hat{\mathbf{p}}_Y$ because the $Y_1^n$ are \emph{not} independent, by virtue of being the outputs of a random permutation block whose inputs $Z_1^n$ are \emph{not} identically distributed.}
    \begin{align}
        \P{g_n(Y_1^n) \neq (\boldsymbol{\mu}_1, \dots, \boldsymbol{\mu}_d)} &\leq \sum_{t=0}^{d(p-1)} 2e^{-2n \prn{\sqrt{\frac{\log_e n}{n}}}^2} \\
        &= \frac{2(d(p - 1) + 1)}{n^2} \\
        &\stackclap{(a)}{\leq} \epsilon \, , \label{eq:high-probability-bound}
    \end{align}
    where (a) holds because $n \geq n_0 \geq \sqrt{\frac{2(d(p - 1) + 1)}{\epsilon}}$. Finally, taking expectation with respect to the messages yields $\Perror{n} \leq \epsilon$ as desired.
\end{proof}

We remark that our randomized encoders and decoder have polynomial time complexity with respect to $n$. The encoders each run in $\bigO{n}$ time since they take $n$ samples from a $\Categorical{\theta_0, \dots, \theta_{p-1}}$ distribution; each sample can be done in $\bigO{p} = \bigO{1}$ time by sampling $U \sim \Uniform{0}{1}$ to a fixed precision and computing
\begin{align}
    X = \min \brc{x \in \bktz{p - 1}: U < \sum_{k=0}^x \theta_k} \, .
\end{align}
Steps 1 through 5 of the decoder cost $\bigO{\poly(d, p)} = \bigO{1}$ time, step 6 costs $\bigO{n + dp} = \bigO{n}$, steps 7 and 8 cost $\bigO{\poly(d, p)} = \bigO{1}$, step 9 costs
\begin{align}
    \bigO{dpm} = \bigO{\prod_{i=1}^d m_i} = \bigO{n^{\frac{1}{2} - \frac{\alpha}{2d(p - 1)}}} \, ,
\end{align}
and steps 10 through 12 cost $\bigO{\poly(d, p)} = \bigO{1}$. Hence, our randomized coding scheme does not suffer from intractable decoding complexity.

Below, we present some technical lemmas used in the proof of \cref{theorem:general-achievability}. The first lemma provides a formal derivation of the marginal distribution of the output letters:

\begin{lemma}[Marginal Distribution of $Y$] \label{lemma:y-marginal}
    Conditioned on sending the messages $(\boldsymbol{\mu}_1, \dots, \boldsymbol{\mu}_d)$,
    \begin{align}
        p_{Y_j}(y) = \sum_{c=1}^d \sum_{b=1}^d \frac{\rho_b}{d} \, p_{Z_{c,b}}(y)
    \end{align}
    for each $j \in [n]$, where $p_{Z_{c,b}}$ is the marginal distribution of $Z_t$ for $t \in \mathcal{J}_{c,b}$.
\end{lemma}

\begin{proof}
    Let $\Pi \in \mathrm{S}_n$ denote the permutation sampled by the random permutation block. Recall from \cref{subsection:formal-model} that $\mathcal{Y}$ denotes the output alphabet of the adder MAC. Then,
    \begin{align}
        p_{Y_j}(y) &\stackclap{(a)}{=} \sum_{z_1^n \in \mathcal{Y}^n} \sum_{\pi \in \mathrm{S}_n} p_{Z_1^n}(z_1^n) \, p_\Pi(\pi) \, p_{Y_j | Z_1^n, \Pi}(y | z_1^n, \pi) \\
        &\stackclap{(b)}{=} \sum_{z_1^n \in \mathcal{Y}^n} \sum_{\pi \in \mathrm{S}_n} \prn{\prod_{t=1}^n p_{Z_t}(z_t)} p_\Pi(\pi) \, p_{Y_j | Z_1^n, \Pi}(y | z_1^n, \pi) \\
        &\stackclap{(c)}{=} \frac{1}{n!} \sum_{z_1^n \in \mathcal{Y}^n} \sum_{\pi \in \mathrm{S}_n} \prn{\prod_{t=1}^n p_{Z_t}(z_t)} p_{Y_j | Z_1^n, \Pi}(y | z_1^n, \pi) \\
        &\stackclap{(d)}{=} \frac{1}{n!} \sum_{z_1^n \in \mathcal{Y}^n} \sum_{\pi \in \mathrm{S}_n} \prn{\prod_{t=1}^n p_{Z_t}(z_t)} \Iv{y = z_{\pi(j)}} \\
        &= \frac{1}{n!} \sum_{\pi \in \mathrm{S}_n} \sum_{z_1^n \in \mathcal{Y}^n} \Biggl\{ p_{Z_{\pi(j)}} \mprn{z_{\pi(j)}} \Iv{y = z_{\pi(j)}} \\
        &\qquad \cdot \prod_{t \neq \pi(j)} p_{Z_t}(z_t) \Biggr\} \\
        &\stackclap{(e)}{=} \frac{1}{n!} \sum_{\pi \in \mathrm{S}_n} \Biggl\{ \Biggl( \sum_{z_{\pi(j)} \in \mathcal{Y}} p_{Z_{\pi(j)}} \mprn{z_{\pi(j)}} \Iv{y = z_{\pi(j)}} \Biggr) \\
        &\qquad \cdot \prod_{t \neq \pi(j)} \sum_{z_t \in \mathcal{Y}} p_{Z_t}(z_t) \Biggr\} \\
        &\stackclap{(f)}{=} \frac{1}{n!} \sum_{\pi \in \mathrm{S}_n} p_{Z_{\pi(j)}}(y) \\
        &= \frac{1}{n!} \sum_{t=1}^n \sum_{\substack{\pi \in \mathrm{S}_n: \\ \pi(j) = t}} p_{Z_t}(y) \\
        &= \frac{1}{n} \sum_{t=1}^n p_{Z_t}(y) \\
        &= \frac{1}{n} \sum_{c=1}^d \sum_{b=1}^d \sum_{t \in \mathcal{J}_{c,b}} p_{Z_t}(y) \\
        &\stackclap{(g)}{=} \sum_{c=1}^d \sum_{b=1}^d \frac{\rho_b}{d} \, p_{Z_{c,b}}(y) \, ,
    \end{align}
    where (a) follows from the assumptions of the formal model, (b) follows from independence of the $Z_1^n$, (c) follows from uniformity of $\Pi$, (d) follows from the meaning of a permutation block, (e) holds due to the distributive property
    \begin{align}
        \sum_{x_1 \in \mathcal{X}_1} \cdots \sum_{x_n \in \mathcal{X}_n} \prod_{i=1}^n f_i(x_i) = \prod_{i=1}^n \sum_{x_i \in \mathcal{X}_i} f_i(x_i) \, ,
    \end{align}
    (f) holds because probabilities in a distribution sum to $1$, and (g) holds because $\abs{\mathcal{J}_{c,b}} = \rho_b \, \frac{n}{d}$ for all $c, b \in [d]$.
\end{proof}

The second lemma characterizes the entries of $\tilde{\mathbf{C}}^{-1}$, where the matrix $\tilde{\mathbf{C}}$ is defined as in \cref{eq:c-tilde-matrix}:

\begin{lemma} \label{lemma:c-tilde-inverse}
    Let $\tilde{\mathbf{C}}$ be defined as in \cref{eq:c-tilde-matrix}. Then $\tilde{\mathbf{C}}$ is invertible, and all non-zero entries of $\tilde{\mathbf{C}}^{-1}$ have unit magnitude, i.e.,
    \begin{align}
        \tilde{\mathbf{C}}^{-1} \in \brc{-1, 0, 1}^{(dp+1) \times (dp+1)} \, .
    \end{align}
\end{lemma}

\begin{proof}
    By an argument similar to the justification of \cref{eq:full-column-rank}, each column of $\tilde{\mathbf{C}}$ is the leftmost column that contains a non-zero entry in some row of $\tilde{\mathbf{C}}$. Thus, $\tilde{\mathbf{C}}$ has full rank as desired.
    
    Index the rows and columns of $\tilde{\mathbf{C}}$ by $s \in [dp + 1]$ and $t \in [dp + 1]$ respectively, and vice-versa for $\tilde{\mathbf{C}}^{-1}$. Given $s$, let
    \begin{align}
        t_s = \min \brc{t \in [dp + 1]: \bkt{\tilde{\mathbf{C}}}_{s,t} \neq 0}
    \end{align}
    be the leftmost column of $\tilde{\mathbf{C}}$ with a non-zero entry in row $s$. We will show by induction on $t_s$ that for each $s \in [dp + 1]$, all non-zero entries of $\ibkt{\tilde{\mathbf{C}}^{-1}}_s$ have unit magnitude and are located in or above row $t_s$. Fix $s \in [dp + 1]$.
    
    Base case: $t_s = 1$. By definition of $t_s$, it follows that $s = 1$, since $\ibkt{\tilde{\mathbf{C}}}_{s,1} = \bkt{\mathbf{e}_1}_s$ is non-zero only when $s = 1$. We have
    \begin{align}
        \bkt{\tilde{\mathbf{C}}^{-1}}_s = \bkt{\tilde{\mathbf{C}}^{-1}}_1 \stackclap{(a)}{=} \mathbf{e}_1 = \mathbf{e}_{t_s} \, ,
    \end{align}
    where (a) holds because $\tilde{\mathbf{C}} \mathbf{e}_1 = \mathbf{e}_1$. This is what we wanted to show.
    
    Inductive step: $t_s > 1$. By definition of $\tilde{\mathbf{C}}$, column $t_s$ of $\tilde{\mathbf{C}}$ corresponds to some column $(c, k)$ of $\mathbf{C}$. Two properties follow from \cref{eq:c1-matrix-characterization} and \cref{eq:c2-matrix}:
    \begin{enumerate}
        \item Row $(c - 1)(p - 1) + k$ in $\mathbf{C}_1$ and row $c$ in $\mathbf{C}_2$ are the only rows in $\mathbf{C}_1$ and $\mathbf{C}_2$ with a non-zero entry at column $(c, k)$. Thus, there are exactly two rows $s_1$ and $s_2$ in $\tilde{\mathbf{C}}$ with a non-zero entry at column $t_s$.
        \item Iff $k = 0$, column $(c, k)$ is the leftmost column in $\mathbf{C}_2$ with a non-zero entry at row $c$. Iff $k > 0$,\footnote{The one exception to this ``only if'' is $(c, k) = (1, 0)$, but observe that the corresponding column in $\tilde{\mathbf{C}}_1$ is not the leftmost column with a non-zero entry at the top row. Therefore, this ``iff'' holds in the context of $\tilde{\mathbf{C}}$.} column $(c, k)$ is the leftmost column in $\mathbf{C}_1$ with a non-zero entry at row $(c - 1)(p - 1) + k$. Thus, exactly one of $t_s = t_{s_1}$ and $t_s = t_{s_2}$ is true.
    \end{enumerate}
    Assume without loss of generality that $t_s = t_{s_1}$. By the two properties, it follows that $s = s_1$. Then
    \begin{align}
        \bkt{\tilde{\mathbf{C}}^{-1}}_s = \mathbf{e}_{t_s} - \bkt{\tilde{\mathbf{C}}^{-1}}_{s_2} \, , \label{eq:inverse-column-s}
    \end{align}
    because
    \begin{align}
        \tilde{\mathbf{C}} \prn{\mathbf{e}_{t_s} - \bkt{\tilde{\mathbf{C}}^{-1}}_{s_2}} &= \bkt{\tilde{\mathbf{C}}}_{t_s} - \tilde{\mathbf{C}} \bkt{\tilde{\mathbf{C}}^{-1}}_{s_2} \\
        &\stackclap{(a)}{=} \prn{\mathbf{e}_{s_1} + \mathbf{e}_{s_2}} - \mathbf{e}_{s_2} \\
        &\stackclap{(b)}{=} \mathbf{e}_s \, ,
    \end{align}
    where (a) holds by the first property and the fact that $\tilde{\mathbf{C}}$ is a binary matrix, and (b) holds because $s = s_1$. Since $t_{s_2} < t_s$, by the induction hypothesis, all non-zero entries of $\ibkt{\tilde{\mathbf{C}}^{-1}}_{s_2}$ have unit magnitude and are located in or above row $t_{s_2}$. Combined with \cref{eq:inverse-column-s}, this shows all non-zero entries of $\ibkt{\tilde{\mathbf{C}}^{-1}}_s$ have unit magnitude and are located in or above row $t_s$, as desired.
\end{proof}

\section{Converse Proof} \label{section:converse-proof}

In this section, we prove \cref{theorem:general-converse}.

\begin{proof}
    By definition of permutation capacity region, we want to show that every achievable rate $d$-tuple $R = (R_1, \dots, R_d)$ has
    \begin{align}
        \sum_{i=1}^d R_i \leq \frac{d(p - 1)}{2} \, .
    \end{align}
    Fix an achievable $R$. Assume there exists a family of encoders and decoders with rate $d$-tuple $R$ such that $\lim_{n \rightarrow \infty} \Perror{n} = 0$. For any $n \geq 2$,
    \begin{align}
        \sum_{i=1}^d R_i &= \frac{1}{\log_2 n} \sum_{i=1}^d \log_2 n^{R_i} \\
        &\stackclap{(a)}{=} \frac{1}{\log_2 n} \sum_{i=1}^d \H{M_i} \\
        &\stackclap{(b)}{=} \frac{1}{\log_2 n} \H{M} \\
        &= \frac{1}{\log_2 n} \underbrace{\H{M \given \hat{M}}}_{\circled{1}} + \frac{1}{\log_2 n} \underbrace{\I{M; \hat{M}}}_{\circled{2}} \, ,
    \end{align}
    where (a) follows because the messages are uniformly distributed and (b) follows because the messages are independent.
    
    We now follow the argument in \cite[Section III-C]{Makur2020b}. We upper-bound $\circled{1}$ using Fano's inequality:
    \begin{align}
        \circled{1} &\stackclap{(a)}{\leq} 1 + \P{\hat{M} \neq M} \H{M \given \hat{M} \neq M} \\
        &= 1 + \Perror{n} \sum_{i=1}^d \H{M_i \given M_1, \dots, M_{i-1}, \hat{M} \neq M} \\
        &\stackclap{(b)}{\leq} 1 + \Perror{n} \sum_{i=1}^d \log_2 n^{R_i} \\
        &= 1 + \Perror{n} \log_2(n) \sum_{i=1}^d R_i \, ,
    \end{align}
    where (a) follows from Fano's inequality and (b) follows from the standard upper bound on Shannon entropy.

    Next, we upper-bound $\circled{2}$. Consider the Markov chain
    \begin{align}
        M \rightarrow \prn{X_1^d}_1^n \rightarrow W_1^n \rightarrow Z_1^n \rightarrow Y_1^n \rightarrow \hat{M} \, .
    \end{align}
    Observe that for every $y_1^n \in \bktz{d(p - 1)}^n$ and $m \in \mathcal{M}$,
    \begin{align}
        p_{Y_1^n | M}(y_1^n | m) = \frac{\prod_{t=0}^{d(p-1)} (n \bkt{\hat{\mathbf{p}}_Y}_t)!}{n!} \, \P{\hat{\mathbf{p}}_Z = \hat{\mathbf{p}}_Y \given M = m} \, .
    \end{align}
    Since $p_{Y_1^n | M}(y_1^n | m)$ depends on $y_1^n$ through $\hat{\mathbf{p}}_Y$, by the Fisher-Neyman factorization theorem, $\hat{\mathbf{p}}_Y$ is a sufficient statistic of $Y_1^n$. Thus,
    \begin{align}
        \circled{2} &\stackclap{(a)}{\leq} \I{M; Y_1^n} \\
        &\stackclap{(b)}{=} \I{M; \hat{\mathbf{p}}_Y} \\ 
        &\stackclap{(c)}{\leq} \I{W_1^n; \hat{\mathbf{p}}_Y} \, ,
    \end{align}
    where (a) and (c) follow from the data processing inequality and (b) follows from sufficiency.

    Our proof uses the following result which we distill from the literature and restate for convenience.

    \begin{lemma}[Permutation Channel Mutual Information {\cite[Eq. 56]{Makur2020b}}] \label{lemma:perturbation-channel-mutual-information}
        Consider the model
        \begin{equation}
            \begin{tikzcd}
                X_1^n \in \bktz{q}^n \arrow{r}{\text{channel}} & Z_1^n \in \bktz{q}^n \arrow{r}{\text{permute}} & Y_1^n \in \bktz{q}^n
            \end{tikzcd} ,
        \end{equation}
        where $X_1^n$ is a codeword, $Z_1^n$ is the result of passing $X_1^n$ letter-wise through a strictly positive $(q+1) \times (q+1)$ DMC, and $Y_1^n$ is a uniformly random permutation of $Z_1^n$. Then, there exist positive constants $\alpha$, $\beta$, and $\gamma$ such that for all sufficiently large $n$,
        \begin{align}
            &\equad \I{X_1^n; \hat{\mathbf{p}}_Y} \\
            &\leq q \prn{\log_2(n + 1) - \frac{1}{2} \log_2(\alpha n) + \frac{\beta}{n} \log_2(\gamma n)} \, .
        \end{align}
    \end{lemma}

    Combining the results above and taking the limit as $n \rightarrow \infty$,
    \begin{align}
        \sum_{i=1}^d R_i &\leq \lim_{n \rightarrow \infty} \Biggl\{ \frac{1}{\log_2 n} + \Perror{n} \sum_{i=1}^d R_i + d(p - 1) \\
        &\quad \cdot \prn{\frac{\log_2(n + 1)}{\log_2 n} - \frac{\log_2(\alpha n)}{2 \log_2 n} + \frac{\beta \log_2(\gamma n)}{n \log_2 n}} \Biggr\} \\
        &= d(p - 1) \prn{1 - \frac{1}{2}} \\
        &= \frac{d(p - 1)}{2}
    \end{align}
    as desired.
\end{proof}

\section{Conclusion}

In this paper, we formulated the PAMAC network model as a natural abstraction of many-to-one communication over a permutation channel. Motivated primarily by theoretical interest in mathematical techniques with overarching relevance to information theory, along with the aforementioned applications to multipath routed networks and wireless communications, we undertook a comprehensive study of the PAMAC's information-theoretic properties and ultimately derived an exact characterization of its permutation capacity region. Our work underscores the fundamental role played by time sharing in establishing achievability results for multiple-access channels, and sheds light on nascent connections between mixed-radix numerical systems and coding schemes for time sharing. Our achievability proofs reaffirm the suitability of encoding messages in the permutation channel setting as samples from a Bernoulli or categorical distribution by defining a correspondence between distribution parameters and messages. Secondarily, we presented a contrasting achievability result for the binary PAMAC, underpinned by the observation that the additive structure of the PAMAC encodes the relevant Bernoulli parameters within the roots of the probability generating function of the adder's output distribution. Leveraging properties of Frobenius companion matrices, we framed our analysis through the lens of spectral stability and notably exploited the Bauer-Fike theorem from matrix perturbation theory to obtain explicit bounds on decoding performance.

We propose three directions for future work. Firstly, our analysis in this paper treats the number of senders $d$ as a constant that is independent of the codeword length $n$. A natural continuation of our line of research may tighten the bounds in our achievability proofs to improve their asymptotic dependence on $d$. Secondly, future work may adapt our results to a variant of the PAMAC with the adder and DMC swapped. As this alternative model entails passing each sender's codeword through a separate DMC, qualitatively distinct subcases may arise depending on whether the DMCs share the same transition probabilities. Lastly, a promising follow-up goal is to extend our time sharing proofs to general MACs, wherein the senders' letters are combined by a general function $\eta: \mathcal{X}^d \rightarrow \mathcal{Y}$ instead of an adder to produce the letters $W_j = \eta(X_{1,j}, \dots, X_{d,j})$.

Overall, our main contributions and proposed future directions highlight the continuing importance of multiple-access permutation channels as a captivating object of theoretical interest, which nonetheless enjoys relevance to a diverse range of downstream applications.

\appendices
\crefalias{section}{appendix}   

\section{Proof of Binary Achievability Using Time Sharing} \label{appendix:binary-achievability-time-sharing}

In this appendix, we prove \cref{theorem:binary-achievability-time-sharing}. Our argument makes use of two auxiliary results (\cref{lemma:y-marginal,lemma:minimum-singular-value}), which we prove at the end of \cref{section:general-achievability} and in \cref{appendix:auxiliary-results}, respectively.

\begin{proof}[Proof of \cref{theorem:binary-achievability-time-sharing}]
    By definition of permutation capacity region, it suffices to show that for all $\alpha > 0$, any rate $d$-tuple satisfying
    \begin{align}
        \sum_{i=1}^d R_i = \frac{d}{2} - \alpha \enspace \text{and} \enspace \forall i \in [d], \, R_i > 0
    \end{align}
    is achievable. Fix $\alpha > 0$ and $(R_1, \dots, R_d) \in \mathbb{R}_+^d$ satisfying the above. By definition of achievable rate tuples, we want to show that
    \begin{align}
        \forall \epsilon > 0, \, \exists n_0 \in \mathbb{N}, \, \forall n \geq n_0, \, \Perror{n} \leq \epsilon \, .
    \end{align}
    Fix $\epsilon > 0$. Choose
    \begin{align}
        n_0 = \max \brc{2^{\frac{d}{\min_{i \in [d]} R_i}}, n_1, \sqrt{\frac{2(d + 1)}{\epsilon}}} \, ,
    \end{align}
    where $n_1 \in \mathbb{N}$ is sufficiently large such that
    \begin{align}
        \forall n \geq n_1, \, \frac{\sigma_{\min}^2 \mprn{\mathbf{P}_{Z|W}}}{\sqrt{2} \prn{d + 2}^{\frac{9}{2}}} \, n^{-\frac{1}{2} + \frac{\alpha}{d}} \geq \sqrt{\frac{\log_e n}{n}} \, .
    \end{align}
    Such an $n_1$ exists because
    \begin{align}
        \lim_{n \rightarrow \infty} \frac{n^{-\frac{1}{2} + \frac{\alpha}{d}}}{\sqrt{\frac{\log_e n}{n}}} = \lim_{n \rightarrow \infty} \frac{n^{\frac{\alpha}{d}}}{\sqrt{\log_e n}} = \infty \, .
    \end{align}
    Fix $n \geq n_0$. For notational simplicity, let $m_i = n^{\frac{R_i}{d}}$ for each $i \in [d]$, and assume without loss of generality that each $m_i \in \mathbb{N}$. Consider the following message sets, encoders, and decoder.
    
    \textbf{Message sets.} Without loss of generality, let $\mathcal{M}_i$ be a $d$-dimensional lattice of evenly spaced points in $[0, 1]^d$. (Contrary to the situation in the proof of \cref{theorem:binary-achievability-root-stability}, there is no need to specifically incorporate padding around the boundary of the hypercube.) Formally,
    \begin{align}
        \forall i \in [d], \enspace \mathcal{M}_i &= \Theta_i^d \, , \\
        \forall i \in [d - 1], \enspace \Theta_i &= \brc{\frac{\ell_i}{m_i - 1}: \ell_i \in \bktz{m_i - 1}} \, , \\
        \Theta_d &= \brc{\frac{\ell_d}{m_d}: \ell_d \in \bktz{m_d - 1}} \, .
    \end{align}
    The denominators $m_i - 1$ are positive, and thus the message sets are well-defined, because
    \begin{align}
        m_i = n^{\frac{R_i}{d}} \geq n_0^{\frac{R_i}{d}} \geq \prn{2^{\frac{d}{\min_{i \in [d]} R_i}}}^{\frac{R_i}{d}} \geq 2
    \end{align}
    for each $i \in [d]$.
    
    We have $\abs{\mathcal{M}_i} = \abs{\Theta_i}^d = m_i^d = n^{R_i}$ for each $i \in [d]$, which satisfies the definition of rate $d$-tuple. We represent a message $\boldsymbol{\mu}_i \in \mathcal{M}_i$ sent by sender $i$ as a $d$-tuple of variables $\boldsymbol{\mu}_i = (\theta_{i,1}, \dots, \theta_{i,d})$.
    
    \textbf{Encoders.} Each encoder $f_{i,n}: \mathcal{M}_i \rightarrow \brc{0, 1}^n$ is randomized and outputs a binary codeword of length $n$. We adopt the segmentation scheme described in \cref{section:general-achievability}. In each segment $c$, sender $i$ is active in subsegment $i$ and passive in all other subsegments. During its active phase, sender $i$'s encoder generates letters by sampling i.i.d. from $\Bernoulli{\theta_{i,c}}$. During its passive phase, sender $i$ outputs all ones if it is one of the first $c - 1$ passive senders, and all zeros otherwise. Formally, for $j \in [n]$, let $c(j)$ and $b(j)$ be the indices of the segment and subsegment containing the $j$th letter, respectively. Then,
    \begin{align}
        &f_{i,n}(\boldsymbol{\mu}_i) = \prn{x_i}_1^n \\
        \text{where} \enspace &\forall j \in [n], \, \begin{cases}
            x_{i,j} \sim \Bernoulli{\theta_{i,c(j)}} \, , & \text{if $i = b(j)$} \, , \\
            x_{i,j} = \Iv{i \leq c(j) - 1} \, , & \text{if $i < b(j)$} \, , \\
            x_{i,j} = \Iv{i \leq c(j)} \, , & \text{if $i > b(j)$} \, .
        \end{cases}
    \end{align}
    This coding scheme ensures the alphabets of $W_j$ from adjacent segments overlap at only one value, namely value $c$ for segments $c$ and $c + 1$.
    
    \textbf{Decoder.} We adopt the definitions of $m$, $\Phi$, and $h$ from the time sharing discussion in \cref{section:general-achievability}. Upon receiving the output codeword $y_1^n$, the decoder $g_n: \bktz{d}^n \rightarrow \mathcal{M}$ executes the following:
    \begin{enumerate}
        \item Form the matrix $\mathbf{C} \in \mathbb{R}^{(d+1) \times d}$, given by
        \begin{align}
            \bkt{\mathbf{C}}_{s,t} = \begin{cases}
                -1 \, , & \text{if $s = t$} \, , \\
                1 \, , & \text{if $s = t + 1$} \, , \\
                0 \, , & \text{otherwise} \, .
            \end{cases}
        \end{align}
        \item Compute the matrix $\mathbf{A} = \mathbf{P}_{Z|W}^\top \mathbf{C} \in \mathbb{R}^{(d+1) \times d}$, which is equivalent to
        \begin{align}
            \forall t \in [d], \enspace \bkt{\mathbf{A}}_t = \bkt{\mathbf{P}_{Z|W}}_{\ang{t}} - \bkt{\mathbf{P}_{Z|W}}_{\ang{t-1}} \, ,
        \end{align}
        where we \emph{one}-index the matrix $\mathbf{A}$. (Recall that we \emph{zero}-index the channel matrix $\mathbf{P}_{Z|W}$.)
        \item Compute the empirical distribution $\hat{\mathbf{p}}_Y \in \mathcal{S}_d$, given by
        \begin{align}
            \forall t \in \bktz{d}, \enspace \bkt{\hat{\mathbf{p}}_Y}_t = \frac{1}{n} \sum_{j=1}^n \Iv{y_j = t} \, .
        \end{align}
        \item Form the vector $\hat{\mathbf{b}} \in \mathbb{R}^{d+1}$, given by
        \begin{align}
            \hat{\mathbf{b}} = d \hat{\mathbf{p}}_Y^\top - \sum_{t=0}^{d-1} \bkt{\mathbf{P}_{Z|W}}_{\ang{t}} \, .
        \end{align}
        \item Using the normal equations, compute the least-squares solution (cf. \cite{Makur2020b})
        \begin{align}
            \tilde{\boldsymbol{\phi}} &= \brc{\tilde{\phi}_c}_{c=1}^d \\
            &= \arg \min_{\boldsymbol{\phi} \in \mathbb{R}^d} \norm{\mathbf{A} \boldsymbol{\phi} - \hat{\mathbf{b}}}_2 \\
            &= \underbrace{\prn{\mathbf{A}^\top \mathbf{A}}^{-1} \mathbf{A}^\top}_{\mathbf{A}^\dagger} \hat{\mathbf{b}} \, ,
        \end{align}
        where $\mathbf{A}^\dagger$ is the \emph{Moore-Penrose pseudoinverse} of $\mathbf{A}$.
        \item Compute $\hat{\boldsymbol{\phi}} = \brc{\hat{\phi}_c}_{c=1}^d \in \Phi^d$ by rounding each entry of $\tilde{\boldsymbol{\phi}}$ to the nearest element in $\Phi$ (cf. \cite{Makur2020b}):
        \begin{align}
            \forall c \in [d], \enspace \hat{\phi}_c = \arg \min_{\phi \in \Phi} \abs{\tilde{\phi}_c - \phi} \, .
        \end{align}
        \item Convert each $\hat{\phi}_c$ into a $d$-tuple in $\Theta_1 \times \cdots \times \Theta_d$ using the inverse mapping $h^{-1}$:
        \begin{align}
            \forall c \in [d], \enspace \brc{\hat{\theta}_{i,c}}_{i=1}^d = h^{-1} \mprn{\hat{\phi}_c} \, .
        \end{align}
        \item Form the predicted messages
        \begin{align}
            \forall i \in [d], \enspace \hat{\boldsymbol{\mu}}_i = \brc{\hat{\theta}_{i,c}}_{c=1}^d \, .
        \end{align}
        \item Return the predicted messages $\prn{\hat{\boldsymbol{\mu}}_1, \dots, \hat{\boldsymbol{\mu}}_d}$.
    \end{enumerate}
    
    \textbf{Proof of decoder well-definedness.} In this part, we verify that $\mathbf{A}^\top \mathbf{A}$ is indeed invertible, and so the decoder is well-defined. By inspection, $\mathbf{C}^\top$ is in row echelon form with a pivot in each row, and so $\mathbf{C}$ has full column rank. Since $\mathbf{A}$ is the product of two matrices with full column rank, $\mathbf{A}$ has full column rank. Since a Gramian matrix has the same rank as its vector realization,
    \begin{align}
        \rank{\mathbf{A}^\top \mathbf{A}} = \rank{\mathbf{A}} = d \, ,
    \end{align}
    and thus $\mathbf{A}^\top \mathbf{A}$ is invertible as desired.
    
    \textbf{Proof of achievability.} In this proof, all probabilities $\P{\cdot}$ are conditioned on sending the true messages $(\boldsymbol{\mu}_1, \dots, \boldsymbol{\mu}_d)$.
    
    \emph{Step 1: Upper-bounding error in least-squares solution.} Let $\mathbf{p}_{W,c,b}$ and $\mathbf{p}_{Z,c,b}$ be the true distributions of $W_j$ and $Z_j$, respectively, given $j \in \mathcal{J}_{c,b}$. The dynamics of the formal model are governed by the equations:
    \begin{align}
        \forall c, b \in [d], \enspace \mathbf{p}_{W,c,b} &= (1 - \theta_{b,c}) \mathbf{e}_{c-1}^\top + \theta_{b,c} \mathbf{e}_c^\top \, , \label{eq:binary-encoders-and-adder} \\
        \forall c, b \in [d], \enspace \mathbf{p}_{Z,c,b} &= \mathbf{p}_{W,c,b} \mathbf{P}_{Z|W} \, , \label{eq:binary-dmc} \\
        \mathbf{p}_Y &= \sum_{c=1}^d \sum_{b=1}^d \frac{\rho_b}{d} \mathbf{p}_{Z,c,b} \, , \label{eq:binary-random-permutation-block}
    \end{align}
    where \cref{eq:binary-encoders-and-adder} models the encoders and adder, \cref{eq:binary-dmc} models the DMC, and \cref{eq:binary-random-permutation-block} models the random permutation block by \cref{lemma:y-marginal}.
    
    Define $\phi_c = h(\theta_{1,c}, \dots, \theta_{d,c})$ for each $c \in [d]$. Define a vector $\mathbf{b} \in \mathbb{R}^{d+1}$ as
    \begin{align}
        \mathbf{b} = d \mathbf{p}_Y^\top - \sum_{t=0}^{d-1} \bkt{\mathbf{P}_{Z|W}}_{\ang{t}} \, .
    \end{align}
    Combining the above,
    \begin{align}
        \mathbf{p}_Y^\top &\stackclap{(a)}{=} \sum_{c=1}^d \sum_{b=1}^d \frac{\rho_b}{d} \Bigl( (1 - \theta_{b,c}) \bkt{\mathbf{P}_{Z|W}}_{\ang{c-1}} \\
        &\quad + \theta_{b,c} \bkt{\mathbf{P}_{Z|W}}_{\ang{c}} \Bigr) \\
        &= \sum_{c=1}^d \sum_{b=1}^d \frac{\rho_b}{d} \bkt{\mathbf{P}_{Z|W}}_{\ang{c-1}} \\
        &\quad + \sum_{c=1}^d \sum_{b=1}^d \frac{\rho_b}{d} \prn{\bkt{\mathbf{P}_{Z|W}}_{\ang{c}} - \bkt{\mathbf{P}_{Z|W}}_{\ang{c-1}}} \theta_{b,c} \\
        &\stackclap{(b)}{=} \frac{1}{d} \sum_{c=1}^d \bkt{\mathbf{P}_{Z|W}}_{\ang{c-1}} \\
        &\quad + \frac{1}{d} \sum_{c=1}^d \prn{\bkt{\mathbf{P}_{Z|W}}_{\ang{c}} - \bkt{\mathbf{P}_{Z|W}}_{\ang{c-1}}} \phi_c \\
        &= \frac{1}{d} \sum_{t=0}^{d-1} \bkt{\mathbf{P}_{Z|W}}_{\ang{t}} \\
        &\quad + \frac{1}{d} \sum_{c=1}^d \prn{\bkt{\mathbf{P}_{Z|W}}_{\ang{c}} - \bkt{\mathbf{P}_{Z|W}}_{\ang{c-1}}} \phi_c \, ,
    \end{align}
    where (a) follows from substituting $\mathbf{p}_{W,c,b}$ and $\mathbf{p}_{Z,c,b}$ into $\mathbf{p}_Y$, and (b) holds by \cref{eq:sum-of-subsegment-proportions} and \cref{eq:weighted-sum-of-subsegment-thetas}. Rearranging,
    \begin{align}
        &\equad \sum_{c=1}^d \prn{\bkt{\mathbf{P}_{Z|W}}_{\ang{c}} - \bkt{\mathbf{P}_{Z|W}}_{\ang{c-1}}} \phi_c \\
        &= d \mathbf{p}_Y^\top - \sum_{t=0}^{d-1} \bkt{\mathbf{P}_{Z|W}}_{\ang{t}} \, .
    \end{align}
    Letting $\boldsymbol{\phi}^* = (\phi_1, \dots, \phi_d)$, this equation can be vectorized as $\mathbf{A} \boldsymbol{\phi}^* = \mathbf{b}$. Therefore $\boldsymbol{\phi}^* = \prn{\mathbf{A}^\top \mathbf{A}}^{-1} \mathbf{A}^\top \mathbf{b}$, and
    \begin{align}
        \norm{\tilde{\boldsymbol{\phi}} - \boldsymbol{\phi}^*}_\infty &\leq \norm{\prn{\mathbf{A}^\top \mathbf{A}}^{-1} \mathbf{A}^\top}_\infty \norm{\hat{\mathbf{b}} - \mathbf{b}}_\infty \\
        &\stackclap{(a)}{\leq} \sqrt{d + 1} \norm{\prn{\mathbf{A}^\top \mathbf{A}}^{-1} \mathbf{A}^\top}_2 \norm{\hat{\mathbf{b}} - \mathbf{b}}_\infty \\
        &\stackclap{(b)}{\leq} \sqrt{d + 1} \underbrace{\norm{\prn{\mathbf{A}^\top \mathbf{A}}^{-1}}_2}_{\circled{1}} \underbrace{\norm{\mathbf{A}^\top}_2}_{\circled{2}} \underbrace{\norm{\hat{\mathbf{b}} - \mathbf{b}}_\infty}_{\circled{3}} \, ,
    \end{align}
    where (a) holds by the equivalence of $\ell^\infty$ and $\ell^2$ norms and (b) holds by the submultiplicativity of the $\ell^2$ matrix norm.
    
    Next, we upper-bound $\circled{1}$:
    \begin{align}
        \circled{1} &= \frac{1}{\sigma_{\min} \mprn{\mathbf{A}^\top \mathbf{A}}} = \frac{1}{\sigma_{\min}^2(\mathbf{A})} \leq \frac{1}{\sigma_{\min}^2 \mprn{\mathbf{P}_{Z|W}} \, \sigma_{\min}^2(\mathbf{C})} \, .
    \end{align}
    Define a matrix $\tilde{\mathbf{C}}$ by prepending one column to $\mathbf{C}$:
    \begin{align}
        \tilde{\mathbf{C}} = \bkt{\begin{array}{@{}c@{}|@{}c@{}}
            \begin{matrix}
                1 \\
                0 \\
                \vdots \\
                0
            \end{matrix} & \makebox[2em]{$\mathbf{C}$}
        \end{array}} \in \mathbb{R}^{(d+1) \times (d+1)} \, .
    \end{align}
    This matrix is invertible, because it is upper triangular with no zeros on its diagonal. By inspection, its inverse is
    \begin{align}
        \bkt{\tilde{\mathbf{C}}^{-1}}_{s,t} = \begin{cases}
            1 \, , & \text{if $s \leq t$} \, , \\
            0 \, , & \text{otherwise} \, .
        \end{cases}
    \end{align}
    It follows that
    \begin{align}
        \frac{1}{\sigma_{\min}(\mathbf{C})} &\stackclap{(a)}{\leq} \frac{1}{\sigma_{\min} \mprn{\tilde{\mathbf{C}}}} \\
        &= \norm{\tilde{\mathbf{C}}^{-1}}_2 \\
        &\leq \frnorm{\tilde{\mathbf{C}}^{-1}} \\
        &\stackclap{(b)}{=} \sqrt{\frac{(d + 1)(d + 2)}{2}} \, ,
    \end{align}
    where (a) holds because adding a column to a tall matrix does not increase its minimum singular value (\cref{lemma:minimum-singular-value}), and (b) follows from the value of $\tilde{\mathbf{C}}^{-1}$. Combining the above,
    \begin{align}
        \circled{1} \leq \frac{(d + 1)(d + 2)}{2 \sigma_{\min}^2 \mprn{\mathbf{P}_{Z|W}}} \, .
    \end{align}
    
    Next, we upper-bound $\circled{2}$:
    \begin{align}
        \circled{2} &\stackclap{(a)}{\leq} \norm{\mathbf{P}_{Z|W}}_2 \frnorm{\mathbf{C}} \\
        &\stackclap{(b)}{=} \norm{\mathbf{P}_{Z|W}}_2 \sqrt{2d} \\
        &\stackclap{(c)}{\leq} \sqrt{d + 1} \norm{\mathbf{P}_{Z|W}}_\infty \sqrt{2d} \\
        &\stackclap{(d)}{=} \sqrt{2d(d + 1)} \\
        &\leq \sqrt{2} \prn{d + 1} \, ,
    \end{align}
    where (a) holds by the submultiplicativity of the $\ell^2$ matrix norm, (b) holds because each column of $\mathbf{C}$ contains two entries with unit magnitude and all other entries $0$, (c) holds by the equivalence of $\ell^2$ and $\ell^\infty$ norms, and (d) holds because each row of $\mathbf{P}_{Z|W}$ sums to $1$ and so the maximum row sum matrix norm \cite[Example 5.6.5]{HornJohnson2013} of $\mathbf{P}_{Z|W}$ is $1$.
    
    By definition of $\mathbf{b}$ and $\hat{\mathbf{b}}$, we have $\circled{3} = d \norm{\hat{\mathbf{p}}_Y - \mathbf{p}_Y}_\infty$. Combining the bounds on $\circled{1}$ to $\circled{3}$,
    \begin{align}
        \norm{\tilde{\boldsymbol{\phi}} - \boldsymbol{\phi}^*}_\infty \leq \frac{(d + 2)^{\frac{9}{2}}}{\sqrt{2} \, \sigma_{\min}^2 \mprn{\mathbf{P}_{Z|W}}} \norm{\hat{\mathbf{p}}_Y - \mathbf{p}_Y}_\infty \, .
    \end{align}
    
    \emph{Step 2: Concentration bound for empirical distribution of $Y$.} We have
    \begin{align}
        &\equad \P{g_n(Y_1^n) \neq (\boldsymbol{\mu}_1, \dots, \boldsymbol{\mu}_d)} \\
        &= \P{\exists i \in [d], \, \hat{\boldsymbol{\mu}}_i \neq \boldsymbol{\mu}_i} \\
        &\stackclap{(a)}{=} \P{\exists i \in [d], \, \exists c \in [d], \, \hat{\theta}_{i,c} \neq \theta_{i,c}} \\
        &\stackclap{(b)}{=} \P{\exists c \in [d], \, \hat{\phi}_c \neq \phi_c} \\
        &\stackclap{(c)}{\leq} \P{\norm{\tilde{\boldsymbol{\phi}} - \boldsymbol{\phi}^*}_\infty \geq \frac{1}{2m}} \\
        &\stackclap{(d)}{=} \P{\norm{\tilde{\boldsymbol{\phi}} - \boldsymbol{\phi}^*}_\infty \geq \frac{1}{2n^{\frac{1}{d} \sum_{i=1}^d R_i}}} \\
        &\stackclap{(e)}{=} \P{\norm{\tilde{\boldsymbol{\phi}} - \boldsymbol{\phi}^*}_\infty \geq \frac{1}{2} n^{-\frac{1}{2} + \frac{\alpha}{d}}} \\
        &\stackclap{(f)}{\leq} \P{\frac{(d + 2)^{\frac{9}{2}}}{\sqrt{2} \, \sigma_{\min}^2 \mprn{\mathbf{P}_{Z|W}}} \norm{\hat{\mathbf{p}}_Y - \mathbf{p}_Y}_\infty \geq \frac{1}{2} n^{-\frac{1}{2} + \frac{\alpha}{d}}} \\
        &\stackclap{(g)}{\leq} \P{\norm{\hat{\mathbf{p}}_Y - \mathbf{p}_Y}_\infty \geq \sqrt{\frac{\log_e n}{n}}} \\
        &\stackclap{(h)}{\leq} \sum_{t=0}^d \P{\abs{\bkt{\hat{\mathbf{p}}_Y}_t - \bkt{\mathbf{p}_Y}_t} \geq \sqrt{\frac{\log_e n}{n}}} \, ,
    \end{align}
    where (a) follows from the definitions of predicted and true messages, (b) holds because $h$ is a bijection, (c) follows from the distance between adjacent elements in $\Phi$, (d) follows from the definitions of $m$ and $m_i$, (e) holds because $\sum_{i=1}^d R_i = \frac{d}{2} - \alpha$, (f) holds due to the upper-bound in Step 1, (g) holds because $n \geq n_1$, and (h) follows from the union bound.
    
    Next, observe that
    \begin{align}
        \bkt{\hat{\mathbf{p}}_Y}_t - \bkt{\mathbf{p}_Y}_t &\stackclap{(a)}{=} \frac{1}{n} \sum_{j=1}^n \Iv{Z_j = t} - \bkt{\mathbf{p}_Y}_t \\
        &\stackclap{(b)}{=} \frac{1}{n} \sum_{j=1}^n \Iv{Z_j = t} - \sum_{c=1}^d \sum_{b=1}^d \frac{\rho_b}{d} \bkt{\mathbf{p}_{Z,c,b}}_t \\
        &\stackclap{(c)}{=} \frac{1}{n} \sum_{j=1}^n \Iv{Z_j = t} \\
        &\quad - \frac{1}{n} \sum_{c=1}^d \sum_{b=1}^d \sum_{j \in \mathcal{J}_{c,b}} \bkt{\mathbf{p}_{Z,c,b}}_t \\
        &= \frac{1}{n} \sum_{j=1}^n \Iv{Z_j = t} - \frac{1}{n} \sum_{j=1}^n \P{Z_j = t} \, ,
    \end{align}
    where (a) holds because $Y_1^n$ is a permutation of $Z_1^n$, (b) follows from substituting in \cref{eq:binary-random-permutation-block}, and (c) holds because $\abs{\mathcal{J}_{c,b}} = \rho_b \, \frac{n}{d}$ for all $c, b \in [d]$. The $Z_1^n$ are conditionally independent given the messages, since the letters $\prn{X_i}_1^n$ are independently generated and the $W_1^n$ are independently passed through the DMC. Applying Hoeffding's inequality (\cref{lemma:hoeffdings-inequality}) with $\tau = \sqrt{\frac{\log_e n}{n}}$,\footnote{We cannot apply Hoeffding directly on $\hat{\mathbf{p}}_Y$ because the $Y_1^n$ are \emph{not} independent, by virtue of being the outputs of a random permutation block whose inputs $Z_1^n$ are \emph{not} identically distributed.}
    \begin{align}
        \P{g_n(Y_1^n) \neq (\boldsymbol{\mu}_1, \dots, \boldsymbol{\mu}_d)} &\leq \sum_{t=0}^d 2e^{-2n \prn{\sqrt{\frac{\log_e n}{n}}}^2} \\
        &= \frac{2(d + 1)}{n^2} \\
        &\stackclap{(a)}{\leq} \epsilon \, ,
    \end{align}
    where (a) holds because $n \geq n_0 \geq \sqrt{\frac{2(d + 1)}{\epsilon}}$. Finally, taking expectation with respect to the messages yields $\Perror{n} \leq \epsilon$ as desired.
\end{proof}

\section{Auxiliary Results} \label{appendix:auxiliary-results}

The following lemma characterizes the size of the message component sets $\brc{\mathcal{L}_i}_{i=1}^d$ defined in \cref{eq:message-sets} using a prefix sum argument:

\begin{lemma}[Cardinality of Message Sets] \label{lemma:cardinality}
    Let $\brc{\mathcal{L}_i}_{i=1}^d$ be defined as in \cref{eq:message-sets}. Then,
    \begin{align}
        \abs{\mathcal{L}_i} = \begin{cases}
            \multibinom{m_i}{p - 1} \, , & \text{if $i < d$} \, , \\
            \multibinom{m_i + 1}{p - 1} - p \, , & \text{if $i = d$} \, .
        \end{cases}
    \end{align}
\end{lemma}

\begin{proof}
    \emph{Case 1: $i < d$.} Let $\mathcal{T}_i$ be the set of non-decreasing $(p - 1)$-tuples of integers in $\bktz{m_i - 1}$:
    \begin{align}
        \mathcal{T}_i = \brc{\brc{t_k}_{k=1}^{p-1} \in \bktz{m_i - 1}^{p-1}: t_1 \leq \cdots \leq t_{p-1}} \, .
    \end{align}
    There exists a bijection between $\mathcal{L}_i$ and $\mathcal{T}_i$, given by the following mapping in which $\brc{t_k}_{k=1}^{p-1} \in \mathcal{T}_i$ are the scaled \emph{prefix sums} of $\brc{\theta_k}_{k=0}^{p-1} \in \mathcal{L}_i$, as visualized in \cref{figure:theta-t-bijection}:
    \begin{align}
        \forall k \in \bktz{p - 1}, \enspace \theta_k &\stackclap{(a)}{=} \frac{t_{k+1} - t_k}{m_i - 1} \, , \\
        \forall k \in [p - 1], \enspace t_k &= (m_i - 1) \sum_{k'=0}^{k-1} \theta_{k'} \, .
    \end{align}
    To avoid creating special cases in (a) for $\theta_0$ and $\theta_{p-1}$, we assume that $t_0 = 0$ and $t_p = m_i - 1$. The lack of a true $t_p$ variable reflects the loss of one degree of freedom incurred by the sum-to-one constraint $\sum_{k=0}^{p-1} \theta_k = 1$ in the definition of $\mathcal{L}_i$. The weak inequalities between consecutive $t_k$ variables correspond to the possibility of some $\theta_k$ variables being zero.
    
    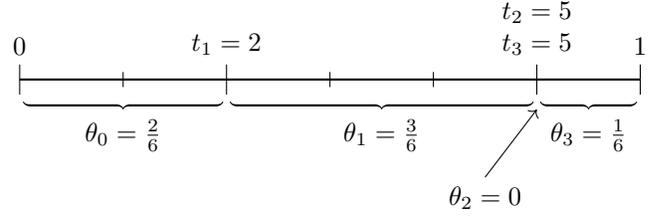
\begin{figure}
        \centering
        \begin{tikzpicture}
            \def\l{8.25}     
            \def\R{0.2}      
            \def\r{0.1}      
            \def\pad{0.05}   
            
            \draw[thick] (0,1) -- (\l,1);
            
            \draw[thin] (0, 1 - \R) -- (0, 1 + \R) node [above] {$0$};
            \draw[thin] (2/6 * \l, 1 - \R) -- (2/6 * \l, 1 + \R) node [above] {$t_1 = 2$};
            \draw[thin] (5/6 * \l, 1 - \R) -- (5/6 * \l, 1 + \R) node [above, align=center] {$t_2 = 5$ \\ $t_3 = 5$};
            \draw[thin] (\l, 1 - \R) -- (\l, 1 + \R) node [above] {$1$};
            
            \draw[thin] (1/6 * \l, 1 - \r) -- (1/6 * \l, 1 + \r);
            \draw[thin] (3/6 * \l, 1 - \r) -- (3/6 * \l, 1 + \r);
            \draw[thin] (4/6 * \l, 1 - \r) -- (4/6 * \l, 1 + \r);
            
            \draw[thick, decorate, decoration={calligraphic brace, mirror}] (\pad, 1 - \R - 0.1) -- (2/6 * \l - \pad, 1 - \R - 0.1)
            node [pos=0.5, anchor=north, yshift=-0.1cm] {$\theta_0 = \frac{2}{6}$};
            \draw[thick, decorate, decoration={calligraphic brace, mirror}] (2/6 * \l + \pad, 1 - \R - 0.1) -- (5/6 * \l - \pad, 1 - \R - 0.1)
            node [pos=0.5, anchor=north, yshift=-0.1cm] {$\theta_1 = \frac{3}{6}$};
            \draw[thick, decorate, decoration={calligraphic brace, mirror}] (5/6 * \l + \pad, 1 - \R - 0.1) -- (\l - \pad, 1 - \R - 0.1)
            node [pos=0.5, anchor=north, yshift=-0.1cm] {$\theta_3 = \frac{1}{6}$};
            
            \draw[->] (9/12 * \l, -0.3) node [below] {$\theta_2 = 0$} -- (5/6 * \l, 1 - \R - 0.2);
        \end{tikzpicture}
        \caption{Visualization of the relationship between $\brc{\theta_k}_{k=0}^{p-1}$ and $\brc{t_k}_{k=1}^{p-1}$ in the case $i < d$, $p = 4$, and $m_i = 7$.}
        \label{figure:theta-t-bijection}
    \end{figure}
    
    By virtue of their non-decreasing quality, each tuple $\brc{t_k}_{k=1}^{p-1} \in \mathcal{T}_i$ can be interpreted as a $(p - 1)$-multisubset of $\bktz{m_i - 1}$. Therefore,
    \begin{align}
        \abs{\mathcal{L}_i} \stackclap{(a)}{=} \abs{\mathcal{T}_i} = \multibinom{m_i}{p - 1}
    \end{align}
    as desired, where (a) holds due to bijectivity.
    
    \emph{Case 2: $i = d$.} Let $\mathcal{T}_d$ be the set of non-decreasing $(p - 1)$-tuples of integers in $\bktz{m_d}$, with the property $(*)$ that any two consecutive entries have difference less than $m_d$:
    \begin{align}
        \mathcal{T}_d = \Biggl \{\brc{t_k}_{k=1}^{p-1} \in \bktz{m_d}^{p-1}: \bigwedge_{k=1}^p 0 \leq \underbrace{t_k - t_{k-1} < m_d}_{(*)} \Biggr \} \, .
    \end{align}
    To avoid creating special cases, we assume that $t_0 = 0$ and $t_p = m_d$. The bijective \emph{prefix sum} mapping between $\mathcal{L}_d$ and $\mathcal{T}_d$ is
    \begin{align}
        \forall k \in \bktz{p - 1}, \enspace \theta_k &= \frac{t_{k+1} - t_k}{m_d} \, , \\
        \forall k \in [p - 1], \enspace t_k &= m_d \sum_{k'=0}^{k-1} \theta_{k'} \, .
    \end{align}
    Property $(*)$ reflects the fact that $\theta_k < 1$ for each $k \in \bktz{p - 1}$. Interpreting each tuple $\brc{t_k}_{k=1}^{p-1} \in \mathcal{T}_d$ as a $(p - 1)$-multisubset of $\bktz{m_d}$, we have
    \begin{align}
        \abs{\mathcal{L}_d} \stackclap{(a)}{=} \abs{\mathcal{T}_d} \stackclap{(b)}{=} \multibinom{m_d + 1}{p - 1} - p
    \end{align}
    as desired, where (a) holds due to bijectivity and (b) holds because there are $p$ multisubsets in $\bktz{m_d}^{p-1}$ which do not satisfy $(*)$.\footnote{A multisubset does not satisfy $(*)$ iff it only contains values $0$ and $m_d$; the $p$ such multisubsets contain different numbers of zeros and $m_d$ values.}
\end{proof}

The following lemma confirms that the subsegment proportions defined in \cref{eq:subsegment-proportions} sum to $1$:

\begin{lemma}[Subsegment Well-Definedness] \label{lemma:subsegment-well-definedness}
    Let $\brc{\rho_b}_{b=1}^d$ be defined as in \cref{eq:subsegment-proportions}. Then, $\sum_{b=1}^d \rho_b = 1$.
\end{lemma}

\begin{proof}
    For $t \in [d]$, define
    \begin{align}
        A_t = \sum_{b=1}^{d-t} \frac{m_b - 1}{\prod_{i=1}^b m_i} + \frac{m_{d-t+1}}{\prod_{i=1}^{d-t+1} m_i} \, .
    \end{align}
    For all $t \in [d - 1]$, observe that $A_t$ satisfies the recurrence
    \begin{align}
        A_t &= \sum_{b=1}^{d-t} \frac{m_b - 1}{\prod_{i=1}^b m_i} + \frac{1}{\prod_{i=1}^{d-t} m_i} \\
        &= \sum_{b=1}^{d-(t+1)} \frac{m_b - 1}{\prod_{i=1}^b m_i} + \frac{m_{d-t} - 1}{\prod_{i=1}^{d-t} m_i} + \frac{1}{\prod_{i=1}^{d-t} m_i} \\
        &= \sum_{b=1}^{d-(t+1)} \frac{m_b - 1}{\prod_{i=1}^b m_i} + \frac{m_{d-t}}{\prod_{i=1}^{d-t} m_i} = A_{t+1} \, .
    \end{align}
    Therefore, $\sum_{b=1}^d \rho_b = A_1 = A_d = 1$ as desired.
\end{proof}

The following lemma provides a justification from first principles of the fact that adding a column to a tall matrix does not increase its minimum singular value:

\begin{lemma}[Minimum Singular Value] \label{lemma:minimum-singular-value}
    Let $\mathbf{A} \in \mathbb{R}^{m \times n}$ be a tall matrix (namely, $m > n$) and let $\mathbf{B} \in \mathbb{R}^{m \times (n+1)}$ be formed by adding a column to $\mathbf{A}$. Then, $\sigma_{\min}(\mathbf{B}) \leq \sigma_{\min}(\mathbf{A})$.
\end{lemma}

\begin{proof}
    Let $S^n = \brc{\mathbf{x} \in \mathbb{R}^n: \norm{\mathbf{x}} = 1}$ be the unit $n$-sphere. We will first show that for any tall or square matrix $\mathbf{C} \in \mathbb{R}^{m \times n}$,
    \begin{align}
        \sigma_{\min}(\mathbf{C}) = \min_{\mathbf{x} \in S^n} \norm{\mathbf{C} \mathbf{x}} \, .
    \end{align}
    Fix $\mathbf{C} \in \mathbb{R}^{m \times n}$ and consider its singular value decomposition $\mathbf{C} = \mathbf{U} \boldsymbol{\Sigma} \mathbf{V}^\top$, where $\mathbf{U} \in \mathbb{R}^{m \times m}$ is orthogonal, $\boldsymbol{\Sigma} \in \mathbb{R}^{m \times n}$ is diagonal with monotonically non-increasing diagonal entries, and $\mathbf{V}^\top \in \mathbb{R}^{n \times n}$ is orthogonal. It follows that
    \begin{align}
        \sigma_{\min}(\mathbf{C}) &= \min_{\mathbf{x} \in S^n} \sqrt{\sum_{j=1}^n \sigma_{\min}^2(\mathbf{C}) \bkt{\mathbf{x}}_j^2} \\
        &\stackclap{(a)}{\leq} \min_{\mathbf{x} \in S^n} \sqrt{\sum_{j=1}^n \sigma_j^2(\mathbf{C}) \bkt{\mathbf{x}}_j^2} \\
        &\stackclap{(b)}{=} \min_{\mathbf{x} \in S^n} \norm{\boldsymbol{\Sigma} \mathbf{x}} \\
        &\stackclap{(c)}{=} \min_{\mathbf{x} \in S^n} \norm{\mathbf{U} \boldsymbol{\Sigma} \mathbf{V}^\top \mathbf{x}} \\
        &= \min_{\mathbf{x} \in S^n} \norm{\mathbf{C} \mathbf{x}} \, ,
    \end{align}
    where (b) holds because $\brc{\sigma_j(\mathbf{C})}_{j=1}^n$ are the diagonal entries of $\boldsymbol{\Sigma}$ and (c) holds because $\mathbf{U}$ and $\mathbf{V}^\top$ are orthogonal matrices. Furthermore, (a) holds with equality because
    \begin{align}
        \min_{\mathbf{x} \in S^n} \sqrt{\sum_{j=1}^n \sigma_j^2(\mathbf{C}) \bkt{\mathbf{x}}_j^2} \leq \sqrt{\sum_{j=1}^n \sigma_j^2(\mathbf{C}) \bkt{\mathbf{e}_n}_j^2} = \sigma_{\min}(\mathbf{C}) \, .
    \end{align}
    Next, let $\mathbf{x}^* \in \arg \min_{\mathbf{x} \in S^n} \norm{\mathbf{A} \mathbf{x}}$. Form $\mathbf{x}^\dagger \in S^{n+1}$ by adding a $0$ to $\mathbf{x}^*$ at the index where an extra column was added to $\mathbf{A}$ to form $\mathbf{B}$. Since $\mathbf{A}$ and $\mathbf{B}$ are tall or square matrices,
    \begin{align}
        \sigma_{\min}(\mathbf{B}) = \min_{\mathbf{x} \in S^{n+1}} \norm{\mathbf{B} \mathbf{x}} \leq \norm{\mathbf{B} \mathbf{x}^\dagger} = \norm{\mathbf{A} \mathbf{x}^*} = \sigma_{\min}(\mathbf{A})
    \end{align}
    as desired.
\end{proof}

We note that this lemma can also be seen as a direct corollary of the \emph{Cauchy interlacing theorem} \cite[Theorem 4.3.17]{HornJohnson2013}, because the singular values of $\mathbf{A}$ and $\mathbf{B}$ are the square roots of the eigenvalues of $\mathbf{A}^\top \mathbf{A}$ and $\mathbf{B}^\top \mathbf{B}$, respectively, and $\mathbf{A}^\top \mathbf{A}$ is a principal submatrix of $\mathbf{B}^\top \mathbf{B}$.

Lastly, we restate Hoeffding's inequality below for convenience:

\begin{lemma}[Hoeffding's Inequality {\cite[Theorem 2.8]{BoucheronLugosiMassart2013}}] \label{lemma:hoeffdings-inequality}
    Let $\brc{X_i}_{i=1}^n$ be independent random variables where $X_i \in [a_i, b_i]$ with $a_i < b_i$ for each $i \in [n]$. Then, for any $\tau > 0$,
    \begin{align}
        \P{\abs{\frac{1}{n} \sum_{i=1}^n X_i - \frac{1}{n} \sum_{i=1}^n \E{X_i}} \geq \tau} \leq 2e^{-\frac{2n^2 \tau^2}{\sum_{i=1}^n (b_i - a_i)^2}} \, .
    \end{align}
\end{lemma}

\section*{Acknowledgment}

The authors would like to thank Japneet Singh for discussions regarding this work.

\bibliographystyle{myIEEEtran}
\bibliography{journal.bib}

\end{document}